\documentclass[11pt,psamsfonts]{amsart}
\usepackage[foot]{amsaddr}
\usepackage{amsmath,amsthm,amsfonts,amssymb}
\usepackage{eucal}
\usepackage{graphicx}
\usepackage{caption}
\usepackage{indentfirst}
\usepackage{bbm}
\usepackage{grffile}
\usepackage[all,knot]{xy}
\usepackage{chngcntr}
\usepackage{floatrow}
\usepackage{subfig}
\usepackage[colorlinks, citecolor=blue, linkcolor=blue]{hyperref}
\usepackage[usenames, dvipsnames]{color}
\usepackage{enumitem,kantlipsum}
\usepackage{amsbsy}
\usepackage{multicol}
\usepackage{multirow}
\usepackage{mathrsfs}
\usepackage[margin=1.3in]{geometry}
\counterwithin{figure}{section}
\xyoption{arc}

\newcommand{\A}{\mathcal{A}}
\newcommand{\M}{\mathcal{M}}

\newcommand{\CC}{\mathcal{C}}
\newcommand{\ZZ}{\mathcal{Z}}

\newcommand{\Rep}{{\rm Rep}}

\DeclareMathOperator{\FPdim}{FPdim}

\DeclareMathOperator{\End}{End}  
\DeclareMathOperator{\Hom}{Hom} 
 
\DeclareMathOperator{\Span}{Span}

\DeclareMathOperator{\Fun}{Fun}
\DeclareMathOperator{\Obj}{Obj}
\DeclareMathOperator{\Tr}{Tr}
\newcommand{\one}{\mathbf{1}}
\newcommand{\C}{\mathbb C}
\newcommand{\mC}{\mathcal{C}}

\newcommand{\mZ}{\mathcal{Z}}
\newcommand{\mQ}{\mathcal{Q}}
\newcommand{\mD}{\mathcal{D}}

\newcommand{\msC}{\mathscr{C}}
\newcommand{\msD}{\mathscr{D}}

\newcommand{\mfD}{\mathfrak{D}}
\newcommand{\Z}{\mathbb Z}

\newcommand{\B}{\mathcal{B}}

\newcommand{\comments}[1]{}

\renewcommand{\vec}[1]{{\mathbf #1}}
\newcommand{\ket}[1]{|#1\rangle}

\renewcommand{\one}{\mathbf{1}}

\renewcommand{\CC}{\mathcal{C}}

\renewcommand{\Z}{\mathbb{Z}}

\newcommand{\overbar}[1]{\mkern 2.3mu\overline{\mkern-2.3mu#1\mkern-2.3mu}\mkern 2.3mu}

\numberwithin{equation}{section}

\newtheorem{theorem}{Theorem}[section]

\newtheorem{corollary}[theorem]{Corollary}

\newtheorem{prop}[theorem]{Proposition}
\theoremstyle{definition}

\newtheorem{remark}[theorem]{Remark}

\newtheorem{definition}[theorem]{Definition}

\begin{document}

\title[Hamiltonian and Algebraic Theories of Gapped Boundaries]{Hamiltonian and Algebraic Theories of Gapped Boundaries in Topological Phases of Matter}
\author{Iris Cong$^{1,4}$}
\email{$^1$irisycong@engineering.ucla.edu}
\address{$^1$Dept of Computer Science, University of California\\
Los Angeles, CA 90095\\
U.S.A.}

\author{Meng Cheng$^{2,4}$}
\email{$^2$m.cheng@yale.edu}
\address{$^2$Dept of Physics\\
	Yale University\\
	New Haven, CT 06520-8120\\
    U.S.A.}

\author{Zhenghan Wang$^{3,4}$}
\email{$^4$zhenghwa@microsoft.com}
\address{$^3$Dept of Mathematics\\
    University of California\\
    Santa Barbara, CA 93106-6105\\
    U.S.A.}
\address{$^4$Microsoft Station Q\\
    University of California\\
    Santa Barbara, CA 93106-6105\\
    U.S.A.}

\begin{abstract}
We present an exactly solvable lattice Hamiltonian to realize gapped boundaries of Kitaev's quantum double models for Dijkgraaf-Witten theories. We classify the elementary excitations on the boundary, and systematically describe the bulk-to-boundary condensation procedure.  We also present the parallel  algebraic/categorical structure of gapped boundaries.
\end{abstract}

\maketitle


\section{Introduction}
\label{sec:intro}

\subsection{Motivations and Main Results}
\label{sec:motivations}
Topological phases of matter without any symmetry are gapped quantum phases of matter at zero temperature which exhibit topological order \cite{WenTO}.
 Recent studies of topological phases of matter revealed that certain topological phases of matter also support gapped boundaries \cite{Bravyi98}. In this paper, we develop an exactly solvable lattice Hamiltonian theory for gapped boundaries in the Dijkgraaf-Witten topological quantum field theories (TQFTs) by modifying Kitaev's quantum double model \cite{Kitaev97}.  Our study of gapped boundaries is through an interplay between their Hamiltonian realization and an algebraic model using category theory. Using a very simple picture of triangles and a ribbon ring around a hole as in Fig. \ref{fig:main-results}, we provide an insightful interpretation of this interplay in the context of many physical processes. As a few significant examples, the red triangles represent tensor functors to move quasi-particles in the bulk of a topological phase, and the green triangle is a tensor functor (specifically, a quotient functor) which governs the bulk-to-boundary condensation of anyons.

A topological phase of matter $\mathcal{H}=\{H\}$ is an equivalence class of gapped Hamiltonians $H$ which realizes a TQFT at low energy.  Elementary excitations in a topological phase of matter $\mathcal{H}$ are point-like anyons.  Anyons can be modeled algebraically as simple objects in a unitary modular tensor category (UMTC) $\mathcal{B}$, which will be referred to as the topological order of the topological phase $\mathcal{H}$.  A salient application of anyons arises in topological quantum computation (TQC) \cite{Free98, Kitaev97, FKLW, Nayak08, W10}. TQC is one elegant proposal to resolve the susceptibility of qubits to local decoherence---one of the major obstacles to current developments of large-scale quantum computers---by encoding quantum information in topological degrees of freedom in many-body quantum systems, instead of relying upon local properties. In many-anyon systems, this information is encoded in the topological ground state degeneracy, which can arise either from non-trivial topology of the space, or from non-abelian anyons even without topology (anyons with quantum dimension $>1$). For abelian anyons in the plane, the ground state manifold is non-degenerate, and hence not useful to TQC.  However, when the topological phase of matter $\mathcal{H}$ supports gapped boundaries, new topological degeneracies can arise, even for abelian anyons in the plane. (Here, a {\it gapped boundary} can be regarded as an equivalence class of gapped extensions of $H \in \mathcal{H}$ to the boundary, when the underlying manifold is enhanced with boundaries.) This is because a gapped boundary is essentially a coherent superposition of anyons, and hence behaves like a non-abelian anyon. As a result, gapped boundaries can have many important applications to TQC, including a promising path towards the first universal topological quantum computer based on a purely abelian TQFT. This path is outlined in Ref. \cite{Cong16a}, and details are to appear in \cite{Cong16c}.

We consider only topological phases of matter $\mathcal{H}$ that can be represented by fixed-point gapped Hamiltonians $H$ of the form $H=-\sum_{i}H_i$ such that all local terms $H_i$ are commuting Hermitian projectors.  Two general classes of such Hamiltonians are the Kitaev quantum double model for Dijkgraaf-Witten TQFTs and the Levin-Wen model for Turaev-Viro-Barrett-Westbury TQFTs \cite{Kitaev97,Levin04}.  Their input data, finite groups $G$ and unitary fusion categories $\mathcal{C}$ respectively,  are dual to each other. When the Kitaev model is extended from finite groups to connected $*$-quantum groupoids \cite{Chang14}, the two models are equivalent because they both realize the same topological orders---the representation categories $\mathfrak{D}(G)$ of quantum doubles $D(G)$ or Drinfeld centers $\mathcal{Z}(\mathcal{C})$ of the input categories $\mathcal{C}$.  For such topological phases of matter, a gapped boundary is an equivalence class of extensions of the ideal gapped Hamiltonian from a closed surface to a local commuting Hamiltonian on the surface with a boundary.  We classify gapped boundaries by the maximal collection of bulk anyons that can be condensed to vacuum on the boundary. While our theory works for any surface with boundaries, we will mainly focus on a planar region $\Lambda$ with many holes $\mathfrak{h_i}$, which are small rectangles removed from $\Lambda$ (see Fig. \ref{fig:boundary} for an example).  We generally imagine the holes $\mathfrak{h_i}$ as small disks or rectangles, but in order to achieve topological protection of the ground state degeneracy, their sizes cannot be too small relative to the coherence length of the condensate. When gapped boundaries become too small, the condensates of anyons at the boundaries break up and decohere into single anyons.

In the UMTC model of a 2D topological order $\mathcal{B}=\mZ(\mathcal{C})$ that is a Drinfeld center, a stable gapped boundary or gapped hole is modeled by a Lagrangian algebra $\mathcal{A}$ in $\mathcal{B}$.\footnote{We will use the terms gapped boundary, gapped hole and hole interchangeably.}  The Lagrangian algebra $\mathcal{A}$ consists of a collection of bulk bosonic anyons that can be condensed to vacuum at the boundary, and the corresponding gapped boundary is a condensate of those anyons which behaves as a non-abelian anyon of quantum dimension $d_{\mathcal{A}}$.  Lagrangian algebras in $\mathcal{B}=\mZ(\mathcal{C})$ are in one-to-one correspondence with indecomposable module categories $\mathcal{M}$  over $\mathcal{C}$, which can also be used to label gapped boundaries.  In the Dijkgraaf-Witten theory for a finite group $G$, the different irreducible hole types are parameterized by pairs $(K,\omega)$, $K\subseteq G$ a subgroup up to conjugation, $\omega \in H^2(K,\C^\times)$, which can be directly used to construct indecomposable module categories over $\textrm{Vec}_G$.

\begin{figure}
\centering
\includegraphics[width = 0.65\textwidth]{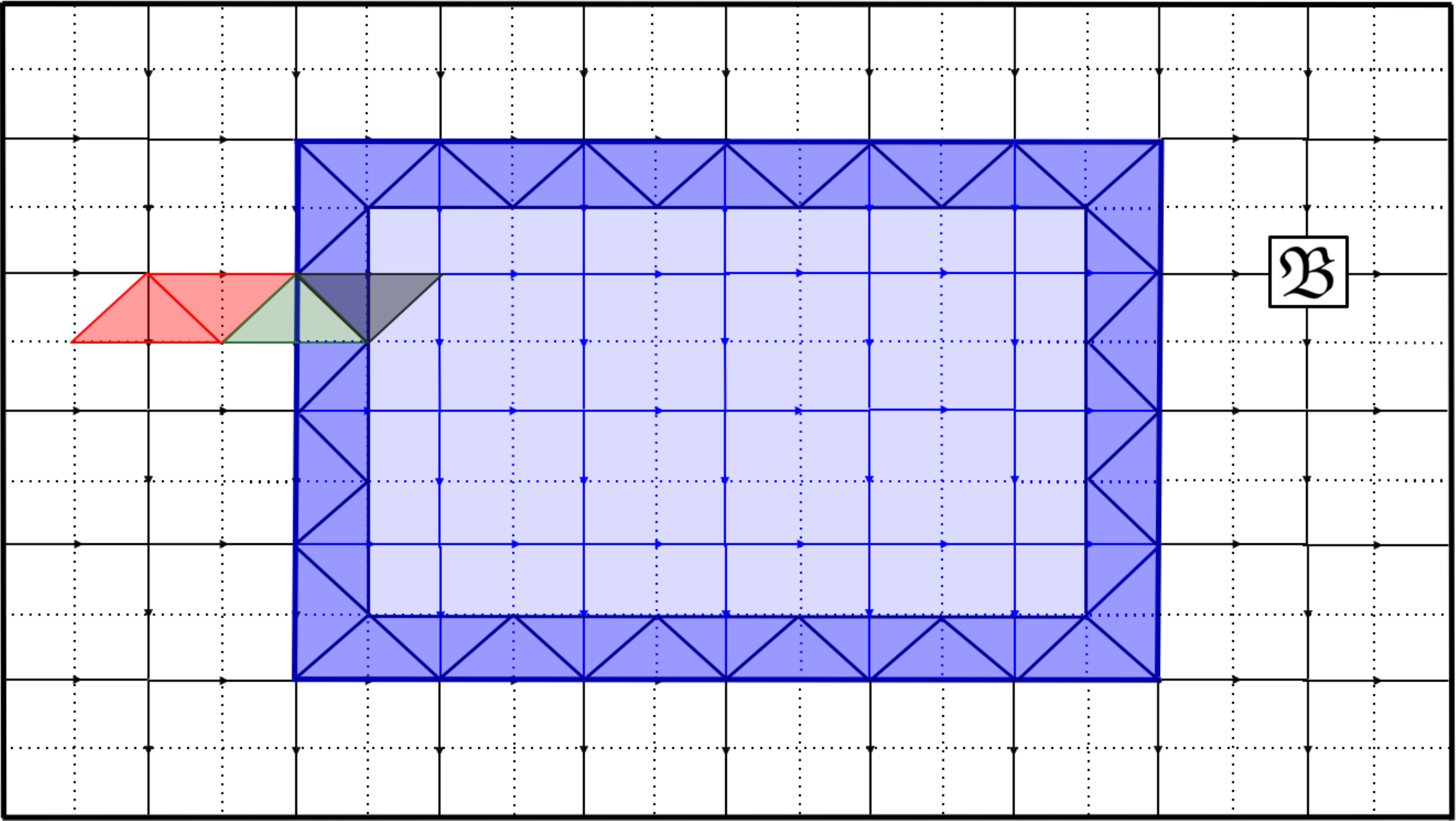}
\caption{Pictorial summary of the Hamiltonian realization and algebraic model of gapped boundaries. In this picture, the hole $\mathfrak{h}$ is the inner boldfaced and blue rectangle.} 
\label{fig:main-results}
\end{figure}

To generate a hole (with a gapped boundary), we modify Kitaev's Hamiltonian using the two-parameter Hamiltonians presented by Bombin and Martin-Delgado in \cite{Bombin08}.  Our resulting Hamiltonian is different from the one presented by Beigi et al. in Ref. \cite{Beigi11} and is explicitly constructed using local projectors on all qudits/spins in the lattice Hilbert space (cf. the one outlined by Bombin and Martin-Delgado in Ref. \cite{Bombin08}).  The main technical tool to analyze the Hamiltonian is the ribbon operators.  We extend ribbon operators to the boundary by defining the {\it boundary ribbon} (e.g. the darker blue ribbon ring in Fig. \ref{fig:main-results}). Furthermore, we show how local and ribbon operator algebras on the boundary form quasi-Hopf and coquasi-Hopf algebras, respectively.

Motivated by the formula in \cite{Bombin08} for bulk excitations, we develop a Fourier transform to express the irreducible types of boundary excitations. This parametrizes the boundary elementary excitation types as

\begin{equation}
\{(T,R): \text{ } T = K r_T K \in K\backslash G/K, \text{ } R \in (K^{r_T})_{\text{ir}}\},
\end{equation}

\noindent
where $K^{r_T}=K\cap r_TKr_T^{-1}$ is a stabilizer group, and $(G)_{\text{ir}}$ denotes the set of irreducible representations of a group $G$. The quantum dimensions of these excitations are given by

\begin{equation}
\textrm{FPdim}(T,R)=\frac{|K|\textrm{dim}(R)}{|K^{r_T}|}.
\end{equation}

\noindent
We provide a simple and systematic method to determine how a bulk anyon $(C,\pi)$ condenses to the boundary into $(T,R)$'s and vice versa. Based on the interplay between local and topological properties of ribbon operators, we define the notion of a {\it condensation channel} and show how multiplicities arise in this condensation procedure. These details are presented in Sections \ref{sec:bd-hamiltonian}-\ref{sec:bd-excitations}.

In the categorical formalism, the bulk of a TQFT is given by a modular tensor category $\B = \mZ(\mC)$ for some unitary fusion category $\mC$, and a (gapped) hole is a Lagrangian algebra $\mathcal{A}=\oplus_{a}n_a a$ in $\B$. In the case of Dijkgraaf-Witten theories, we have $\mC = \text{Vec}_G$. We prove an equivalent formulation for the separability condition of a  Lagrangian algebra, which gives a straightforward way to enumerate gapped boundaries of any such TQFT.  For most purposes, $\A$ can be regarded as a (composite) non-abelian anyon of quantum dimension $d_{\mathcal{A}}$. Gapped boundaries are in one-to-one correspondence with indecomposable module categories $\mathcal{M}_i$ over $\mC$. Then, elementary excitations on $\mathcal{M}_i$ are the simple objects in the functor fusion category $\mC_{ii} = \textrm{Fun}_{\mC}(\mathcal{M}_i, \mathcal{M}_i)$.  In this formalism, the condensation functor is a tensor functor (specifically, a quotient functor) from $\mathcal{Z}(\mC)$ to $\mC_{ii}$. The collections of fusion categories $\mC_{ii}$ and their bimodule categories $\mC_{ij}$ form a multi-fusion category $\mathfrak{C}$. From this multi-fusion category, we can find quantum dimensions of boundary excitations. In addition, we provide analogs of the $\theta-3j$ and $F-6j$ symbols of modular tensor categories in the context of bulk-to-boundary condensation, which we call the $M-3j$ and $M-6j$ symbols.

Topological degeneracies in the presence of holes $\mathfrak{h_i}$ labeled by $\mathcal{A}_i$ are described by the morphism space $\textrm{Hom}(\mathcal{A_\infty}, \otimes_i\mathcal{A}_i)$, where the outermost boundary is labeled by $\mathcal{A_\infty}$. $\mathcal{A_\infty}$ can be either an anyon type or a boundary type.

\subsection{Previous Works}

The first example of gapped boundaries appeared in \cite{Bravyi98} as the smooth and rough boundaries of the $\Z_2$ toric code. Boundaries of Kitaev's quantum double model for general finite groups $G$ were studied by Beigi et al. in \cite{Beigi11}. In that work, they generated gapped boundaries with a different Hamiltonian and described condensations to vacuum, but left the description of the boundary excitations as an important open problem.  In 2009, Kitaev contemplated the categorical formulation that a gapped boundary is modeled by a condensable Frobenius algebra \cite{Kitaev09}. Later, a related categorical description with some details is outlined in \cite{KitaevKong, Fuchs2014, Kapustin10b}.  Further clarifications appeared in \cite{Kong} on boundary excitations, but no explicit general solvable Hamiltonian is presented.  The mathematics of such a theory is in \cite{KO,Davydov12}. 

A physical theory of gapped boundaries related to defects for abelian topological phases of matter is developed in \cite{ Levin13, Bark13a,Bark13b,Bark13c, Kapustin14}. For recent works on a physical understanding of gapped boundaries in more general topological phases and in the closely related topic of anyon condensation, see \cite{Kapustin10, Bais09, Kong13, Eliens13, LWW, Kong15, Neupert16a, Neupert16b, Wan16}. Topological degeneracy has been studied using various techniques in \cite{Kapustin14, Bark13b, LWW,WW,HW}.

\subsection{Notations}

The notations we adopt throughout the paper are presented in Appendix \ref{sec:notations}.  

Throughout the paper, all algebras and tensor categories are over the complex numbers $\C$.  All fusion and modular tensor categories are unitary.  Unitary fusion categories are spherical.

In this paper, gapped boundaries will always be oriented so that the bulk of the topological phase is on the left hand side when traversing the boundary; this allows us to consider gapped boundaries as indecomposable left module categories. See Sections \ref{sec:hamiltonian} and \ref{sec:algebraic} for details.

When appropriate, the font $D(G)$ will be used to denote the quasi-triangular Hopf algebra which is the quantum double of a finite group $G$. The font $\mfD(G)$ will be used for the modular tensor category which is the Drinfeld center of $\text{Vec}_G$, the category of $G$-graded vector spaces, or $\Rep(G)$, the representation category of $G$. In general, we have $\mfD(G) = \mZ(\Rep(G)) = \mZ(\text{Vec}_G) = \Rep(D(G))$, and these notations may be used interchangeably.

\subsection{Acknowledgment}
The authors thank Maissam Barkeshli, Shawn Cui, and Cesar Galindo for answering many questions.  We thank Alexei Davydov for pointing out the example that two different Lagrangian algebras can have the same underlying object, and we thank Bowen Shi for pointing out an error in the original version. I.C. would like to thank Michael Freedman and Microsoft Station Q for hospitality in hosting the summer internship and visits during which this work was done. M.C. thanks Chao-Ming Jian for collaborations on related topics. Z.W. is partially supported by NSF grants DMS-1108736 and DMS-1411212.

\vspace{2mm}
\section{Hamiltonian realization of gapped boundaries}
\label{sec:hamiltonian}

In this section, we present a Hamiltonian realization of gapped boundaries in any Kitaev quantum double model for the (untwisted) Dijkgraaf-Witten theory based on a finite group $G$. Sections \ref{sec:kitaev-hamiltonian} and \ref{sec:ribbon-operators} review the Hamiltonian for the standard Kitaev model and its corresponding ribbon operators. In Section \ref{sec:bd-hamiltonian}, we review existing works on Hamiltonians for special cases of Kitaev models with boundaries, and develop our own, more general Hamiltonian, which may also be adapted for TQC. Section \ref{sec:hamiltonian-gsd-condensation} presents the ground state degeneracy for this model, and Section \ref{sec:bd-excitations} classifies the elementary excitations on the boundary and systematically describes the bulk-to-boundary condensation procedure. Finally, in Section \ref{sec:ds3-hamiltonian-example}, we provide a concrete example for all of the theory by considering the non-abelian theory $\mfD(S_3)$.

\subsection{Kitaev quantum double models}
\label{sec:kitaev-hamiltonian}

Kitaev's famous toric code paper \cite{Kitaev97} presents a Hamiltonian to realize the Dijkgraaf-Witten theory based on any finite group $G$ on a general directed lattice and perform topological quantum computation. For simplicity of illustration and calculation, we will assume throughout our paper that the lattice is the square lattice in the plane; however, it is clear that all of the developed theory here extends to arbitrary lattices. In this model, a data qudit is placed on each edge of the lattice, as shown in Fig. \ref{fig:kitaev}. The Hilbert space for each qudit has an orthonormal basis given by $\{\ket{g}: g \in G\}$, so the total Hilbert space is $\mathcal{L}=\otimes_{e}\mathbb{C}[G]$.

\begin{figure}
\centering
\includegraphics[width = 0.4\textwidth]{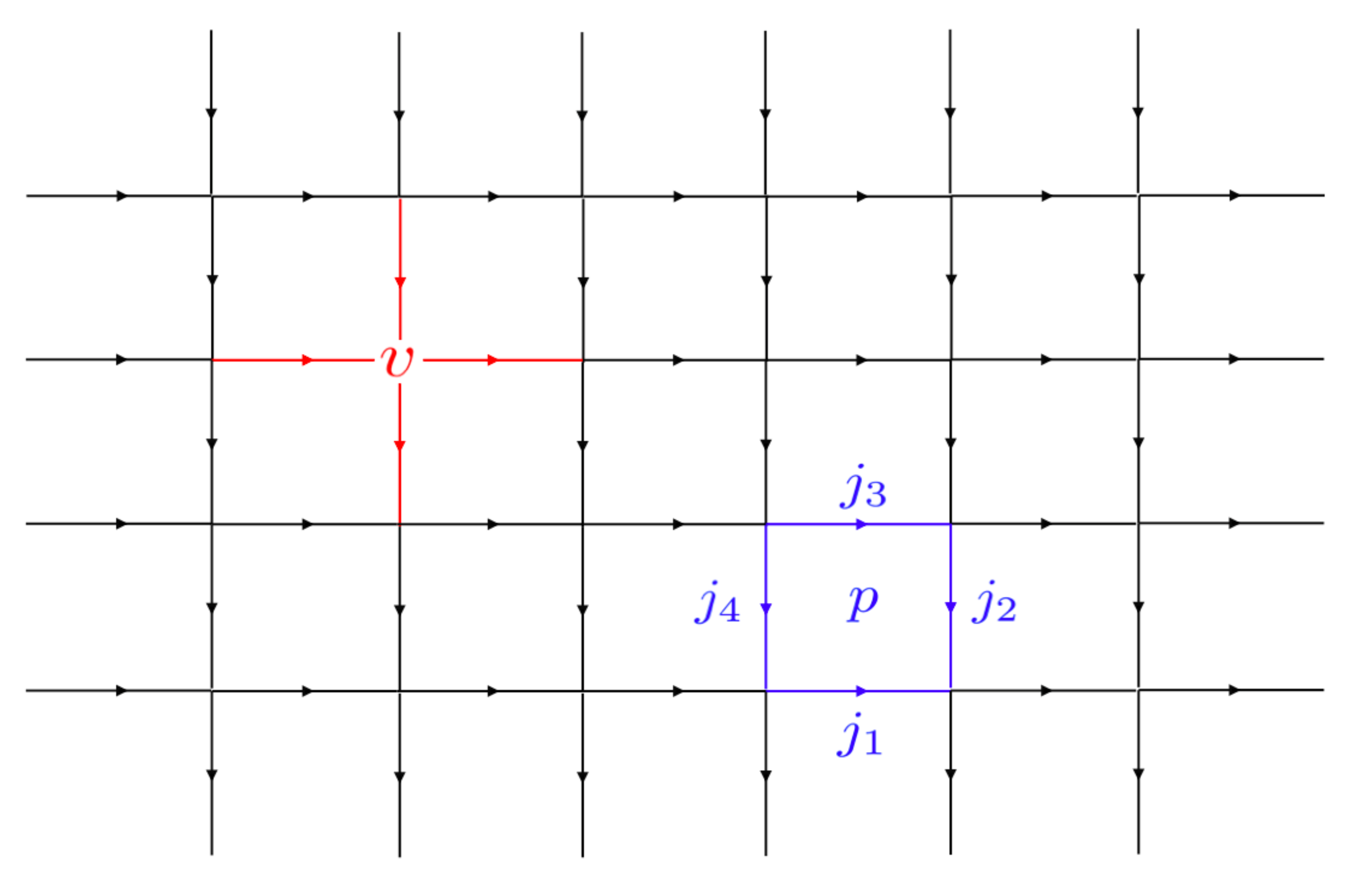}
\caption{Lattice for the Kitaev model. For simplicity of illustration and calculation, we use a square lattice, but in general, one can use an arbitrary lattice. If the group $G$ is non-abelian, it is necessary to define orientations on edges, as we have shown here. The edges $j$ and $j_1,...j_m$, used to obtain $A^g(v)$ and $B^h(p)$, are illustrated for this example of $v,p$.}
\label{fig:kitaev}
\end{figure}

As discussed in Ref. \cite{Kitaev97}, a Hamiltonian is used to transform the high-dimensional Hilbert space of all data qudits into a topological encoding. This Hamiltonian is built from several basic operators on a single data qudit:

\begin{equation}
\label{eq:L}
L^{g_0}_+ \ket{g} = \ket{g_0g}
\end{equation}
\begin{equation}
L^{g_0}_- \ket{g} = \ket{gg_0^{-1}}
\end{equation}
\begin{equation}
T^{h_0}_+ \ket{h} = \delta_{h_0,h}\ket{h}
\end{equation}
\begin{equation}
\label{eq:T}
T^{h_0}_- \ket{h} = \delta_{h_0^{-1},h}\ket{h}
\end{equation}

\noindent
where $\delta_{i,j}$ is the Kronecker delta function. These operators are defined for all elements $g_0,h_0 \in G$, and provide a faithful representation of the left/right multiplication and comultiplication in the Hopf algebra $\C[G]$. Using these operators, local gauge transformations and magnetic charge operators are defined as follows, on each vertex $v$ and plaquette $p$ \cite{Kitaev97}:

\begin{equation}
\label{eq:kitaev-vertex-g-term}
A^{g}(v,p) = A^{g}(v) = \prod_{j \in \text{star}(v)} L^g(j,v)
\end{equation}

\begin{equation}
\label{eq:kitaev-plaquette-h-term}
B^h(v,p) = \sum_{h_1 \cdots h_k = h} \prod_{m=1}^k T^{h_m}(j_m, p)
\end{equation}

Here, $j_1, ..., j_k$ are the boundary edges of the plaquette $p$ in counterclockwise order (see Fig. \ref{fig:kitaev}), and $L^g$ and $T^h$ are defined as follows: if $v$ is the origin of the directed edge $j$, $L^g(j,v) = L^g_-(j)$, otherwise $L^g(j,v) = L^g_+(j)$; if $p$ is on the left (resp., right) of the directed edge $j$, $T^h(j,p) = T^h_-(j)$ ($T^h_+(j)$) \cite{Kitaev97}.

Note that since the $A^g(v)$ satisfy $A^g(v) A^{g'}(v) = A^{gg'}(v)$, the set of all $A^g(v)$ (for fixed $v$) form a representation of $G$ on the entire Hilbert space $\mathcal{L}=\otimes_{e}\mathbb{C}[G]$ of all data qudits \cite{Bombin08}.

In fact, we can define operators

\begin{equation}
\label{eq:bulk-local-operators}
D^{(h,g)} (v,p) = B^h (v,p) A^g (v,p)
\end{equation}

\noindent
that act on a {\it cilium} $s = (v,p)$, where $v$ is a vertex of $p$. These operators act locally, and they form the basis of a quasi-triangular Hopf algebra $\mathcal{D} = \Span\{D^{(h,g)}\}$, the {\it quantum double} $D(G)$ of the group $G$. The specific multiplication, comultiplication, and antipode for the Hopf algebra are presented in Ref. \cite{Kitaev97}. As vector spaces, we have

\begin{equation}
D(G) = F[G] \otimes \C[G],
\end{equation}
where $F[G]$ are complex functions on $G$.

In the next sections, we will see the importance of these local operators in determining the topological properties of excitations in this group model.

Finally, two more linear combinations of these $A^g$ and $B^h$ operators are required to define the Hamiltonian:

\begin{equation}
\label{eq:kitaev-vertex-term}
A(v) = \frac{1}{|G|} \sum_{g \in G} A^g(v,p) 
\end{equation}

\begin{equation}
\label{eq:kitaev-plaquette-term}
B(p) = B^1(v,p).
\end{equation}

\noindent
The Hamiltonian\footnote{Note: We call this Hamiltonian $H_{(G,1)}$, as this model is the Dijkgraaf-Witten theory with trivial cocycle (twist). In general, this Hamiltonian may be twisted by a 3-cocycle $\omega \in H^3(G,\C^\times)$, and may be written as $H_{(G,\omega)}$.} for the Kitaev model is then defined as

\begin{equation}
\label{eq:kitaev-hamiltonian}
H_{(G,1)} = \sum_v (1-A(v)) + \sum_p (1-B(p))
\end{equation}

It is important to note that all terms in the Hamiltonian $H_{(G,1)}$ commute with each other. By the spectral theorem, this means that these operators share simultaneous eigenspaces. Each individual operator is a projector and has eigenvalues $\lambda = 0,1$. (Specifically, the $A(v)$ terms project onto the trivial representation, and the $B(p)$ terms project onto trivial flux \cite{Bombin08}.) The ground state of the Hamiltonian corresponds to the eigenspace with overall eigenvalue (energy) $\lambda = 0$, and states with excitations will have positive energy. Here, an excitation or a quasi-particle is defined so that exactly one of terms $(1-A(v))$ and one of the terms $(1-B(p))$ is in the $\lambda = 1$ eigenstate; we say the quasi-particle is located at the cilium $s=(v,p)$. The resulting quantum encoding is hence ``topological'': regardless of how densely we place the data qudits, there will always be a constant energy gap between the ground state and the first excited state, and between each excited state.

\subsection{Ribbon operators}
\label{sec:ribbon-operators}

In this section, we review the algebra of bulk ribbon operators for the Kitaev models as presented in Refs. \cite{Bombin08,Kitaev97}. These definitions play a crucial role in this section, as a major contribution of this section will be the presentation of boundary ribbon operators in Sections \ref{sec:bd-hamiltonian} and \ref{sec:bd-excitations}.

\subsubsection{Basic definitions}

Before we proceed to define the algebra of ribbon operators, let us first review the following basic definitions. In what follows, the {\it direct} lattice will denote the original lattice of the Kitaev model (cf. the {\it dual} lattice, in which vertices and plaquettes of the direct lattice are switched). Both lattices are shown in Fig. \ref{fig:ribbon-defs}.

\begin{figure}
\centering
\includegraphics[width = 0.4\textwidth]{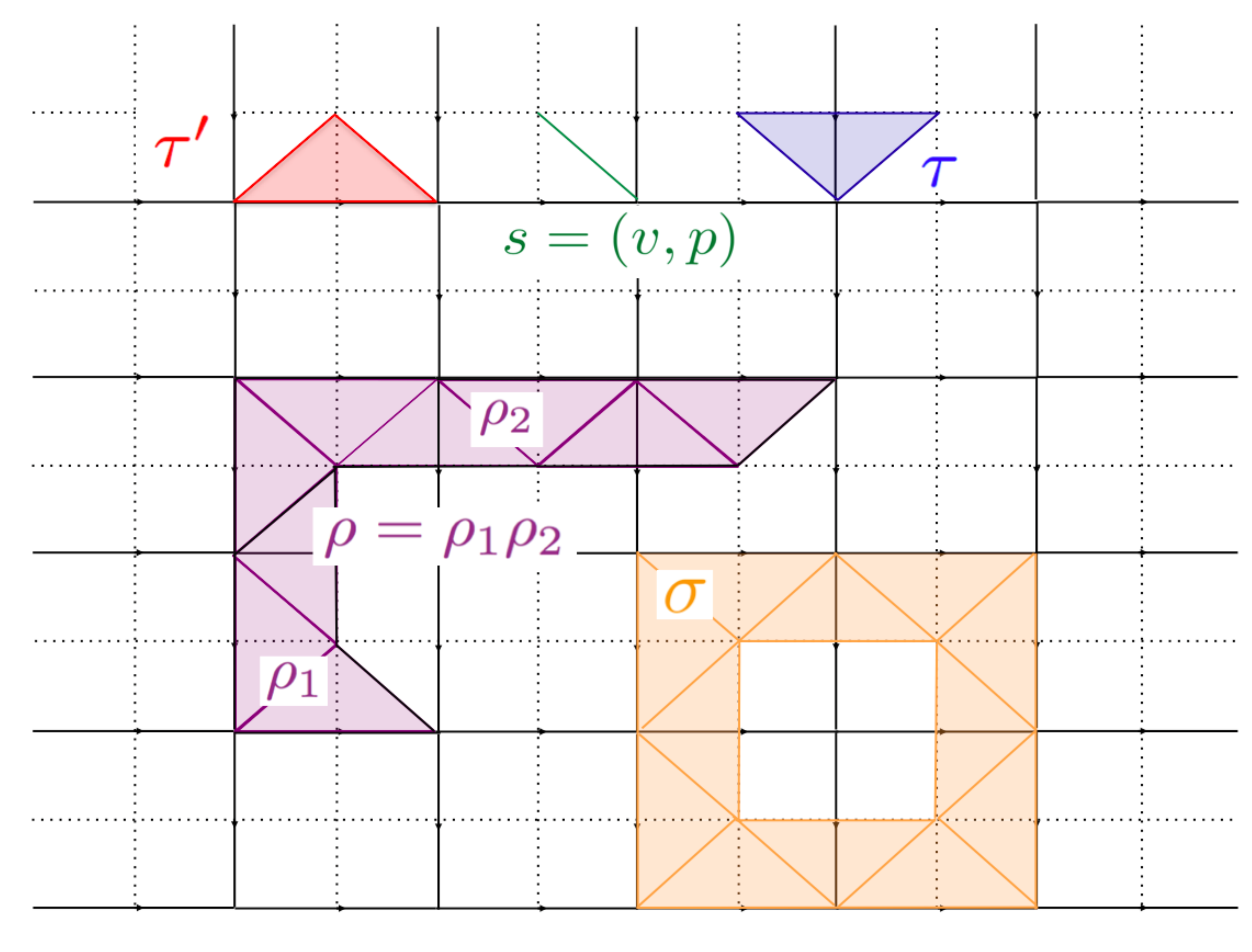}
\caption{Illustration of Definitions \ref{cilium-def}-\ref{ribbon-def}. The direct lattice is shown as before, and the dual lattice is shown in dotted lines. $s=(v,p)$ is a cilium. $\tau$ is a dual triangle, and $\tau'$ is a direct triangle. $\rho = \rho_1 \rho_2$ is a composite ribbon, formed by gluing the last site of $\rho_1$ to the first site of $\rho_2$. $\rho$ is an open ribbon, and $\sigma$ is a closed ribbon.}
\label{fig:ribbon-defs}
\end{figure}

\begin{definition}
\label{cilium-def}
A {\it cilium} is a pair $s = (v,p)$, where $p$ is a plaquette in the lattice, and $v$ is a vertex of $p$. These are visualized (e.g. in Fig. \ref{fig:ribbon-defs}) as lines connecting $v$ to the center of $p$ (i.e. the dual vertex corresponding to $p$).
\end{definition}

\begin{definition}
\label{triangle-def}
A {\it direct (dual) triangle} $\tau$ consists of two adjacent cilia $s_0,s_1$ connected via an edge $e$ on the direct (dual) lattice, as shown in Fig. \ref{fig:ribbon-defs}. We write $\tau = (s_0, s_1, e) = (\partial_0 \tau, \partial_1 \tau, e)$, listing sides in counterclockwise order. Throughout the section, $\tau$ will be used to denote a dual triangle, and $\tau'$ a direct triangle.
\end{definition}

\begin{definition}
\label{ribbon-def}
A {\it ribbon} $\rho$ is a oriented strip of triangles $\tau_1, ... \tau_n$, alternating direct/dual, such that $\partial_1 \tau_i = \partial_0 \tau_{i+1}$ for each $i = 1,2,...n-1$, and the intersection $\tau_i \cap \tau_j$ has zero area if $i \neq j$ (i.e. $\rho$ does not intersect itself).

$\rho$ is said to be {\it closed} if $\partial_1 \tau_n = \partial_0 \tau_1$. $\rho$ is {\it open} if it is not closed. Examples of closed and open ribbons are shown in Fig. \ref{fig:ribbon-defs}.
\end{definition}

\begin{definition}
\label{open-ribbon-operator-def}
Let $\rho$ be an open ribbon, with endpoint cilia $s_0 = (v_0, p_0)$, $s_1 = (v_1, p_1)$. A {\it ribbon operator} on $\rho$ is an operator $F_\rho$ that commutes with all terms of the Hamiltonian (\ref{eq:kitaev-hamiltonian}) except possibly the terms corresponding to $v_0,p_0,v_1,$ or $p_1$.
\end{definition}

The goal of this section is hence to determine the algebra $\mathcal{F}$ of ribbon operators in the bulk of the Kitaev model.

\subsubsection{Triangle operators and the gluing relation}

The ribbon operators are defined recursively \cite{Bombin08,Kitaev97}. As discussed in Refs. \cite{Bombin08,Kitaev97}, given a ribbon $\rho$, the set of ribbon operators on $\rho$ has basis elements $F^{(h,g)}_\rho$ indexed by two elements of $G$. The simplest ribbon is the empty ribbon $\epsilon$, for which the ribbon operators are given by

\begin{equation}
F^{(h,g)}_\epsilon = \delta_{1,g}.
\end{equation}

The next simple case is when $\rho$ is a single triangle. Let $\tau = (s_0,s_1,e)$ be any dual triangle, and let $\tau'=(s_0',s_1',e')$ be any direct triangle. The ribbon operators are defined as follows:

\begin{equation}
\label{eq:triangle-operator-def}
F^{(h,g)}_\tau := \delta_{1,g} L^h (e), \qquad F^{(h,g)}_{\tau'} := T^g (e')
\end{equation}

\noindent
In this definition, the choice of $+$ or $-$ for the $L,T$ operators is determined by the orientation of the edge on each triangle.

Finally, we define a ``gluing relation'' on ribbon operators. Let $\rho = \rho_1 \rho_2$ be the ribbon formed by gluing the last cilium of $\rho_1$ to the first cilium of $\rho_2$ (see Fig. \ref{fig:ribbon-defs}). We define the ribbon operator on this composite ribbon to be

\begin{equation}
\label{eq:ribbon-gluing}
F^{(h,g)}_\rho := \sum_{k \in G} F^{(h,k)}_{\rho_1} F^{(k^{-1}hk,k^{-1}g)}_{\rho_2}.
\end{equation}

It is simple to check that this definition makes $F^{(h,g)}_\rho$ independent of the particular choice of $\rho_1,\rho_2$ \cite{Bombin08}.

The operators $F^{(h,g)}$ also form a basis for a quasi-triangular Hopf algebra $\mathcal{F}$, as shown in Ref. \cite{Kitaev97}. In fact, Ref. \cite{Kitaev97} also shows that the algebra $\mathcal{F}$ is precisely the dual Hopf algebra to the quantum double $\mathcal{D} = D(G)$.

\subsubsection{Elementary excitations in the Kitaev model}

One of the most important applications of ribbon operators is to classify the elementary excitations\footnote{There are many terms in the literature that all refer to essentially the same thing: an elementary excitation, a simple quasi-particle, an anyon, or a simple object of $\mathcal{Z}(\Rep(G))$.  A topological charge or a superselection sector is an isomorphism class of all the above.} in the Kitaev model. As shown in the previous section, the operators $F^{(h,g)}_\rho$ create a pair of excitations at the endpoints of the ribbon $\rho$. However, these excitations may be a superposition of ``elementary'' excitations, so they can be easily split by random small perturbations.  Therefore, such composite excitations are not stable and may easily decohere into anyons. Let us formally define these elementary excitations as follows:

\begin{definition}
Let $\mathcal{E}$ denote the space of excitations that can be created at any cilium $s = (v,p)$ by applying a linear combination of the operators $F^{(h,g)}_\rho$ to some ribbon $\rho$ terminating at $s$. An {\it elementary excitation}
or {\it simple quasi-particle}
is given by a subspace of $\mathcal{E}$ that is preserved under the action of local operators $D^{(h,g)}(s)$ 
(defined in Eq. (\ref{eq:bulk-local-operators})),
that cannot be further decomposed (non-trivially) into the direct sum of such subspaces.
\end{definition}

Since these subspaces cannot be modified by local operators, they determine the ``topological charge'' of the excitation. On the other hand, the degrees of freedom within this subspace are purely local properties of the excitation. We define the quantum dimension of the excitation to be the square root of the dimension of this subspace.

It turns out that the basis $F^{(h,g)}$ for the algebra $\mathcal{F}$ is not useful in classifying elementary excitations. Instead, by Ref. \cite{Bombin08}, we have the following theorem:

\begin{theorem}
\label{bulk-anyon-types}
The elementary excitations of the Kitaev model with group $G$ are given by pairs $(C,\pi)$, where $C$ is a conjugacy class of $G$ and $\pi$ is an irreducible representation of the centralizer $E(C)$ of $C$.  Recall that the pairs $(C, \pi)$ are in bijection with the irreducible representations of the double of $G$.
\end{theorem}

\begin{proof}
To prove this theorem, we construct a change-of-basis for the ribbon operator algebra $\mathcal{F}$. Following Ref. \cite{Bombin08}, let us construct a new basis as follows:

\begin{enumerate}
\item
Choose an arbitrary element $r_C \in G$ and form its conjugacy class $C = \{gr_Cg^{-1}: g \in G\}$. Index the elements of $C$ so that $C = \{ c_i \}_{i=1}^{|C|}$.
\item
Form the centralizer $E(C) = \{ g \in G: gr_C = r_Cg \}$.\footnote{It is not hard to show that, up to conjugation, $E(C)$ depends only on $C$ and not on $r_C \in C$.}
\item
Form a set of representatives $P(C) = \{ p_i \}_{i=1}^{|C|}$ of $G/E(C)$, so that $c_i = p_i r_C p_i^{-1}$.
\item
Choose a basis for each irreducible representation $\pi$ of $E(C)$. Let $\Gamma_\pi(k)$ denote the corresponding unitary matrix for the representation of $k \in G$.
\item
The new basis is
\begin{align}
\begin{split}
\label{eq:elementary-ribbon-basis}
\{F^{(C,\pi);({\bf u,v})}_\rho: \text{ } C &\text{ a conjugacy class of }G, \text{ } \pi \in (E(C))_{\text{ir}},
\\&{\bf u} = (i,j), {\bf v} = (i', j'), 1 \leq i, i' \leq |C|, 1 \leq j,j' \leq \dim(\pi)\},
\end{split}
\end{align}
where $E(C)_{\text{ir}}$ denotes the irreducible representations of $E(C)$, and each $F^{(C,\pi);({\bf u,v})}_\rho$ is given by
\begin{equation}
\label{eq:bulk-ribbon-FT}
F^{(C,\pi);({\bf u,v})}_\rho := \frac{\dim(\pi)}{|E(C)|}
\sum_{k \in E(C)} \left(\Gamma_\pi^{-1}(k)\right)_{jj'}F^{(c_i^{-1}, p_i k p_{i'}^{-1})}.
\end{equation}
\end{enumerate}

We can also construct the inverse change of basis \cite{Bombin08}. Suppose we are given $g,h \in G$. Then:

\begin{enumerate}
\item
Let $C$ be the conjugacy class of $h^{-1}$. Index the elements of $C$ so that $C = \{ c_i \}_{i=1}^{|C|}$.
\item
Let $E(C)$ be the centralizer of $C$ as above.
\item
Form a set of representatives $P(C)$ of $G/E(C)$ as above.
\item
Any $g \in G$ belongs to a unique coset $p_i E(C) \in G/E(C)$. 
For each $g \in G$, let $i(g)$ denote the index of the corresponding $p_i$.
\item
Let $k_{(h,g)} = \left(p_{i(h^{-1})}\right)^{-1}g p_{i(g^{-1}h^{-1}g)}$.
\item
The inverse change-of-basis is
\begin{equation}
F^{(h,g)}_\rho =
\sum_{\pi \in E(C)_{\text{ir}}}\sum_{j,j' = 1}^{\dim(\pi)}
\left(\Gamma_\pi(k_{(h,g)})\right)_{jj'} F^{(C,\pi);(\bf{u,v})}_\rho
\end{equation}
where ${\bf u} = (i(h^{-1}), j)$, ${\bf v} = (i(g^{-1}h^{-1}g), j')$.
\end{enumerate}

The basis (\ref{eq:elementary-ribbon-basis}) is particularly useful because the parameters $(C,\pi)$ completely encode the global degrees of freedom of the particles created, and the $(\bf{u,v})$ completely encode the local degrees of freedom. Specifically, as shown in Ref. \cite{Bombin08}, different operators $F^{(C,\pi);(\bf{u,v})}_\rho$ with the same $(C,\pi)$ but different $(\bf{u,v})$ may be changed into one another by applying the local operators $D^{(h,g)}$ at the two endpoints $s_0,s_1$ of $\rho$. Similarly, if two ribbon operators in this new basis have different $(C,\pi)$ pairs, any operator that can change one to another must have support that connects $s_0$ and $s_1$. It follows that the elementary excitations of the Kitaev model are described precisely by pairs $(C,\pi)$, where $C$ is a conjugacy class of the original group $G$, and $\pi$ is an irreducible representation of the centralizer of $C$.

\end{proof}

Physically, in the basis (\ref{eq:elementary-ribbon-basis}), $C$ represents magnetic charge, and $\pi$ represents electric charge. The quantum dimension of an elementary excitation $(C,\pi)$ is given by the square root of the dimension of the subalgebra spanned by all $F^{(C,\pi);({\bf u,v})}_\rho$, or

\begin{equation}
\FPdim(C,\pi) = |C|\dim(\pi).
\end{equation}

As a special case, the simple particle $(C,\pi)$, where $C = \{1\}$ is the conjugacy class of the identity element and $\pi$ is the trivial representation, is the vacuum particle (i.e. absence of excitation). The vacuum particle always has a quantum dimension of 1.

When two anyons (given by pairs $(C_1,\pi_1)$ and $(C_2, \pi_2)$, say) are brought by ribbon operators to the same cilium on the lattice, one can essentially consider them as one composite anyon. Specifically, one can again consider the local operators $D^{(h,g)}$ acting on this cilium, which will determine new sets of local/global degrees of freedom on the new composite anyon. This process is known as {\it anyon fusion}. It can be shown that anyon fusion in this group model is described by the fusion rules of the unitary modular tensor category $\mZ(\Rep(G))$.

We would like to note that the change of basis to (\ref{eq:elementary-ribbon-basis}) and its inverse is essentially a generalized Fourier transform and its inverse. However, this generalized Fourier transform acts on the {\it operator algebra} $\mathcal{F}$, not on the vectors themselves. The numbers of the two algebra generators match because each pair $(C,\pi)$ corresponds to an irreducible representation of the quantum double $D(G)$. For general (non-abelian) groups, the Fourier basis is precisely given by matrix elements of the irreducible representations \cite{Moore06}.

\subsection{Quantum double models with boundaries}
\label{sec:bd-hamiltonian}

In previous sections, we have defined the Kitaev quantum double model on a sphere, or an infinitely large lattice on the plane. However, it is also important to consider the case where the lattice has boundaries/holes (e.g. in Fig. \ref{fig:boundary}). As discussed in Ref. \cite{Cong16a}, this is a powerful model with degeneracy that can allow for universal topological quantum computation. In this section, we present the Hamiltonian and ribbon operators for the Kitaev model with boundary. The Hamiltonian will be adapted from previous works on gapped boundaries and domain walls by Beigi et al. \cite{Beigi11} and Bombin and Martin-Delgado \cite{Bombin08}.

\begin{figure}
\centering
\includegraphics[width = 0.47\textwidth]{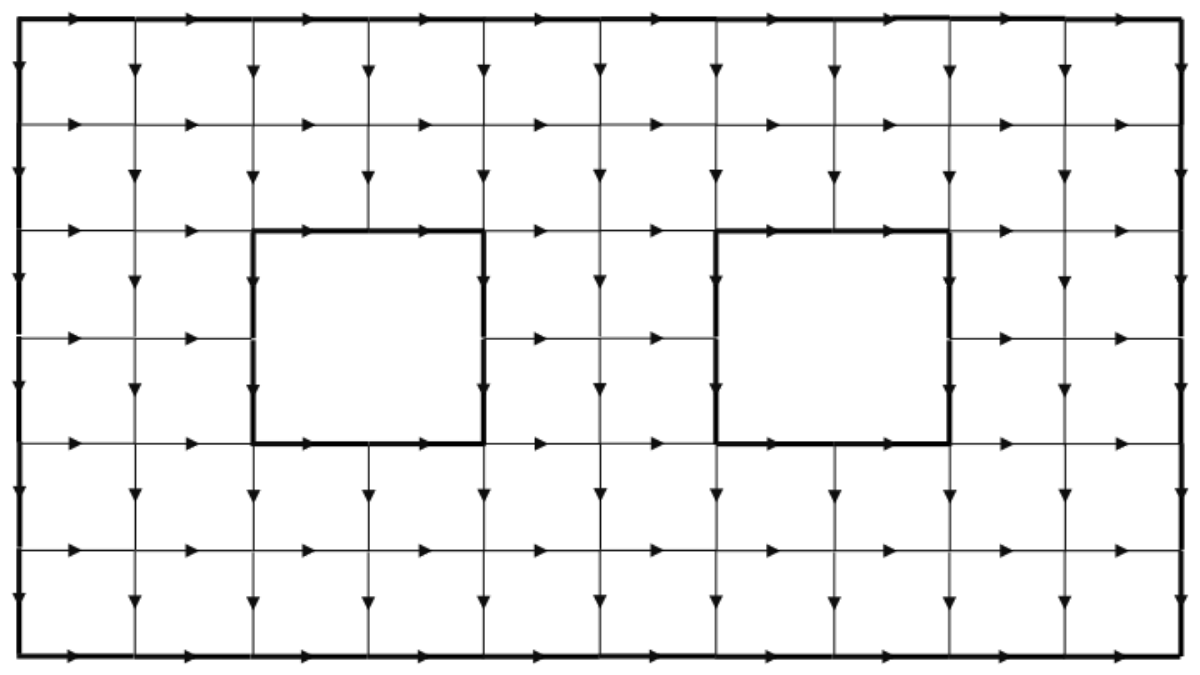}
\caption{Lattice for the Kitaev model with boundary. For any fixed group $G$, there can be multiple ways to define projection operators at the boundary such that all terms in the new Hamiltonian still commute. These are studied in Section \ref{sec:bd-hamiltonian}.}
\label{fig:boundary}
\end{figure}

\subsubsection{Hamiltonians for quantum double models with boundaries}

We will consider the model in which a gapped boundary is determined by a subgroup $K \subseteq G$ up to conjugation. In general, as shown in Ref. \cite{Beigi11}, a boundary is determined by both $K$ and a 2-cocycle $\phi \in H^2(K,\C^\times)$, and it is straightforward to generalize our results. To define a Hamiltonian for gapped boundaries, we must first define an {\it orientation} for each boundary.  Without loss of generality, we define all orientations so that the bulk is always on the left hand side when we traverse each boundary.

Let us now define some new projector terms, as in Ref. \cite{Bombin08}:

\begin{equation}
\label{eq:LK}
L^K(e) := \frac{1}{|K|} \sum_{k \in K} (L^k_+ (e) + L^k_- (e)),
\end{equation}

\begin{equation}
\label{eq:TK}
T^K(e) := \sum_{k \in K} T^k_+(e)
\end{equation}

The definitions of $L^k$ and $T^k$ for Eqs. (\ref{eq:LK}-\ref{eq:TK}) are based on Eqs. (\ref{eq:L}-\ref{eq:T}). 

Following Ref. \cite{Bombin08}, we can now define the following Hamiltonian\footnote{As before, we write $H^{(K,1)}_{(G,1)}$ to leave room for the generalized version, where a boundary depends also on a 2-cocycle $\phi$ of $K$.}: 

\begin{equation}
\label{eq:bd-hamiltonian-K}
H^{(K,1)}_{(G,1)} = \sum_e ((1-T^K(e)) + (1-L^K(e))
\end{equation}

It is important to note that as in the Hamiltonian (\ref{eq:kitaev-hamiltonian}), all terms in this Hamiltonian commute with each other. Hence $H^{(K,1)}_{(G,1)}$ is also gapped.

\begin{figure}
\centering
\includegraphics[width = 0.47\textwidth]{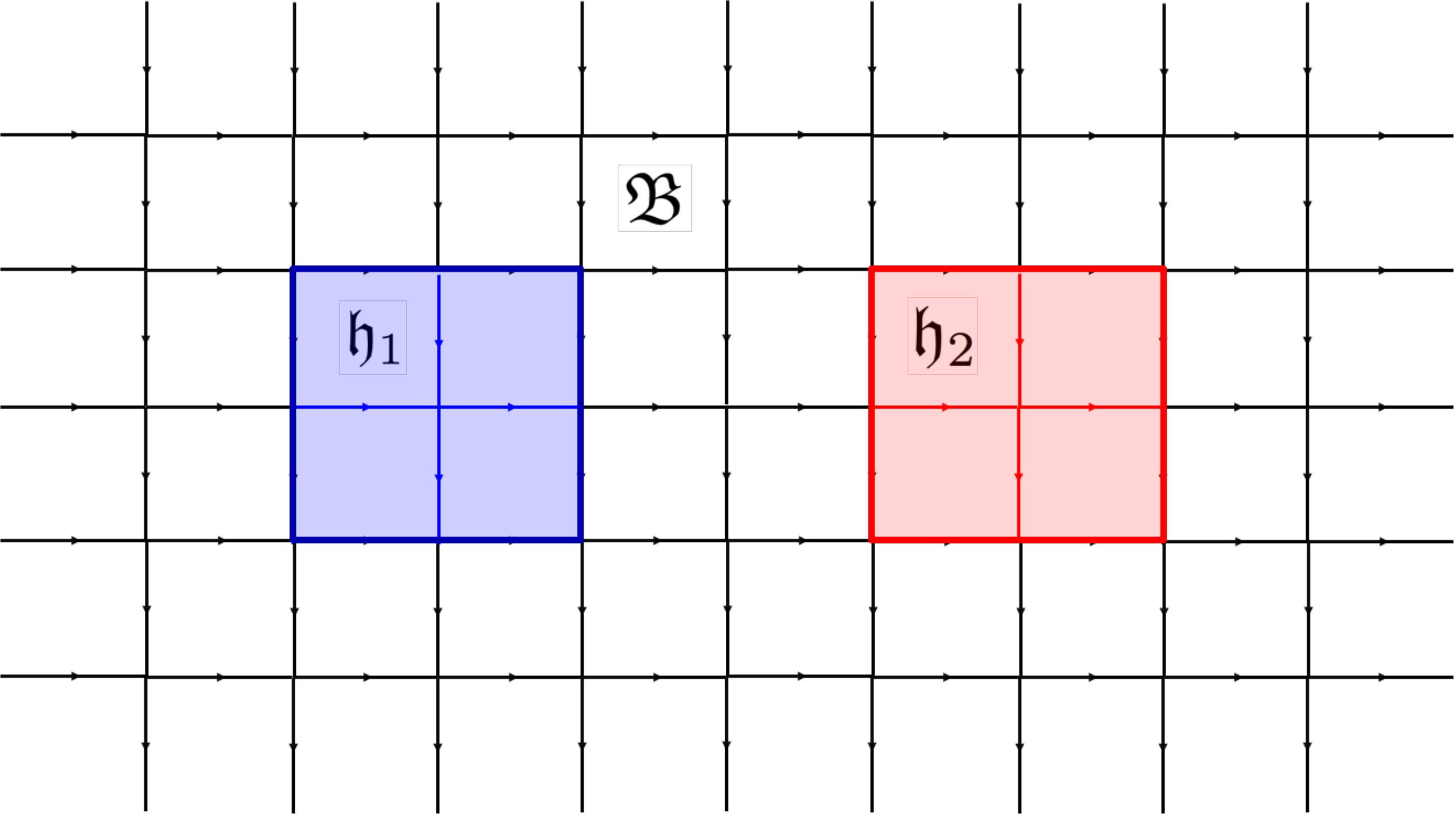}
\caption{Example: defining the Hamiltonian (\ref{eq:gapped-bds-hamiltonian}), in the case of two holes on an infinite lattice. The new Hamiltonians $H^{(K_1,1)}_{(G,1)}$ (resp., $H^{(K_2,1)}_{(G,1)}$) are applied to all edges within and along the blue (red) shaded rectangle.}
\label{fig:boundary-hamiltonian}
\end{figure}

We wish to take the standard Kitaev model, but modify the Hamiltonian in the presence of $n$ holes $\mathfrak{h}_1, ... \mathfrak{h}_n$ in the lattice, given by subgroups $K_1,...K_n$, respectively. Each hole is defined to contain all vertices, plaquettes, and edges within and along its border. Let $\mathfrak{B}$ denote the bulk, i.e. the complement of $\cup_i \mathfrak{h}_i$. The situation is shown in Fig. \ref{fig:boundary-hamiltonian}. The new Hamiltonian for this gapped boundary model will be defined as follows:

\begin{equation}
\label{eq:gapped-bds-hamiltonian}
H_{\text{G.B.}} = H_{(G,1)}(\mathfrak{B}) + \sum_{i=1}^{n} H^{(K_i,1)}_{(G,1)}(\mathfrak{h}_i).
\end{equation}

Here, $H^{(K_i,1)}_{(G,1)}(\mathfrak{h}_i)$ indicates that the Hamiltonian $H^{(K_i,1)}_{(G,1)}$ is acting on all edges of the hole $\mathfrak{h}_i$, and similarly for $H_{(G,1)}(\mathfrak{B})$. As in the cases of (\ref{eq:kitaev-hamiltonian}) and (\ref{eq:bd-hamiltonian-K}), all terms in the Hamiltonian commute with each other, and $H_{\text{G.B.}}$ is also gapped.

\begin{remark}
\label{bd-ribbon-def}
As discussed in Ref. \cite{Bombin08}, the Hamiltonian $H^{(K,1)}_{(G,1)}$, $K \subseteq G$, reduces the gauge symmetry of the original Hamiltonian $H_{(G,1)}$ to the trivial one (equivalent to vacuum) in all areas to which it is applied. Hence, if it is preferable, we may simply have $H^{(K,1)}_{(G,1)}$ act on a border of the hole with a width of a single plaquette, and empty space beyond it. Because of this, the term ``boundary'' will henceforth be used to refer to the ribbon that runs along the line dividing two different Hamiltonians and lies within the region of $H^{(K,1)}_{(G,1)}$, as anything beyond this boundary ribbon is essentially vacuum. Similarly, boundary cilia will be cilia along this ribbon. This configuration is illustrated in Fig. \ref{fig:boundary-hamiltonian-2}.

Similarly, one can consider the case where the lattice has an external boundary given by subgroup $K_0$. In this case, it is not practical or necessary to have $H^{(K_0,1)}_{(G,1)}$ act on all (i.e. infinitely many) data qudits outside the original lattice. Instead, we will simply have the Hamiltonian $H^{(K_0,1)}_{(G,1)}$ act on a border of the entire lattice with a width of a single plaquette. Anything beyond the border may then be regarded as empty space. This is also illustrated in Fig. \ref{fig:boundary-hamiltonian-2}.
\end{remark}

\begin{remark}
\label{single-bd-creation}
We would like to note that the Hamiltonian $H_{\text{G.B.}}$ can be used to create just a single gapped boundary, unlike anyons in the bulk, which must be created in pairs.  Hamiltonians that generate gapped boundaries of different sizes and locations can be physically connected by a path of gapped Hamiltonians.  Therefore, a gapped boundary may be moved via adiabatic Hamiltonian tuning of the Hamiltonian $H_{\text{G.B.}}$, to enlarge, shrink, or move the hole.
\end{remark}

\begin{figure}
\centering
\includegraphics[width = 0.47\textwidth]{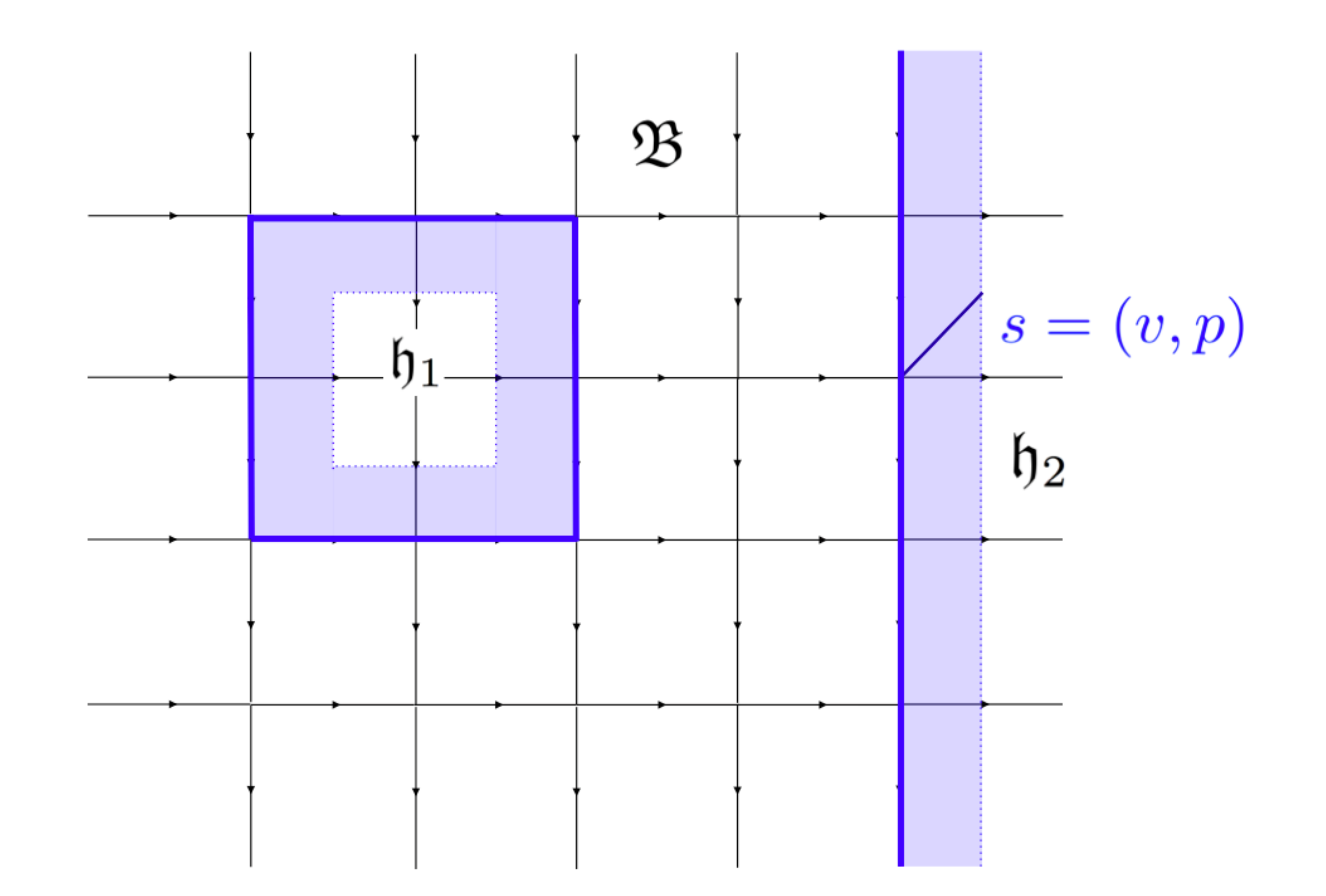}
\caption{Definitions of boundary ribbon/cilia. Given boundary lines (boldfaced) dividing regions of bulk/boundary Hamiltonians as shown, the boundary ribbons are the shaded ribbons. A cilium on the boundary ribbon (e.g. $s$ in the figure) is said to be a boundary cilium. Since the Hamiltonians $H^{(K,1)}_{(G,1)}$ break all gauge symmetries, anything beyond the boundary ribbon in the region of $H^{(K,1)}_{(G,1)}$ is essentially vacuum; one may ignore all Hamiltonian terms there, if desired.}
\label{fig:boundary-hamiltonian-2}
\end{figure}

\begin{figure}
\centering
\includegraphics[width = 0.47\textwidth]{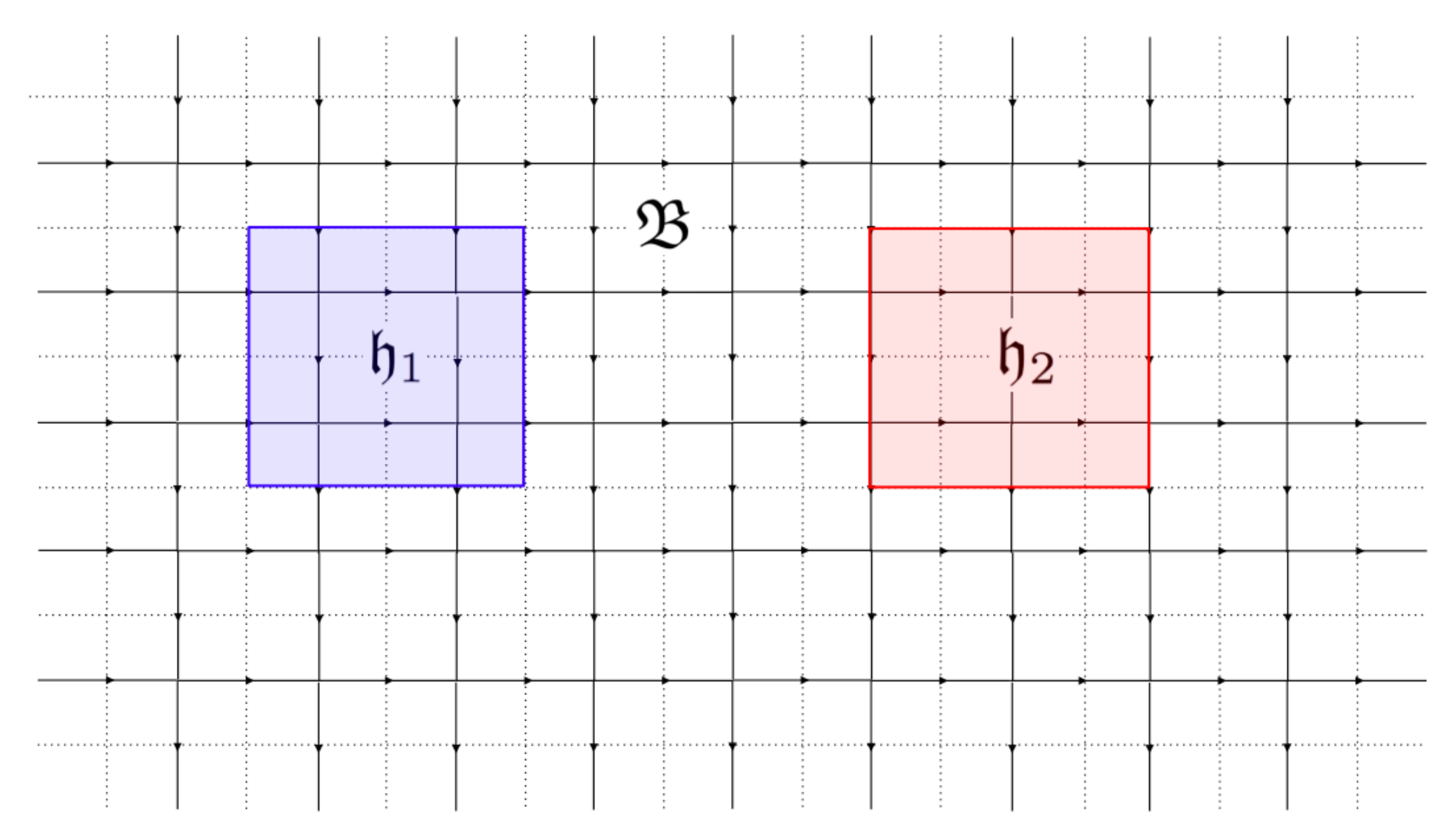}
\caption{Definition of the Hamiltonian $H_{\text{G.B.}}$, in cases where some (e.g. $\mathfrak{h}_2$) or all (e.g. $\mathfrak{h}_1$) of the hole's sides lie on the dual lattice.}
\label{fig:boundary-hamiltonian-3}
\end{figure}

\begin{remark}
\label{dual-bd-rmk}
In this section, we have defined the Hamiltonian $H_{\text{G.B.}}$ that can create holes in the lattice, whose sides lie on the direct lattice (see Fig. \ref{fig:boundary-hamiltonian}). More generally, we can create holes where some or all of the sides lie on the dual lattice, such as $\mathfrak{h}_1$ and $\mathfrak{h}_2$, respectively, in Fig. \ref{fig:boundary-hamiltonian-3}. For both cases, we say that the Hamiltonian $H^{(K_i,1)}_{(G,1)}$ acts on edges of the boldfaced boundary of the square. In general, the properties of holes completely on the dual lattice such as $\mathfrak{h}_1$ are almost the same as those of holes on the direct lattice; the only difference is up to an electric-magnetic duality in the model. However, holes such as $\mathfrak{h}_2$ that are partially on the dual lattice are of more interest. In this case, the hole is essentially associated with two different boundary types (e.g. two different subgroups $K_1, K_2 \subseteq G$), related to each other by this symmetry. Such holes were briefly considered by Fowler et al. in Ref. \cite{Fowler12} in the special case of the toric code. Their properties are discussed in detail in Ref. \cite{Cong16b}.
\end{remark}

\subsubsection{Local operators}
\label{sec:bd-local-operators}

We now introduce an algebra of local operators $\mathcal{Z} = Z(G,1,K,1)$ that act on a cilium $s = (v,p)$ on the boundary. In fact, these local operators form a {\it quasi-Hopf algebra} as defined by Drinfeld \cite{Drinfeld89}. We call this the {\it group-theoretical quasi-Hopf algebra} based on the group $G$, subgroup $K$, and trivial cocycles $\omega \in H^3(G, \C^\times)$ and $\psi \in H^2(K, \C^\times)$, following the definition of the group-theoretical category $\mC(G,\omega,K,\psi)$ in Ref. \cite{Etingof05}. The algebra was first constructed by Zhu in Ref. \cite{Zhu01}, and is also discussed in detail in Ref. \cite{Schauenburg02}. The following operators form a basis for $\mathcal{Z}$:

\begin{equation}
\label{eq:bd-local-operators}
Z^{(hK,k)}(v,p) = B^{hK}(v,p)A^k(v,p)
\end{equation}

\noindent
where we have defined

\begin{equation}
B^{hK}(v,p) = \sum_{j \in hK} B^j(v,p).
\end{equation}

\noindent
Here, $k \in K$ is an element of the subgroup, and $hK = \{hk: k \in K\}$ is a left coset. Here, we only need to consider the local vertex operators $A^k$ where $k \in K$, as the actions of $A^g(v,p)$, $g \notin K$ on the representations of $K$ at the vertex $v$ and edges in $\text{star}(v)$ are linear combinations of the actions of $A^k(v,p)$. Similarly, since the $B^K$ terms of the Hamiltonian (\ref{eq:bd-hamiltonian-K}) project onto flux in the subgroup $K$, we only need to consider the generalization where the flux lies within a left coset (instead of restricting to a particular group element). Hence, as vector spaces, we have

\begin{equation}
Z(G,1,K,1) = F[G/K] \otimes \C[K],
\end{equation}

\noindent
i.e. $Z(G,1,K,1)$ is the tensor product of the algebra of complex functions over the left cosets of $K$ and the group algebra $\C[K]$. The multiplication, comultiplication and antipode for the quasi-Hopf algebra $\mathcal{Z}$ are presented in Refs. \cite{Zhu01} and \cite{Schauenburg02} (Section 4.5), where the authors show that they satisfy all the quasi-Hopf algebra axioms. Below, we present them in the context of these local operators. In Section \ref{sec:algebraic-condensation}, we show how the representation category of $Z(G,1,K,1)$ is the group-theoretical category $\mC(G,1,K,1)$ as defined in \cite{Etingof05}.

\subsubsection{Quasi-Hopf algebra structure of $\mZ$}

We now present the quasi-Hopf algebra structure of the local operator algebras $\mathcal{Z} = Z(G,1,K,1)$. We only need to consider the case where $G$ is a finite group. This construction is similar to the construction of the quasi-triangular Hopf structures of the quantum double $\mathcal{D} = D(G)$ in Ref. \cite{Kitaev97}. 

As discussed above, the basis vectors of $\mZ$ are of form

\begin{equation}
Z^{(hK,k)} = B^{hK}A^k
\end{equation}

\noindent
for some $k \in K$, $hK \in G / K$. For convenience, let us first define the following notations. Let $R$ be a set of representatives of the left cosets $G / K$. Then, every $g \in G$ can be written uniquely as $g = r(g) \{g \}^{-1}$ for some $r(g) \in R$, $\{g\} \in K$.

By definition of the operators $B^{hK}$, $A^k$, we have the following rule for multiplication in $\mZ$ (where we assume without loss of generality $h_1, h_2 \in R$):

\begin{align}
\begin{split}
Z^{(h_1 K, k_1)} Z^{(h_2 K, k_2)} &=
\sum_{\substack{h \in R \\
				j,k \in K}}
\delta_{h_1 K, hK} \delta_{k_1, k} \delta_{h_2 K, k^{-1}h K} \delta_{k_2, k^{-1} j} Z^{(hK,j)}\\
&= \delta_{h_1 K, k_1 h_2 K} Z^{(h_1 K, k_1 k_2)}
\end{split}
\end{align}

\noindent
(Here, we note that certain Kronecker deltas may be taken between left cosets, e.g. $\delta_{h_1 K, hK}$.) 
Similarly, the following rule is used to define comultiplication in $\mZ$:

\begin{equation}
\Delta(Z^{(hK, k)}) =
\sum_{h_1  \in R}
(Z^{(r({h_1^{-1} h)}K, \{k^{-1} h_1\})} \otimes Z^{(h_1 K, k)})
\end{equation}

\noindent
The antipode in $\mZ$ is given by:

\begin{equation}
S(Z^{(hK,k)}) = Z^{( r({r({k^{-1}h})^{-1}}) K, \{ k^{-1}h \}^{-1} )} 
\end{equation}

\noindent
In both equations above, we assume $h \in R$. The specific elements $\alpha, \beta \in A$ corresponding to the quasi-Hopf structure of the antipode are:

\begin{equation}
\alpha = Z^{(K,1)}, \qquad \beta = \sum_{h \in R} Z^{(hK, \{h^{-1}\})}
\end{equation}

\noindent
Finally, the Drinfeld associator is given by:

\begin{equation}
\Phi = \sum_{h_1, h_2, h_3 \in R} Z^{(h_1 K, 1)} \otimes Z^{(h_2 K, 1)} \otimes Z^{(h_3 K, \{ h_1 h_2 \})}
\end{equation}

\subsubsection{Ribbon operators}
\label{sec:bd-ribbon-operators}

We will now describe the (dual) coquasi-Hopf algebra $\mathcal{Y} = Y(G,K)$ of ribbon operators which create all possible excited states on the boundary that may result from pushing a bulk anyon into the boundary (see Fig. \ref{fig:condensation}). In this case, a {\it boundary ribbon operator} is defined as an operator that is supported on a boundary ribbon (and acts trivially elsewhere) and commutes with all vertex and plaquette terms in the Hamiltonian $H_{\text{G.B.}}$ except at the two end cilia. These are defined recursively, as in the case of the bulk ribbon operators. As with the local operators $\mZ$, these will also be indexed by a pair $(hK,k)$, $hK$ a left coset, and $k \in K$. As before, we have the following definition for the trivial ribbon:

\begin{equation}
Y^{(hK,k)}_\epsilon := \delta_{1,k}
\end{equation}

Similarly, let $\tau = (s_0,s_1,e)$ be any dual triangle, and let $\tau'=(s_0',s_1',e')$ be any direct triangle. We define

\begin{equation}
\label{eq:bd-triangle-operator-def}
Y^{(hK,k)}_\tau := \delta_{1,k} L^{hK} (e), \qquad Y^{(hK,k)}_{\tau'} := T^k (e')
\end{equation}

\noindent
where we have defined

\begin{equation}
L^{hK}(e) := \frac{1}{|hK|} \sum_{j \in hK} L^j(e).
\end{equation}

\noindent
and the $+$, $-$ orientations for $L$ and $T$ depend on the orientation of $e$ on the triangle.

Finally, we have the following gluing relation: If $\rho = \rho_1 \rho_2$ is a composite ribbon on the boundary, then

\begin{equation}
\label{eq:Y-gluing-formula}
Y^{(hK,k)}_\rho = \sum_{j \in K} Y^{(hK,j)}_{\rho_1} Y^{(j^{-1}hjK,j^{-1}k)}_{\rho_2}.
\end{equation}

Simple group theory manipulations show that for any $h\in G,j \in K$ the left coset $j^{-1}hjK$ depends only on the left coset $hK$, and not on the particular representative $h$. By the same argument as in Ref. \cite{Bombin08}, the operator $Y^{(hK,k)}_\rho$ is independent of the particular choice of $\rho_1,\rho_2$. Hence, Eq. (\ref{eq:Y-gluing-formula}) is well-defined.

Looking at Eq. (\ref{eq:bd-triangle-operator-def}), we see that as vector spaces,

\begin{equation}
\label{eq:Y-vector-space}
Y(G,1,K,1) = \C[G/K] \otimes F[K].
\end{equation}

\noindent
By performing the same detailed analysis of the ribbon operators as we did with the local operator algebra $\mZ$, one can show as an exercise using these operators directly that $\mathcal{Y}$ forms a coquasi-Hopf algebra dual to $\mathcal{Z}$ \cite{Drinfeld89}. (In the abstract context, the multiplication and comultiplication for the coquasi-bialgebra $\mathcal{Y}$ have been presented in Refs. \cite{Zhu01,Schauenburg02}. Because $G$ is finite, one can also define an antipode map $S: \mathcal{Y} \rightarrow \mathcal{Y}$ so that the resulting coquasi-bialgebra is indeed a coquasi-Hopf algebra as in Corollary 3.7 of \cite{Schauenburg02b}.) In particular, we would like to note now that the gluing formula (\ref{eq:Y-gluing-formula}) is in fact the comultiplication of $\mathcal{Y}$. Furthermore, we see that Eq. (\ref{eq:Y-gluing-formula}) corresponds to the multiplication of $\mathcal{Z}$, confirming the duality of $\mathcal{Y}$ and $\mathcal{Z}$.

This gluing procedure is very similar to the movement of a bulk anyon via bulk ribbon operators and the gluing formula of Eq. (\ref{eq:ribbon-gluing}). However, there is one very important difference: After applying the gluing formula (\ref{eq:ribbon-gluing}) to move a bulk anyon from the endpoint $s_1$ of the ribbon $\rho_1$ to the endpoint $s_2$ of $\rho_2$, the resulting operator $F^{(h,g)}_\rho$ on the ribbon $\rho = \rho_1 \rho_2$ now commutes with the terms in the Hamiltonian $H_{(G,1)}$ of Eq. (\ref{eq:kitaev-hamiltonian}) at the cilium $s_1$. Instead, the only places where $F^{(h,g)}_\rho$ does not commute with the terms in $H_{(G,1)}$ are $s_2$ and the other endpoint cilium of $\rho_1$. In this new case, applying the gluing formula (\ref{eq:Y-gluing-formula}) on such a ribbon $\rho = \rho_1 \rho_2$ on the boundary will not allow $Y^{(hK,k)}_\rho$ to commute with the edge terms $L^K(e), T^K(e)$ surrounding $s_1$. Hence, the excitations in the boundary are {\it confined}: the energy required to move an excitation along a boundary ribbon $\rho$ is linearly proportional to the length of $\rho$ (measured in the number of dual triangles). Physically, this means that $L^K, T^K$ in the Hamiltonian (\ref{eq:bd-hamiltonian-K}) represent string tension terms, which break all gauge symmetries past the boundary.

\begin{remark}
The above definition of boundary ribbon operators can create all excitations that can be formed by the condensation a bulk anyon to the boundary (discussed in the next two sections). By the detailed analysis of Ref. \cite{Bombin08}, all excitations within the region of the Hamiltonian $H^{(K,1)}_{(G,1)}$ are confined, and no particles are deconfined. This means any definition of excitation-creating operators and the gluing relation will always make the energy cost to move an excitation linear in the ribbon length. For instance, another set of excitation-creating operators for $H^{(K,1)}_{(G,1)}$ would be those that act only on edges, and not on triangles. In this case, the vertex and plaquette terms would serve to confine particles, instead of the edge terms $L^K, T^K$. 
\end{remark}

\subsection{Degeneracy and condensations to vacuum}
\label{sec:hamiltonian-gsd-condensation}

The standard Kitaev model has no ground state degeneracy on a surface with trivial topology like the infinite plane or the sphere, as the only operators that commute with the Hamiltonian $H_{(G,1)}$ of (\ref{eq:kitaev-hamiltonian}) are closed ribbon operators on contractible loops, which can be expressed as a linear combination of products of $A(v)$ or $B(p)$ \cite{Kitaev97}. However, once boundaries are introduced, one can construct operators that commute with the Hamiltonian $H_{\text{G.B.}}$ of (\ref{eq:gapped-bds-hamiltonian}) that cannot be expressed as such a product.

\begin{figure}
\centering
\includegraphics[width = 0.47\textwidth]{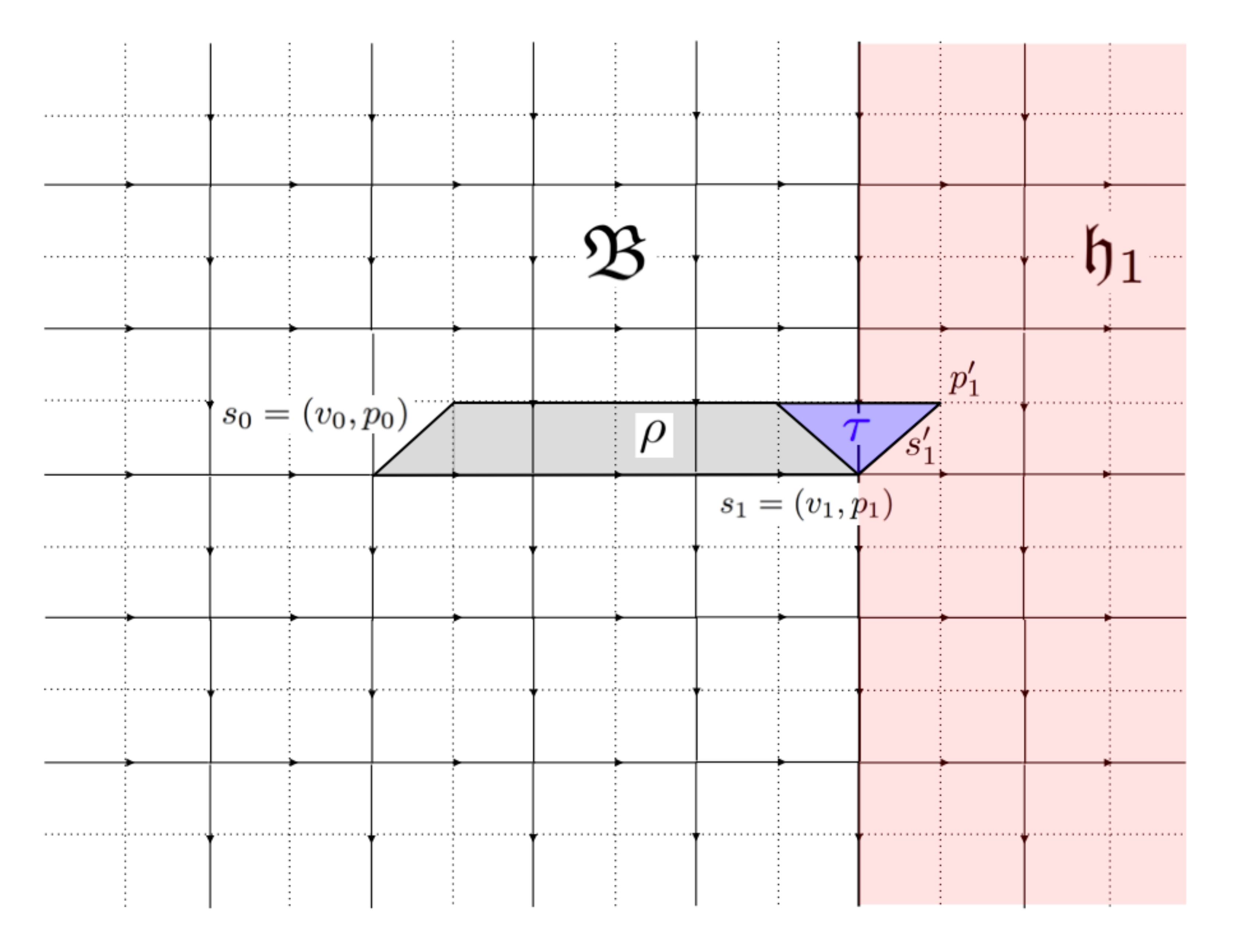}
\caption{Illustration of the bulk-to-boundary condensation procedure. If the new ribbon operator $F_{\rho \cup \tau}$ commutes with the Hamiltonian terms $A_{v_1}$ and $B_{p_1'}$  at the cilium $s_1'$, we say that the anyon $a$ has condensed to vacuum on the boundary.}
\label{fig:condensation}
\end{figure}

Let us first consider a scenario where we create a pair of anyons $a,\overbar{a}$ in the bulk from vacuum, by applying a ribbon operator $F^{(C,\pi);(\bf{u,v})}_\rho$. Without loss of generality, we suppose we have chosen $\rho$ with endpoints $s_0 = (v_0,p_0)$ and $s_1 = (v_1,p_1)$ such that $a$ is located at $s_1$ and is as close to the boundary as possible, as shown in Fig. \ref{fig:condensation}. Specifically, $s_1 = (v_1,p_1)$, where $v$ already lies on the line separating two regions with different Hamiltonians.

Suppose we would like to extend $\rho$ to $\rho \cup \tau$ and push the anyon $a$ into the boundary. We can apply the gluing formula (\ref{eq:ribbon-gluing}) as we would in the bulk, and the original excitation is pushed to the boundary cilium $s_1'=(v_1,p_1')$. However, since the Hamiltonian terms at $s_1'$ are different from those at $s_1$, it is possible that the new ribbon operator now commutes with all Hamiltonian terms in the vicinity of $s_1'$. The only terms that do not commute with $F^{(C,\pi);(\bf{u,v})}_\rho$ are now the terms corresponding to $s_0$. In this case, we say the anyon $a$ has {\it condensed to vacuum} in the boundary.

\begin{figure}
\centering
\includegraphics[width = 0.47\textwidth]{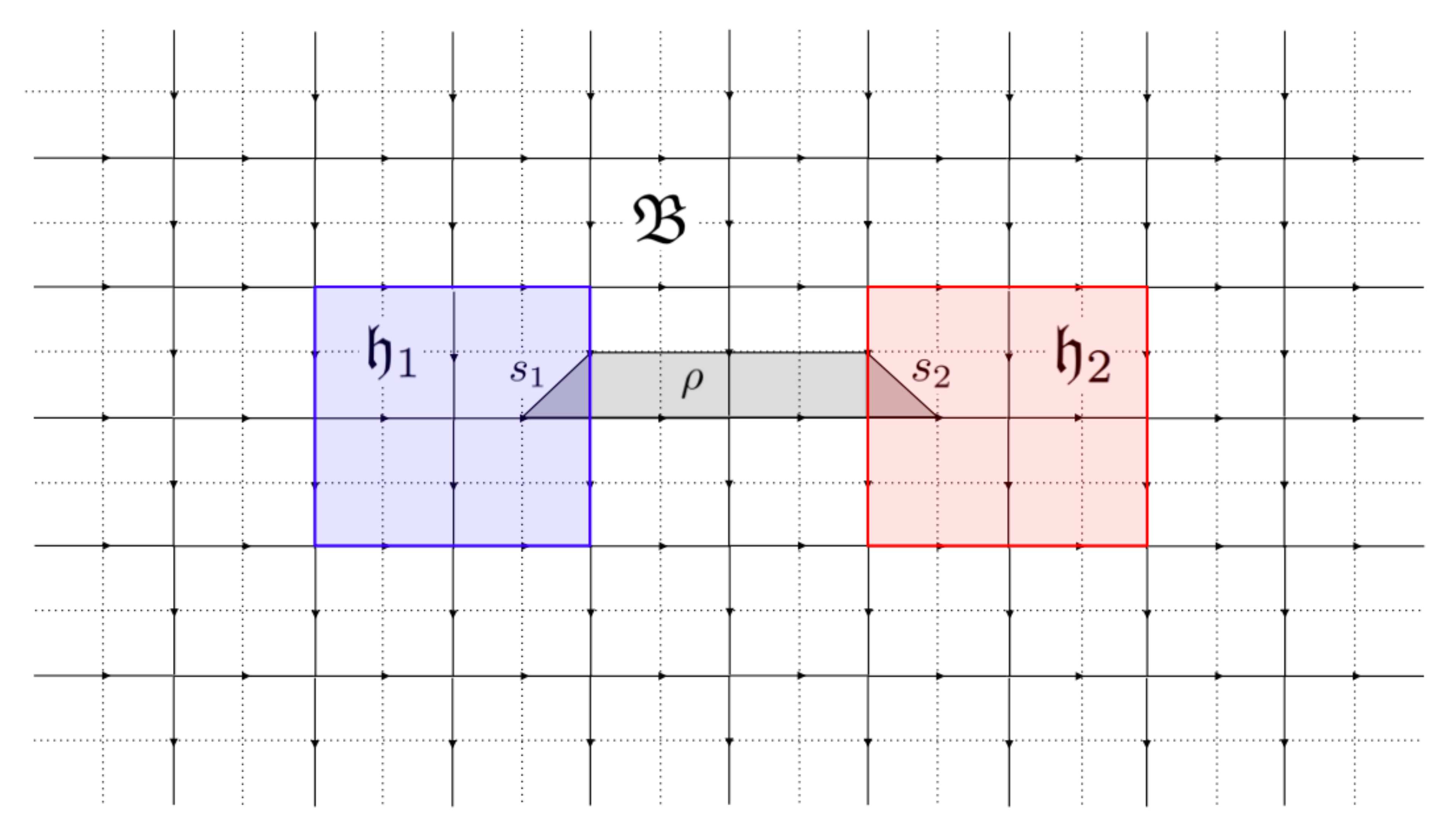}
\caption{Ground state degeneracy in the Kitaev model with boundary. By definition, a ribbon operator $F_\rho$ on $\rho$ commutes with all Hamiltonian terms in the bulk. Because the cilia $s_1$ and $s_2$ lie in areas with ribbon operators, it is possible that $F_\rho$ may also commute with the Hamiltonian terms at these cilia. The algebra of such operators $F_\rho$ will form the degenerate ground state of this system.}
\label{fig:degeneracy}
\end{figure}

Suppose we now have two holes $\mathfrak{h_1},\mathfrak{h_2}$ in the lattice, as shown in Fig. \ref{fig:degeneracy}. $\mathfrak{h_1},\mathfrak{h_2}$ are in the ground state of the Hamiltonians $H^{(K_1,1)}_{(G,1)}$, $H^{(K_2,1)}_{(G,1)}$, respectively. Again, let us consider a ribbon operator $F^{(C,\pi);(\bf{u,v})}_\rho$ which creates anyons $a,\overbar{a}$ in the bulk; it may now be possible be possible to condense $a$ to vacuum along the boundary of $\mathfrak{h_1}$, and $\overbar{a}$ to vacuum along the boundary of $\mathfrak{h_2}$, if the ribbon operator commutes with the boundary Hamiltonians in $\mathfrak{h_1}$, $\mathfrak{h_2}$.

It is now clear that the ground state of the Hamiltonian (\ref{eq:gapped-bds-hamiltonian}) may have nontrivial degeneracy. Specifically, each such operator $F^{(C,\pi);(\bf{u,v})}_\rho$ now corresponds to an operator $W_{(C,\pi);(\bf{u,v})}$ that commutes with the Hamiltonian. As discussed in Ref. \cite{Cong16a}, this is a powerful degeneracy that may be harnessed for purposes such as universal topological quantum computation.

\subsection{Excitations on the boundary}
\label{sec:bd-excitations}

\subsubsection{Elementary excitations}

In Section \ref{sec:bd-ribbon-operators}, we defined a basis for the coquasi-Hopf algebra $\mathcal{Y}$ of ribbon operators that create arbitrary excited states on the boundary of the Kitaev model given by subgroup $K$ that can result from bulk-to-boundary condensation. However, as in the case of bulk ribbon operators, we would like to classify the elementary excitations on the boundary. This is described in the following theorem:

\begin{theorem}
\label{bd-anyon-types}
The elementary excitations on a subgroup $K$ boundary of the Kitaev model with group $G$ are given by pairs $(T,R)$, where $T \in K\backslash G / K$ is a double coset, and $R$ is an irreducible representation of the stabilizer $K^{r_T} = K \cap r_T K r_T^{-1}$ ($r_T \in T$ is any representative of the double coset).
\end{theorem}

\begin{proof}
As before, we must perform a Fourier transform on the group-theoretical coquasi-Hopf algebra to obtain a new basis. This change-of-basis formula is constructed as follows:

\begin{enumerate}
\item
Choose a representative $r_T \in G$ and construct the corresponding double coset $T = K r_T K \in K\backslash G/K$.
\item
Construct the subgroup $K^{r_T} = K \cap r_T K r_T^{-1}$.
\item
Construct a set of representatives $Q$ of $K/K^{r_T}$. Label the elements of $Q$ as $Q = \{ q_i \}_{i=1}^{|Q|}$.
\item
Choose an irreducible representation $R$ of the subgroup $K^{r_T}$. Choose a basis for $R$ and denote the resulting unitary representation matrices  $\Gamma_R(k)$ for $k \in K^{r_T}$.
\item
For each $i = 1,2,...|Q|$, let $s_i = q_i r_T q_i^{-1}$. Construct the set of right cosets $S_R(T) = \{ K s_i\}_{i=1}^{|Q|}$. Simple group theory shows that the set $S_R(T)$ forms a partition of $T$. Similarly, the set of left cosets $S_L(T) = \{ s_i^{-1}K\}_{i=1}^{|Q|}$ is also a partition of $T$.
\item
The new basis is
\begin{align}
\begin{split}
\label{eq:elementary-bd-ribbon-basis}
\{Y^{(T,R);({\bf u,v})}_\rho: \text{ } T = K r_T K \in K\backslash G/K, \text{ } R \in (K^{r_T})_{\text{ir}},
\\{\bf u} = (i,j), {\bf v} = (i', j'), 1 \leq i, i' \leq |Q|, 1 \leq j,j' \leq \dim(R)\},
\end{split}
\end{align}
where each $Y^{(T,R);({\bf u,v})}_\rho$ is given by
\begin{equation}
\label{eq:bd-ribbon-FT}
Y^{(T,R);({\bf u,v})}_\rho := \frac{\dim(R)}{|K^{r_T}|}
\sum_{k \in K^{r_T}} \left(\Gamma_R^{-1}(k)\right)_{jj'}Y^{(s_i^{-1}K, q_i k q_{i'}^{-1})}.
\end{equation}
\end{enumerate}

As before, this Fourier basis for $\mathcal{Y}$ completely separates the topological and local degrees of freedom in the created excitations. It is straightforward to show that linear combinations of the local operators $Z^{(hK,k)}$ at the endpoints $s_0,s_1$ of $\rho$ may be used to transform any $Y^{(T,R);({\bf u,v})}_\rho$ into another basis operator that differs in only the pair $({\bf u,v})$. Similarly, if two operators in the basis (\ref{eq:elementary-bd-ribbon-basis}) have different pairs $(T,R)$, any operator that can change one to another must have support that connects $s_0$ and $s_1$. We can now conclude that the elementary excitations on the boundary of the Kitaev model are described precisely by pairs $(T,R)$, where $T= K r_T K$ is a double coset, and $R$ is an irreducible representation of the group $K^{r_T}$. 
\end{proof}

In the basis (\ref{eq:elementary-bd-ribbon-basis}), the quantum dimension of $(T,R)$ is given by the square root of the dimension of the subalgebra spanned by all $Y^{(T,R);({\bf u,v})}_\rho$, or

\begin{equation}
\FPdim(T,R) = |Q|\dim(R) = \frac{|K|}{|K^{r_T}|} \dim(R).
\end{equation}

As a special case, the simple particle $(T,R)$, where $T = K1K = K$ is the double coset of the identity element and $R$ is the trivial representation, is the vacuum particle (i.e. absence of excitation). The vacuum particle always has a quantum dimension of 1.

As in the case of the bulk, one can also fuse boundary excitations by bringing two excitations to the same boundary cilium via boundary ribbon operators, and consider the local operators $Z^{(hK,k)}$ acting on the new composite excitation. One can then show that the excitations on the boundary have a \lq\lq topological order" given by a unitary fusion category, as we will discuss in Section \ref{sec:algebraic}.  This new kind of boundary topological order exists only in the presence of the bulk topological order, and we will refer to it as \lq\lq bordered topological order".  In particular, the fusion category is the representation category of the group-theoretical quasi-Hopf algebra $\mathcal{Z}$ introduced in Section \ref{sec:bd-local-operators} (or equivalently, the representation category of the coquasi-Hopf algebra $\mathcal{Y}$). In fact, this category is Morita equivalent to the representation category $\Rep(G)$; its Drinfeld center is indeed equivalent to $\mZ(\Rep(G))$.

\subsubsection{Products of bulk-to-boundary condensation}

In Section \ref{sec:hamiltonian-gsd-condensation}, we informally described how a ground state degeneracy can result from the ability for certain bulk particles to condense to vacuum on the boundary. Now that we have formally defined and classified the elementary excitations of the boundary, we can provide a formal classification of these special bulk particles. More generally, given any elementary excitation $(C,\pi)$ in the bulk, we present a way to determine the products $(T,R)$ that are formed by condensation to the boundary. 

Suppose we have a boundary given by subgroup $K$, and a bulk anyon $a = (C,\pi)$ to condense to the boundary. In terms of ribbon (triangle) operators, if the border line between the bulk Hamiltonian $H_{(G,1)}$ and the boundary Hamiltonian $H^{(K,1)}_{(G,1)}$ lies on the direct lattice, the condensation procedure is always described by a dual triangle operator on a triangle such as the triangle $\tau$ in Fig. \ref{fig:condensation}. To bring $a$ to the boundary, we simply apply one of the operators $F^{(C,\pi);(\bf{u,v})}_\tau$. So far, this movement operator is the same as moving the anyon to anywhere else in the bulk.

The difference arises when $a$ crosses the boundary. Once this happens, $a$ may no longer be an elementary excitation: instead, it could be a superposition of the elementary excitations of the boundary that we classified earlier. Mathematically, this corresponds to the fact that the ribbon operators $F^{(C,\pi);(\bf{u,v})}_\tau$ no longer form a basis for the triangle operators in the boundary, so we must express them as a linear combination of the basis operators $Y^{(T,R);({\bf u,v})}_\tau$. This linear combination is constructed as follows:

Since $\tau$ is a dual triangle, by Equations (\ref{eq:triangle-operator-def}) and (\ref{eq:bd-triangle-operator-def}), we have (before the Fourier transform)
\begin{equation}
F^{(h,g)}_\tau = \delta_{1,g} L^h(e) \qquad Y^{(hK,k)}_\tau = \delta_{1,k} L^{hK}(e)
\end{equation}

By Equations (\ref{eq:bulk-ribbon-FT}) and (\ref{eq:bd-ribbon-FT}), we have (after the Fourier transform)
\begin{align}
\begin{split}
F^{(C,\pi);(\bf{u,v})}_\tau & := \frac{\dim(\pi)}{|E(C)|}
\sum_{k \in E(C)} \left(\Gamma_\pi^{-1}(k)\right)_{jj'}\delta_{1,p_i k p_{i'}^{-1}} L^{c_i^{-1}}(e)
\\ &= \frac{\dim(\pi)}{|E(C)|}\left(\Gamma_\pi^{-1}(p_i^{-1} p_{i'})\right)_{jj'}L^{c_i^{-1}}(e).
\end{split}
\end{align}
\begin{align}
\begin{split}
Y^{(T,R);(\bf{u,v})}_\tau & := \frac{\dim(R)}{|K^{r_T}|}
\sum_{k \in K^{r_T}} \left(\Gamma_R^{-1}(k)\right)_{jj'}\delta_{1,q_i k q_{i'}^{-1}} L^{s_i^{-1}K}(e)
\\ &= \frac{\dim(R)}{|K^{r_T}|}\left(\Gamma_R^{-1}(q_i^{-1} q_{i'})\right)_{jj'}L^{s_i^{-1}K}(e).
\end{split}
\end{align}
(In both cases, it is possible that $p_i^{-1} p_{i'} \notin C$ or $q_i^{-1} q_{i'} \notin K^{r_T}$; if that happens, we simply have $F^{(C,\pi);(\bf{u,v})}_\tau = 0$ or $Y^{(T,R);(\bf{u,v})}_\tau = 0$).

The following theorem hence governs the products of condensation:

\begin{theorem}
\label{condensation-products}
Let $(T,R)$ and $(C,\pi)$ be given elementary excitations of the boundary and bulk, respectively. The term $Y^{(T,R);(\bf{u_2,v_2})}_\tau$ has a nonzero coefficient in the decomposition of $F^{(C,\pi);(\bf{u_1,v_1})}_\tau$ (for some quadruple $(\bf{u_1,v_1,u_2,v_2})$) if and only if the following two conditions hold:
\begin{enumerate}
\item
The intersection $C \cap T$ is nonempty. When this condition holds, we assume the double coset representative $r_T$ is also in $C$.
\item
There exists an $x \in G$ such that the following is true: Let $x\triangleright \pi$ denote the representation of $x E(C) x^{-1}$ obtained from $\pi$ where $y$ acts as $x^{-1} y x$. Let $\rho_{x \triangleright \pi}$ be the (possibly reducible) representation of the subgroup $(xE(C)x^{-1}) \cap K^{r_T}$ resulting from the restriction of ${x \triangleright \pi}$ to $(xE(C)x^{-1}) \cap K^{r_T}$; let $\rho_R$ be the representation of the same subgroup formed by restricting $R$. Decompose $\rho_{x \triangleright \pi}$, $\rho_R$ into irreducible representations of $(xE(C)x^{-1}) \cap K^{r_T}$:
\begin{equation}
\rho_{x \triangleright \pi} = \oplus_\sigma n^{{x \triangleright \pi}}_{\sigma} \sigma
\end{equation}
\begin{equation}
\rho_R = \oplus_\sigma n^{R}_{\sigma} \sigma
\end{equation}
\noindent
There must exist some irreducible representation $\sigma$ of $(xE(C)x^{-1}) \cap K^{r_T}$ such that $n^{x \triangleright \pi}_{\sigma} \neq 0$ and $n^{R}_{\sigma} \neq 0$.
\end{enumerate}
In particular, let $X(C)$ be a set of representatives of the double cosets $K\backslash G/E(C)$. For given $(C,\pi)$ let us write the decomposition after condensation as
\begin{equation}
\label{eq:condensation}
(C,\pi) = \oplus n_{(T,R)}^{(C,\pi)} (T,R).
\end{equation}
Then, we have 
\begin{equation}
\label{eq:condensation-coefficients}
n_{(T,R)}^{(C,\pi)} = \sum_{\substack{x \in X(C) \text{ s.t. } xr_cx^{-1}\in T \\ \sigma \in ((x E(C) x^{-1}) \cap K^{r_T})_{\text{ir}}}} n^R_\sigma n^{{x \triangleright \pi}}_\sigma
\end{equation}
Furthermore, these coefficients imply that the two sides of Eq. (\ref{eq:condensation}) will always have the same quantum dimensions.
\end{theorem}

Similarly, we may also consider the process of pulling a boundary excitation back into the bulk. The situation here is exactly the inverse of the above: we wish to write the $Y^{(T,R);(\bf{u,v})}_\tau$ as a linear combination of $F^{(C,\pi);(\bf{u,v})}_\tau$. Hence, by the same reasoning as above, we have the following theorem:

\begin{theorem}
\label{inverse-condensation-products}
Let $(T,R)$ and $(C,\pi)$ be given elementary excitations of the boundary and bulk, respectively. The term $F^{(C,\pi);(\bf{u_1,v_1})}_\tau$ has a nonzero coefficient in the decomposition of $Y^{(T,R);(\bf{u_2,v_2})}_\tau$ (for some quadruple $(\bf{u_1,v_1,u_2,v_2})$) if and only if the conditions (1) and (2) of Theorem \ref{condensation-products} hold.

In particular, let us write the decomposition of the simple boundary excitation as
\begin{equation}
\label{eq:inverse-condensation}
(T,R) = \oplus n_{(C,\pi)}^{(T,R)} (C,\pi).
\end{equation}
Then, we have 
\begin{equation}
n_{(C,\pi)}^{(T,R)} = n^{(C,\pi)}_{(T,R)},
\end{equation}
where $n^{(C,\pi)}_{(T,R)}$ is defined as in Theorem \ref{condensation-products}. Furthermore, these coefficients imply that the quantum dimension of the right hand side of Eq. (\ref{eq:inverse-condensation}) will always be $|G|$ times that of the left hand side.
\end{theorem}

Mathematically, Theorem \ref{condensation-products} represents the decomposition of irreducible representations of $D(G)$ under the restriction homomorphism $D(G) \twoheadrightarrow Z(G,1,K,1)$; Theorem \ref{inverse-condensation-products} represents the decomposition of irreducible representations of $Z(G,1,K,1)$ under the induction homomorphism $Z(G,1,K,1) \hookrightarrow D(G)$. The correspondence between the two theorems is essentially Frobenius reciprocity.

Theorem \ref{inverse-condensation-products} gives us a straightforward way to determine which quasi-particles $(C,\pi)$ may be condensed to vacuum on a given boundary based on subgroup $K$: we can simply find all quasi-particles appearing with nonzero coefficient in the decomposition (\ref{eq:inverse-condensation}) with $(T,R)$ trivial. In general, we will use these anyon types (and their corresponding condensation coefficients $n^{(\{1\},1)}_{(C,\pi)}$) to label the corresponding gapped boundary.

We would like to note that the above two theorems are exactly consistent with the mathematical results presented by Schauenburg for group-theoretical categories in Ref. \cite{Schauenburg15}.  In particular, they are derived using formulas in Section 3 of  Ref. \cite{Schauenburg15}.

\begin{figure}
\centering
\includegraphics[width = 0.47\textwidth]{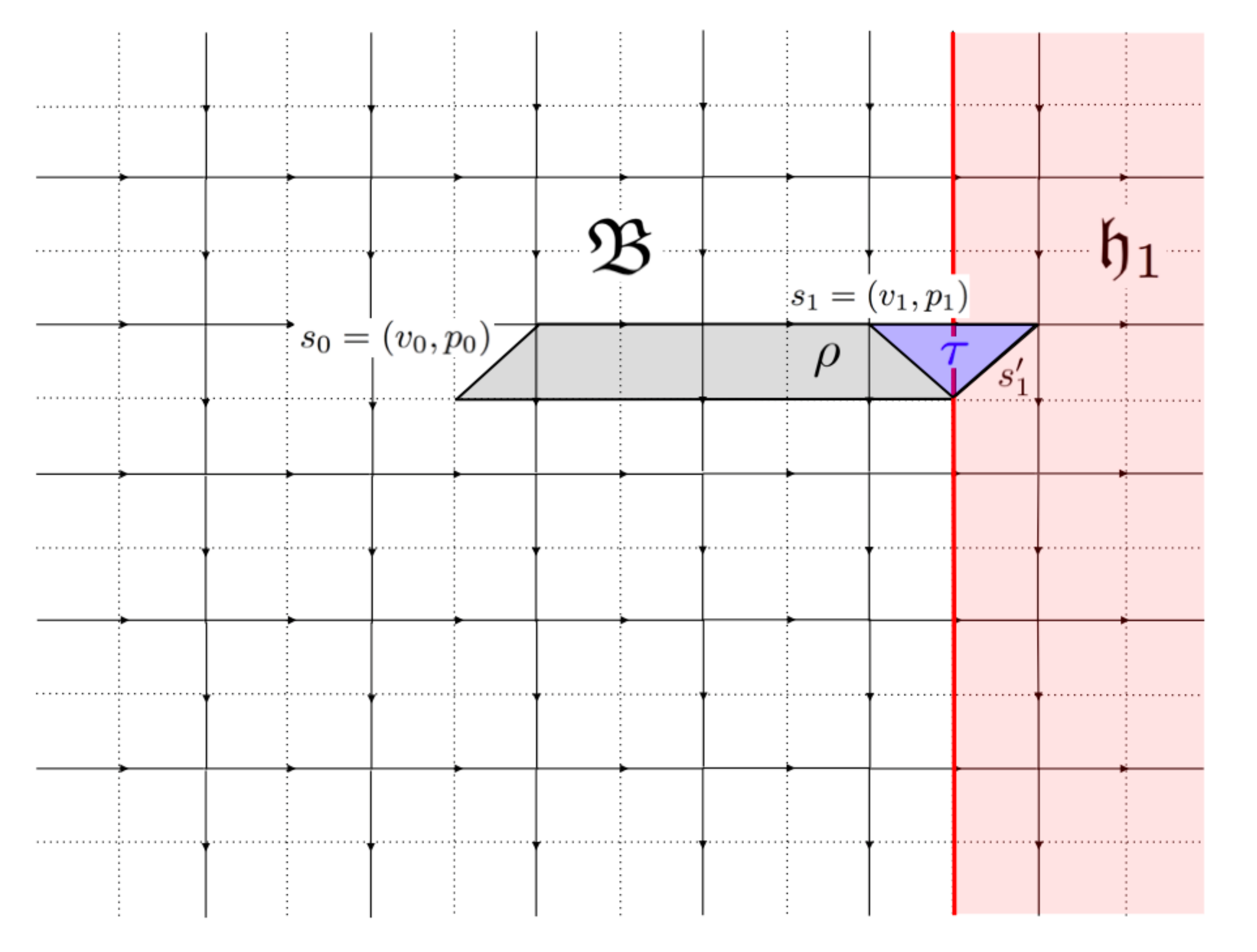}
\caption{Illustration of the bulk-to-boundary condensation procedure, when the boundary lies on the dual lattice. In this case, the condensation triangle is direct.}
\label{fig:condensation-2}
\end{figure}

\begin{remark}
\label{dual-bd-rmk-2}
In Remark \ref{dual-bd-rmk}, we noted that it is also possible to create a boundary line on the dual lattice. In this case, the ``condensation triangle'' $\tau$ of Fig. \ref{fig:condensation} is now a direct triangle instead of a dual triangle (see Fig. \ref{fig:condensation-2}). For this case, there are analogous results to Theorems \ref{condensation-products} and \ref{inverse-condensation-products}, which are obtained using ribbon operators on the direct triangle. In general, these two methods of creating boundaries with the same subgroup $K$ can result in different boundary types (i.e. different condensation formulas as in Equations (\ref{eq:condensation}) and (\ref{eq:inverse-condensation})). However, it is straightforward to show that the boundary type corresponding to the dual lattice boundary may also be created using a different subgroup on a direct lattice boundary.
\end{remark}

\subsubsection{Multiple condensation channels}
\label{sec:multiple-condensation-channels}

As seen in Theorems \ref{condensation-products} and \ref{inverse-condensation-products}, it is possible (e.g. in the case of $G = S_3$ which we study in Section \ref{sec:ds3-hamiltonian-example}) that one can have condensation multiplicities $n^{(C,\pi)}_{(T,R)}$ greater than 1. In this case, we say that there are multiple {\it condensation channels}. Physically, this is very similar to multiple fusion channels in the bulk (such as in the UMTC given by $\text{SU}(3)_3$), where we have fusion rule coefficients greater than 1. 

The origin of the condensation multiplicity can be traced back to the local degrees of freedom in the definition of the (bulk) ribbon operators. Recall that a ribbon operator $F^{(C, \pi);(\vec{u},\vec{v})}$ has local degrees of freedom indexed by $\vec{u}$ (or $\vec{v}$) at the two ends of the ribbon, resulting in quantum dimension $|C|\text{dim}(\pi)$. Application of local operators $D^{(h,g)}$ at the ends can mix the local states completely. However, if we move one end of the ribbon to a gapped boundary, to distinguish the local degrees of freedom, we can now only apply operators that commute with the boundary Hamiltonian. Therefore, on the boundary, we may not be able to distinguish all the local degrees of freedom of the ribbon operator completely without creating additional excitations; the remaining degeneracy becomes the condensation multiplicity.

A common situation where such multiplicity arises in quantum double models is when $K=\{1\}$. In this case, it is easy to see that all gauge charges condense to vacuum on the boundary, and the multiplicity is given by $n^{(\{1\}, \pi)}_{(\{1\},1)}=\mathrm{dim}(\pi)$. Let us understand the multiplicity in this example more concretely in terms of ribbon operators. Recall that for a gauge charge corresponding to an irreducible representation $\pi$ of $G$, the ribbon operators read
\begin{equation}
	F_\rho^{(\{1\},\pi);(\bf{u,v})}=\frac{\mathrm{dim}(\pi)}{|G|}\sum_{g\in G} \big(\Gamma_\pi^{-1}(g)\big)_{jj'}F^{(1, g)}_\rho.
	\label{}
\end{equation}
Suppose we now use the above ribbon operator to create a pair of charges in the bulk, and then move one of them, say the end corresponding to the $\vec{u}$ index, to the boundary.
In the bulk, one can easily show that applying $A^h$ at the end of the ribbon mixes the local indices. However, on the boundary, the $A^h$ operators do not commute with the $T^K$ boundary terms, unless one applies the product of all such $A$'s along the entire boundary. Therefore, different indices $\vec{u}$ are now locally indistinguishable, which explains the origin of the multiplicity.

From this example, we also see that condensation channels are topologically protected. In order to change from one channel to another without leaving any trace of excitation, one must apply a boundary ribbon operator that completely encircles the boundary. Since boundary particles are confined, such an operator would require energy input proportional to the perimeter of the boundary.

\subsection{Example: $\mfD(S_3)$}
\label{sec:ds3-hamiltonian-example}

In this section, we present an example using the group $G = S_3 = \{r,s|r^3 = s^2 = srsr = 1 \}$, the permutation group on three elements, to illustrate our theory on the simplest non-abelian group. Since this group is already quite complicated, we will not explicitly write out the Hamiltonian in full, although the interested reader can easily obtain it from Eq. (\ref{eq:kitaev-hamiltonian}). Instead, we focus our attention on the elementary excitations and gapped boundaries of this example.

\subsubsection{Elementary excitations}

To determine the elementary excitations of this model, we again only need to find the pairs $(C,\pi)$, as described in Eq. (\ref{eq:elementary-ribbon-basis}). The conjugacy classes $C$ of $S_3$, the corresponding centralizers $E(C)$ and their irreducible representations are:

\begin{enumerate}
\item
$C_0 = \{1\}: E(C_0) = S_3$. Three irreducible representations: trivial ($A$), sign ($B$), and the two-dimensional one ($C$).
\item
$C_1 = \{s, sr, sr^2\}: E(C_1) = \{1,s\} = \Z_2$. Two irreducible representations: trivial ($D$), sign ($E$).
\item
$C_2 = \{r, r^2\}: E(C_2) = \{1,r,r^2\} = \Z_3$. Three irreducible representations: trivial ($F$), $\{1, \omega, \omega^2\}$\footnote{Here $\omega = e^{2\pi i/3}$ is the third root of unity. $\{1, \omega, \omega^2\}$ means the representation where $1 \rightarrow 1$, $r \rightarrow \omega$, $r^2 \rightarrow \omega^2$.} ($G$),
$\{1, \omega^2, \omega^4\}$ ($H$).
\end{enumerate}

Hence, there are 8 anyon types for this model, namely $A-H$ as listed above.

\subsubsection{Gapped boundaries}

The group $G = S_3$ has 4 distinct subgroups up to conjugation, namely the trivial subgroup, $\Z_2$, $\Z_3$, and $G$ itself. In what follows, we solve for the 4 gapped boundaries corresponding to these subgroups. We will follow the method of Section \ref{sec:bd-excitations} to determine the excitations on the boundary, and the bulk anyons that can condense to vacuum.

\vspace{2.5mm}
\noindent
\underline{{\it Case I}}: $K = \{1\}$.
\vspace{2mm}

Since $K$ is trivial, there are 6 distinct double cosets, corresponding to each element $r_T \in G$. In each case, the only representation of $K^{r_T}$ is the trivial one. There are hence 6 elementary excitations on the boundary; let us label each excitation by the corresponding choice of $r_T$.

By Theorem \ref{inverse-condensation-products}, the bulk elementary excitations that condense to the trivial excitation on the boundary are precisely the particles corresponding to the trivial conjugacy class $C = \{1\}$, i.e. the chargeons. In general, for any finite group $G$, a simple argument shows that $K = \{1\}$ will always form the charge condensate boundary. More specifically, the ``boundary topological order'' corresponding to this boundary will always be described by the fusion category $\C[G]$.

More generally, we can use Theorem \ref{condensation-products} to determine the result of condensing each simple bulk anyon to the boundary:

\begin{multicols}{2}
\begin{enumerate}[label=(\roman*),leftmargin=0.5in]
\item
$A,B \rightarrow 1$
\item
$C \rightarrow 2 \cdot 1$
\item
$D,E \rightarrow s \oplus sr \oplus sr^2$
\item
$F,G,H \rightarrow r \oplus r^2$
\end{enumerate}
\end{multicols}

Similarly, by Theorem \ref{inverse-condensation-products}, we can determine the bulk anyons that result from pulling an elementary excitation out of the boundary:

\vspace{2.5mm}
\begin{enumerate}[label=(\roman*),leftmargin=0.5in]
\item
$1 \rightarrow A \oplus B \oplus 2C$
\item
$s, sr, sr^2 \rightarrow D \oplus E$
\item
$r, r^2 \rightarrow F \oplus G \oplus H$
\end{enumerate}
\vspace{2.5mm}

To indicate the anyon types that can condense to vacuum, we say that this subgroup forms an $A+B+2C$ boundary.

\vspace{2.5mm}
\noindent
\underline{{\it Case II}}: $K = \Z_2 = \{1,s\}$.
\vspace{2mm}

In this case, we see that there are only 2 double cosets, which give 3 elementary boundary excitations:

\vspace{2.5mm}
\begin{enumerate}[label=(\roman*),leftmargin=0.5in]
\item
$r_{T_1} = 1: T_1 = \{1,s\} = K^{r_{T_1}}$. 2 irreducible representations of $K^{r_{T_1}}$: the trivial one ($A$), the sign one ($B$).
\item
$r_{T_2} = r: T_2 = \{r,r^2,sr,sr^2\}, K^{r_{T_2}} = \{1\}$. There is only one trivial representation ($C$) of $K^{r_{T_2}}$.
\end{enumerate}
\vspace{2.5mm}

Using ribbon operator techniques, it is possible to show that this is in fact a boundary topological order given by the fusion category $\Rep(S_3)$.

We apply Theorem \ref{condensation-products} to determine the result of condensing each simple bulk anyon to the boundary:

\begin{multicols}{2}
\begin{enumerate}[label=(\roman*),leftmargin=0.5in]
\item
$A \rightarrow {A}$
\item
$B \rightarrow {B}$
\item
$C \rightarrow {A} \oplus {B}$
\item
$D \rightarrow {A} \oplus {C}$
\item
$E \rightarrow {B} \oplus {C}$
\item
$F,G,H \rightarrow {C}$
\end{enumerate}
\end{multicols}

Similarly, by Theorem \ref{inverse-condensation-products}, we have

\vspace{2.5mm}
\begin{enumerate}[label=(\roman*),leftmargin=0.5in]
\item
${A} \rightarrow A \oplus C \oplus D$
\item
${B} \rightarrow B \oplus C \oplus E$
\item
${C} \rightarrow D \oplus E \oplus F \oplus G \oplus H$
\end{enumerate}
\vspace{2.5mm}

Hence, $K = \Z_2$ corresponds to the $A+C+D$ boundary. We note that for this case, it does not matter which of the three $\Z_2$ subgroups we choose, since they are equivalent up to conjugation; in the end, they all give the same boundary condensation rules.

\vspace{2.5mm}
\noindent
\underline{{\it Case III}}: $K = \Z_3 = \{1,r,r^2\}$.
\vspace{2mm}

It is simple to show that this case is the same as $A+B+2C$ with $C,F$ switched.

\vspace{2.5mm}
\noindent
\underline{{\it Case IV}}: $K = G = S_3$.
\vspace{2mm}

It is simple to show that this case is the same as $A+C+D$ with $C,F$ switched. In general, the subgroup $K = G$ with the trivial cocycle always yields a pure flux condensate.

\vspace{2mm}
\section{Algebraic model of gapped boundaries}
\label{sec:algebraic}

In this section, we present a mathematical model of gapped boundaries using Lagrangian algebras in a modular tensor category. Throughout the section, we will assume the reader is familiar with the concepts of a fusion category and a modular tensor category; for reference on these topics, see Ref. \cite{BakalovKirillov,Etingof05}.

\subsection{Topological order}
\label{sec:topological-order}

In this section, we briefly review the mathematical theory that describes elementary excitations and anyons.

As we saw in Section \ref{sec:ribbon-operators}, the topological charges/anyon types in the Kitaev model with group $G$ are given by irreducible representations of the quantum double $D(G)$, which are pairs $(C,\pi)$ of a conjugacy class of $G$ and an irreducible representation of the centralizer of $C$. Equivalently, these are the simple objects of the Drinfeld center $\mfD(G) = \mZ(\text{Vec}_G) = \Rep(D(G))$.

If we think dually, we can also view the Kitaev model in terms of the representation category of $G$. In this case, every edge of the lattice will be labeled by an object in the unitary fusion category $\mC = \Rep(G)$, the complex linear representations of the group $G$. The elementary excitations in this model will be given by simple objects in the modular tensor category $\B = \mfD(G) = \mZ(\Rep(G))$, the Drinfeld center of the representation category. Because $\text{Vec}_G$ is Morita equivalent to $\Rep(G)$, these simple objects are given precisely by the same pairs $(C,\pi)$. In fact, this dualization exactly gives the same topological order as the Kitaev model, using the Levin-Wen Hamiltonian \cite{Levin04}.

By using ribbon operator techniques as presented in Section \ref{sec:hamiltonian}, one can in principle compute the twists and braidings of all of the elementary excitations in the Kitaev model. In doing so, one can determine all of the $\mathcal{S},\mathcal{T}$ matrix entries for this anyon system. It is conjectured that these two matrices uniquely determine a modular tensor category. If this conjecture holds, using such an analysis, one can show with ribbon operators that the topological order of the Kitaev model is indeed described by the modular tensor category $\B = \mZ(\Rep(G))$.

In fact, it is widely believed that modular tensor categories can be used to describe not only the topological order of Kitaev models, but also of Levin-Wen models \cite{Levin04}. These models also use a lattice (similar to Fig. \ref{fig:kitaev}), with the modification that the lattice should be trivalent (e.g. the honeycomb lattice). Here, the label on each edge is given by a simple object in a unitary fusion category $\mC$. As shown in Ref. \cite{Levin04}, string operators may also be defined for this model, although it is not as simple to characterize the elementary excitations and anyon fusion. The topological order would be given by the Drinfeld center $\B = \mZ(\mC)$, although one must also at least compute the $\mathcal{S},\mathcal{T}$ matrices using string operators to verify this for each particular case.

In the rest of this section, we present an algebraic theory for gapped boundaries for any model whose topological order is given by a doubled theory $\B = \mZ(\mC)$ for some unitary fusion category $\mC$.

\subsection{Lagrangian algebras}
\label{sec:frobenius-algebras}

We will now describe the gapped boundaries in a theory with topological order given by $\B = \mZ(\mC)$. Let us first state a few definitions and theorems that will be crucial for the rest of the paper.

\begin{theorem}
\label{indecomposable-module-repG}
Let $G$ be some finite group. There exists a one-to-one correspondence between the indecomposable module categories of $\Rep(G)$ (defined in Ref. \cite{Ostrik03}) and the pairs $(\{K\},\omega)$, where $\{K\}$ is an equivalence class of subgroups $K \subseteq G$ up to conjugation, and $\omega \in H^2(K,\C^\times)$ is a 2-cocycle of a representative $K \in \{K\}$.
\end{theorem}

\begin{proof}
See Ref. \cite{Ostrik03}, Theorem 2.
\end{proof}

\begin{definition}
\label{lagrangian-algebra-def}
A {\it Lagrangian algebra} $\A$ in a modular tensor category $\B$ is an algebra with multiplication $m: \A \otimes \A \rightarrow \A$ such that:
\begin{enumerate}
\item
$\A$ is {\it commutative}, i.e. $\A \otimes \A \xrightarrow{c_{\A\A}} \A \otimes \A \xrightarrow{m} \A$ equals $\A \otimes \A \xrightarrow{m} \A$, where $c_{\A\A}$ is the braiding in the modular category $\B$.
\item
$\A$ is {\it separable}, i.e. the multiplication morphism $m$ admits a splitting $\mu:\A \rightarrow \A \otimes \A$ which is a morphism of $(\A,\A)$-bimodules.
\item
$\A$ is {\it connected}, i.e. $\Hom_\B(\one_\B, \A) = \C$, where $\one_\B$ is the tensor unit of $\B$.
\item
The Frobenius-Perron dimension (a.k.a. quantum dimension) of $\A$ is the square root of that of the modular tensor category $\B$,
\begin{equation}
\label{eq:lagrangian-algebra-dim}
\FPdim(\A)^2 = \FPdim(\B).
\end{equation}
\end{enumerate}
\end{definition}

\begin{remark}
We note that an algebra satisfying conditions (2) and (3) in the above definition is often known in the literature as an {\it \'etale} algebra.
\end{remark}

As discussed in Section \ref{sec:bd-hamiltonian}, in a Kitaev model for the untwisted Dijkgraaf-Witten theory based on group $G$, every subgroup $K \subseteq G$ (up to conjugation) with a cocycle $\omega \in H^2(K,\C^\times)$ determines a distinct gapped boundary of the model (i.e. a unique equivalence class of boundary Hamiltonians). It follows from Theorem \ref{indecomposable-module-repG} that there is an injection from the indecomposable modules of the category $\mZ(\Rep(G))$ to the set of gapped boundaries of the Kitaev model. Furthermore, because we choose orientations so that the bulk is always on the left hand side when we traverse each boundary, we view gapped boundaries as indecomposable left module categories.

By Proposition 4.8 of Ref. \cite{Davydov12}, we may state the following theorem:

\begin{theorem}
\label{indecomposable-module-lagrangian-algebra}
Let $\mC$ be any fusion category, and let $\B = \mZ(\mC)$. There exists a one-to-one correspondence between the indecomposable modules of $\mC$ and the Lagrangian algebras of $\B$.
\end{theorem}

As a result, gapped boundaries of the Kitaev model may be determined by enumerating the Lagrangian algebras in $\mZ(\Rep(G))$.

In fact, it is proposed \cite{KitaevKong} that in any Levin-Wen model based on unitary fusion category $\mC$, the gapped boundaries are in one-to-one correspondence with the indecomposable modules $\M$ of $\mC$. In this case, by determining all the Lagrangian algebras of the Drinfeld center $\B = \mZ(\mC)$, we can also obtain gapped boundaries of the Levin-Wen model. In what follows, the theory we develop will be applicable to any model where gapped boundaries are given by indecomposable modules of the input fusion category.

\begin{remark}
The above definition of a Lagrangian algebra is the same as a special, symmetric Frobenius algebra, with an additional restriction on the quantum dimension of the algebra. This condition enforces that $\A$ has the maximal quantum dimension possible. Physically, this makes $\A$ into a gapped boundary (or equivalently, a domain wall between the $\B$ and the trivial category $\text{Vec}$), as we have discussed above. A special, symmetric Frobenius algebra with smaller quantum dimension would correspond to a domain wall between $\B$ and another topological phase.

We note that while this section deals purely with Lagrangian algebras and gapped boundaries, our work generalizes to the case of domain walls. In fact, a domain wall is mathematically equivalent to a gapped boundary, by using the ``folding'' technique. This is discussed in Refs. \cite{Beigi11} and \cite{KitaevKong} and justified rigorously in \cite{Fuchs2013}.  
\end{remark}

To find all Lagrangian algebras of a modular tensor category, we will first state the following propositions:

\begin{prop}
\label{bosons}
$\A$ is a commutative algebra in a modular category $\B$ if and only if the object $\A$ decomposes into simple objects as $\A = \oplus_s n_s s$, with $\theta_s = 1$ (i.e. $s$ is bosonic) for all $s$ such that $n_s \neq 0$. 
\end{prop}

\begin{proof}
See Proposition 2.25 in Ref. \cite{Frohlich06}.
\end{proof}

\begin{prop}
\label{separability-prop}
$\A$ is a separable algebra in a unitary fusion category $\B$ if and only if for every $a,b \in \Obj(\B)$, there exists a partial isometry from $\Hom(a,\A) \otimes \Hom(b,\A) \rightarrow \Hom(a \otimes b, \A)$.\footnote{$\B$ is a fusion category, so all hom-spaces in $\B$ have vector space structure. The tensor product of hom-spaces is just the usual tensor product for vector spaces.}
\end{prop}

Our proof below is given using the graphical calculus of fusion categories.  Each segment of a diagram is labeled by some object in the category and vertices are labeled by some chosen morphisms.  In our diagrams below, segments are labeled by either the algebra $\A$ or some objects such as $a, b$, and the trivalent vertices are labeled either by the multiplication or comultiplication of $\A$.   Each diagram represents a morphism in a certain hom-space, which is a composition of elementary morphisms from the bottom to the top.  A closed diagram represents a morphism in $\Hom(1,1)$ between the tensor unit $1$.  Hence a closed diagram, which is always equal to $\lambda \cdot \textrm{id}_1$ for some scalar $\lambda$, represents the number $\lambda$.

\begin{proof}
Fix $a,b \in \Obj(\B)$. Define a map $M$ from $\Hom(a,\A) \otimes \Hom(b,\A)$ to $\Hom(a \otimes b, \A)$ as follows:

By definition of a tensor category, there exists an injective map $\gamma: \Hom(a,\A) \otimes \Hom(b,\A) \rightarrow \Hom(a \otimes b, \A \otimes \A)$. Suppose we are given two morphisms $f \in \Hom(a,\A)$, $g \in \Hom(b,\A)$. Then $M(f \otimes g) = (m \circ \gamma)(f \otimes g)$ is a morphism in the hom-space $\Hom(a \otimes b, \A)$. We now show that this map $M$ is injective.

Suppose $M(f \otimes g) = 0$. Since $\B$ is a unitary fusion category, $M(f \otimes g) = 0$ if and only if the following trace is equal to 0: (note here that the pictures are read bottom-up)

\begin{equation}
\label{eq:M-trace}
\vcenter{\hbox{\includegraphics[width = 0.2\textwidth]{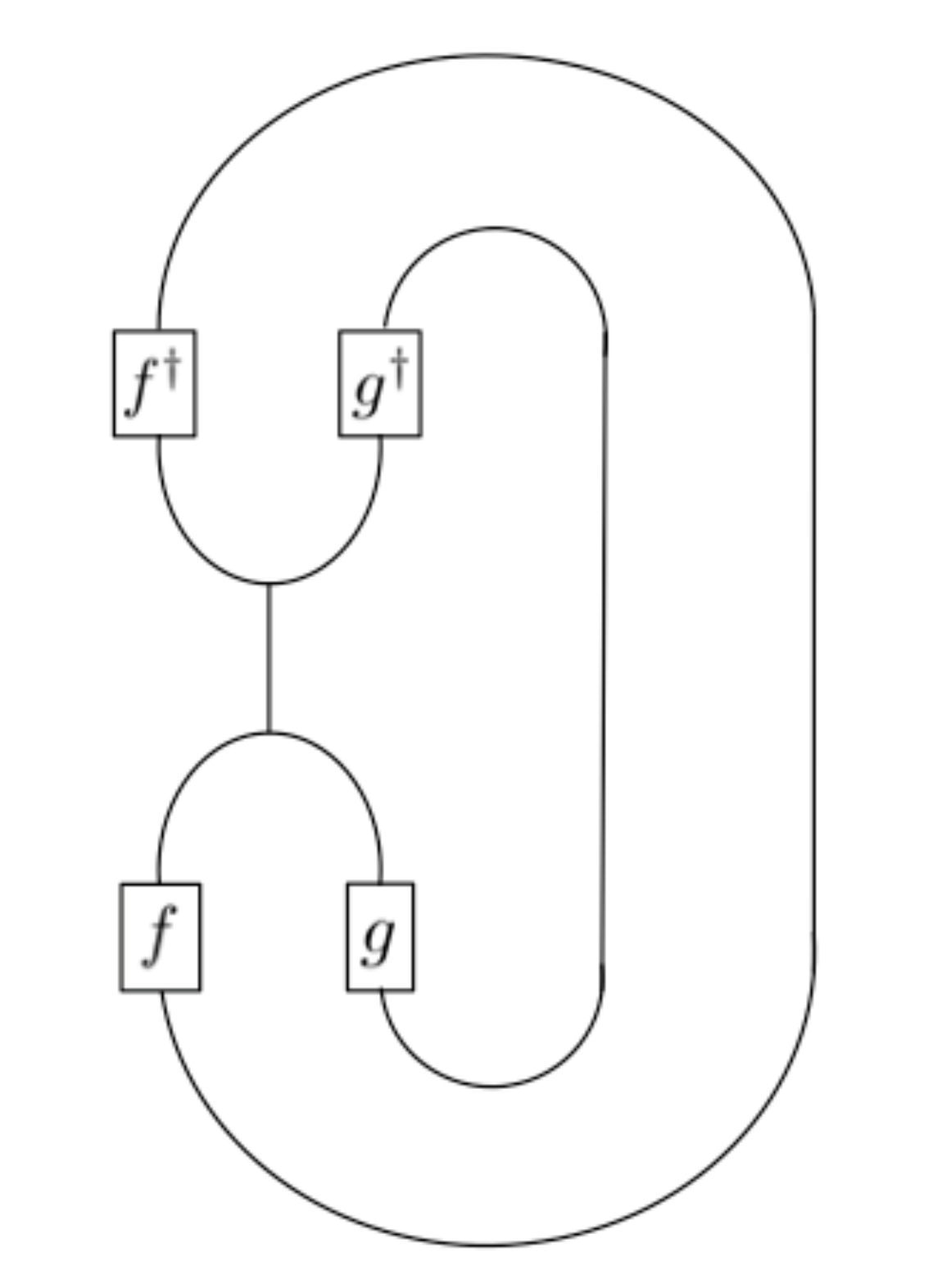}}} = 0.
\end{equation}

By definition of an $(\A,\A)$-bimodule, the condition (2) in Definition \ref{lagrangian-algebra-def} is equivalent to the following two conditions \cite{Muger12}:

\begin{equation}
\label{eq:separability}
\vcenter{\hbox{\includegraphics[width = 0.28\textwidth]{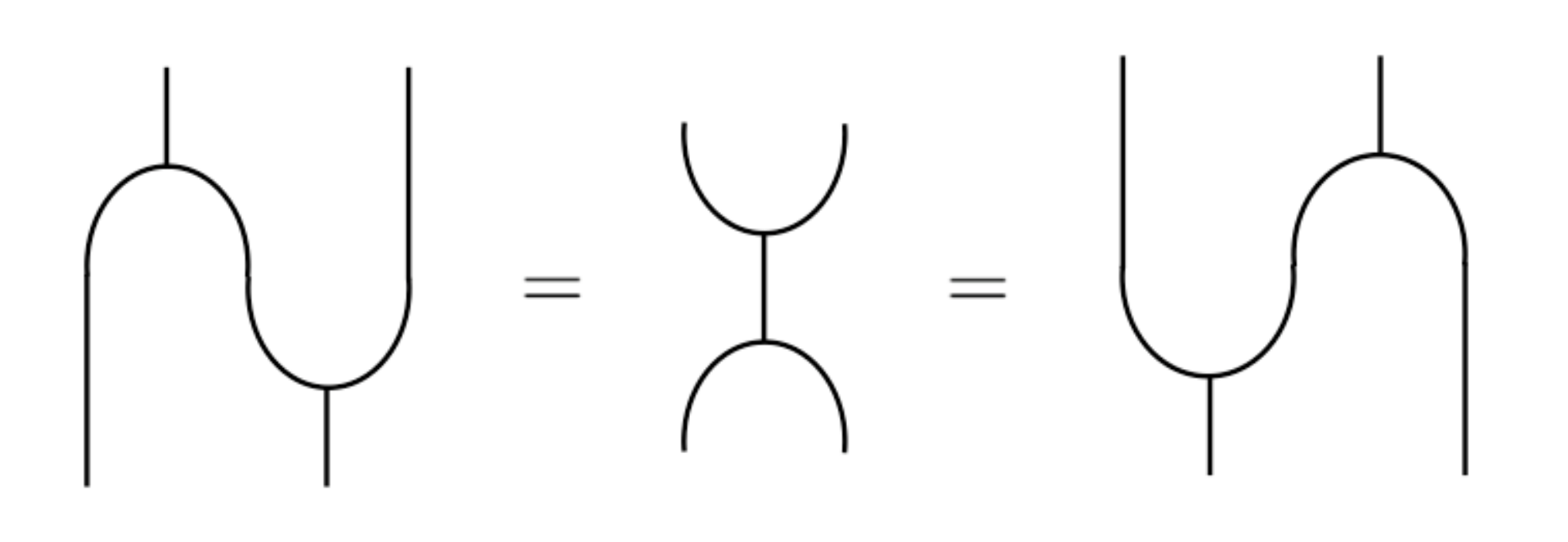}}}
\end{equation}

\noindent
and

\begin{equation}
\label{eq:separability-2}
\vcenter{\hbox{\includegraphics[width = 0.12\textwidth]{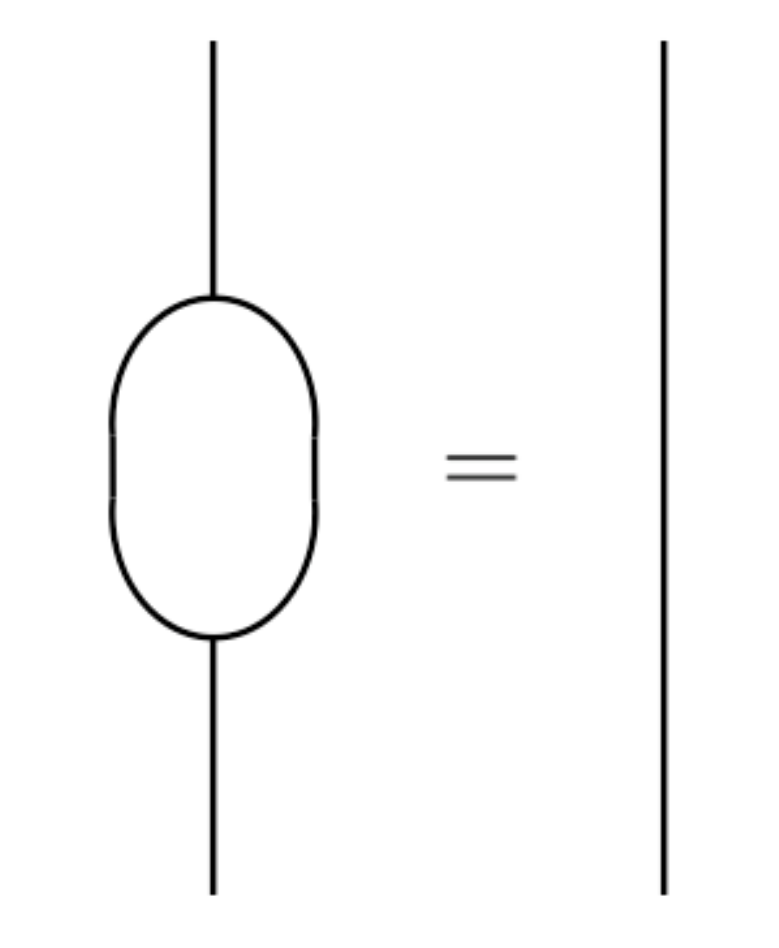}}}
\end{equation}

It follows that Eq. (\ref{eq:M-trace}) becomes a sum of the following diagrams, with $X$ being a simple object in $\A$: 

\begin{equation}
\label{eq:M-trace-2}
\vcenter{\hbox{\includegraphics[width = 0.18\textwidth]{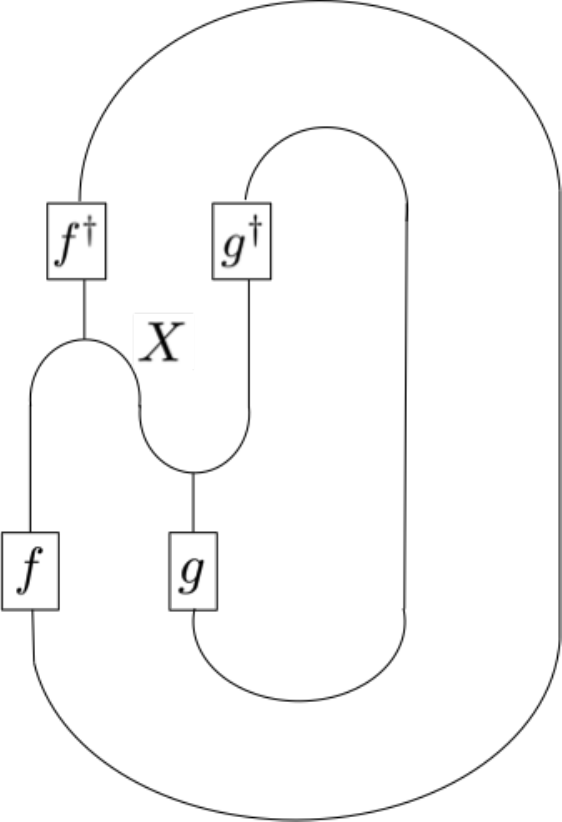}}} = 0,
\end{equation}

\noindent
By reading the above equation horizontally, we see that the part of the diagram consisting of the inner loop with the segment labeled by $X$ is in $\Hom(X,1)$ (potentially up to a non-zero scalar due to the rotation of the diagram). $\Hom(X,1)$ is zero unless $X$ is the tensor unit $1$ (we refer to this fact in graphical calculus as the ``no tadpole'' rule). Thus, Eq. (\ref{eq:M-trace}) holds if and only if Eq. (\ref{eq:M-trace-2}) does when $X = 1$ is the tensor unit.

We hence have the following picture:

\begin{equation}
\label{eq:M-trace-3}
\vcenter{\hbox{\includegraphics[width = 0.18\textwidth]{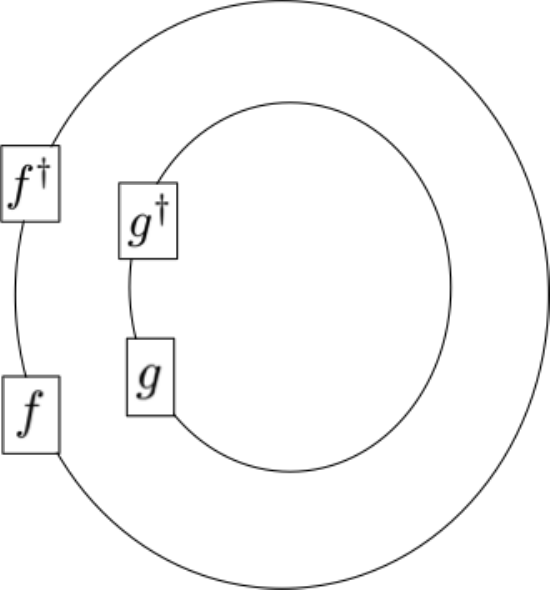}}} = 0.
\end{equation}

The left hand side of Eq. (\ref{eq:M-trace-3}) is precisely given by $\Tr(f)\Tr(g)$. Since $\B$ is a unitary fusion category, this equation holds if and only if $f = 0$ or $g = 0$, i.e. $f \otimes g = 0$.

Finally, $M$ is a partial isometry if and only if Eq. (\ref{eq:separability-2}) holds, which completes the forward direction of the proof.

Note that all steps in this proof were reversible, so that both directions of the Proposition hold.
\end{proof}

\begin{corollary}
A commutative connected algebra $\A = \oplus_s n_s s$ with $\FPdim(\A)^2 = \FPdim(\B)$ is a Lagrangian algebra in the unitary modular category $\B$ if and only if the following inequality holds for all $a,b \in \Obj(\B)$:

\begin{equation}
\label{eq:lagrangian-algebra-inequality}
n_a n_b \leq \sum_c N_{ab}^c n_c
\end{equation}

\noindent
where $N_{ab}^c$ are the coefficients given by the fusion rules of $\B$.
\end{corollary}

\begin{remark}
We would like to note that the algebra object $\A$ is not enough to uniquely identify the gapped boundary.  Let $G$ be the order$-64$ class $3$ group in Sec. IIIA of \cite{Davydov14}.  The standard Cardy Lagrangian algebra of $\mathcal{Z}(G\oplus G)$ has another different Lagrangian structure given by a soft braided auto-equivalence of $\mathcal{Z}(G)$.
\end{remark}

\subsection{Ground state degeneracy}
\label{sec:algebraic-gsd}

In Section \ref{sec:hamiltonian-gsd-condensation}, we used ribbon operators to present the ground state degeneracy of the Kitaev model with gapped boundaries. In this section, we will present this same degeneracy using the algebraic model we have developed in this section.

Theorems \ref{condensation-products} and \ref{inverse-condensation-products} of Section \ref{sec:bd-excitations} described how an anyon in the bulk can condense to the boundary. Given any anyon $a$ in the bulk, the condensation space of $a$ to a boundary given by the Lagrangian algebra $\A$ can be modeled precisely by the hom-space $\Hom(a,\A)$. Specifically, as discussed in Section \ref{sec:bd-excitations}, the number of condensation channels in condensing to vacuum is equivalent to the number of times the particle $a$ appears in the decomposition of $\A$ into simple objects (obtained using Theorem \ref{inverse-condensation-products} on the boundary vacuum particle); since $\B$ is a unitary fusion category, this is exactly the dimension of the hom-space. As in previous sections, this is expected to hold for theories with a topological order $\B = \mZ(\mC)$, not just $\B = \mZ(\Rep(G))$.

More generally, we also can describe the ground state of the TQFT with many boundaries using hom-spaces. Suppose we have a system with topological order given by the modular tensor category $\B$, and $n$ gapped boundaries given by Lagrangian algebras $\A_1, \A_2, ... \A_n$, as shown in Fig. \ref{fig:algebraic-gsd-n}. The outside boundary is taken to have total charge vacuum; one may alternatively view this as $n$ holes on a sphere.

\begin{figure}
\centering
\includegraphics[width = 0.3\textwidth]{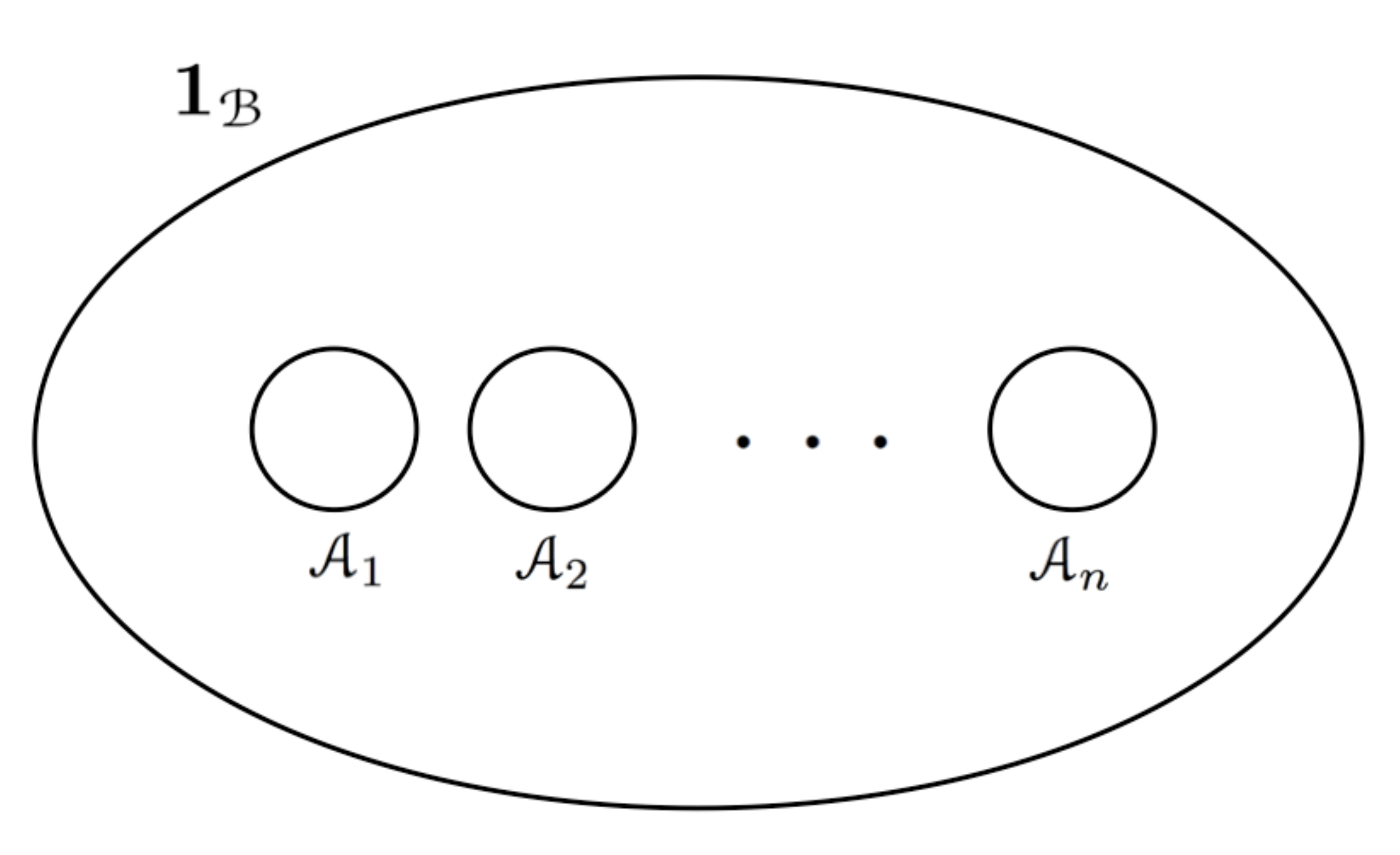}
\caption{Algebraic picture of the TQFT with $n$ gapped boundaries (holes). The outside charge is take to be vacuum, and each boundary type is given by a Lagrangian algebra $\A_i$.}
\label{fig:algebraic-gsd-n}
\end{figure}

The ground state degeneracy of this model is given by the number of ways we can create a pair of anyons from vacuum, admissibly split them into $n$ anyons, and condense all $n$ of them to vacuum onto the boundary, as shown in Fig. \ref{fig:algebraic-gsd-n-2}. Hence, the ground state of this model is given by

\begin{equation}
\label{eq:ground-state-algebraic}
\text{G.S.} = \Hom(\one_\B, \A_1 \otimes \A_2 \otimes ... \otimes \A_n)
\end{equation}

\noindent
where $1_\B$ is the tensor unit of $\B$ and represents the vacuum particle. In particular, as presented in Fig. \ref{fig:algebraic-gsd-n-2}, a basis for this ground state is $\{ (a_1, a_1', ... a_{n-1}', \mu_1, ... \mu_n)\}$. Here, the $a_i, a_i'$ are bulk anyon labels such that the splittings are admissible and the final products may all condense to vacuum on the respective boundaries, and the $\mu_i$ correspond to multiplicities in condensing to vacuum on the boundary (see for instance Theorem \ref{condensation-products} of Section \ref{sec:bd-excitations}, and Section \ref{sec:multiple-condensation-channels}). A straightforward generalization is used in cases with bulk fusion multiplicities.

\begin{figure}
\centering
\includegraphics[width = 0.4\textwidth]{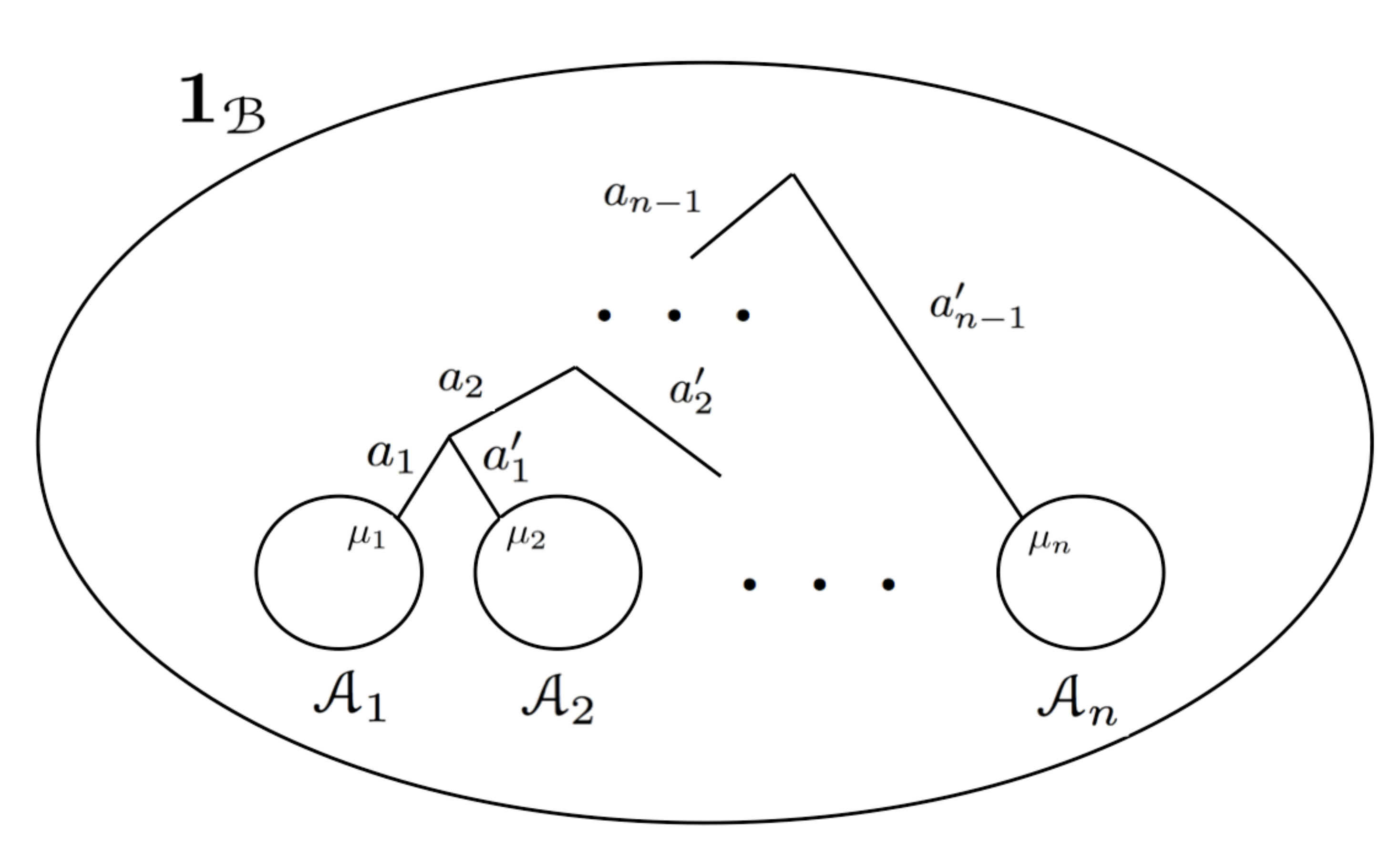}
\caption{Ground state of the model presented in Fig. \ref{fig:algebraic-gsd-n}. We assume here that all edges are directed to point downward. The $a_i$, $a_i'$ are simple bulk anyons such that the splittings are admissible, and the final products may all condense to vacuum on the respective boundaries. The $\mu_i$ correspond to multiplicities in condensing to vacuum on the boundary (see for instance Theorem \ref{condensation-products} of Section \ref{sec:bd-excitations}, and Section \ref{sec:multiple-condensation-channels}). The set of all such tuples $(a_1, a_1', ... a_{n-1}', \mu_1, ... \mu_n)$ form a basis for the ground state.}
\label{fig:algebraic-gsd-n-2}
\end{figure}

As discussed in Ref. \cite{Cong16a}, this ground state degeneracy may be used to encode a qudit for topological quantum computation.

\subsection{Condensation and elementary excitations on the boundary}
\label{sec:algebraic-condensation}

In Section \ref{sec:bd-excitations}, we saw that the elementary excitations on a boundary of subgroup $K$ in the Kitaev model based on group $G$ are given by the irreducible representations of the group-theoretical quasi-Hopf algebra corresponding to $G,K$. In that section, we presented a method using the Hamiltonian to obtain the products of the condensation procedure and its reverse. In this section, we will present this procedure categorically, and show that these two views are actually the same in the case of group models. Let us first make the following definition:

\begin{definition}
\label{quotient-cat-def}
Let $\B$ be a category, and let $\A$ be any object in $\B$. The {\it quotient pre-category} $\B/\A$ is the pre-category such that:
\begin{enumerate}
\item
The objects of $\B/\A$ are the same as the objects of $\B$.
\item
The morphisms of $\B/\A$ are given by
\begin{equation}
\Hom_{\B/\A}(X,Y) = \Hom_{\B}(X,\A \otimes Y).
\end{equation}
\end{enumerate}
\end{definition}

In our case, we would like to consider condensation of an anyon onto a gapped boundary. Here, the elementary excitations on the boundary seem to be given by the simple objects in the quotient pre-category $\widetilde{\mathcal{Q}} = \B / \A$. However, there are problems with the above quotient pre-category for our application: it may not be possible to compose morphisms, and $\widetilde{\mathcal{Q}}$ may not be semisimple. As a result, the following definition/proposition of Ref. \cite{Muger03} is needed to fully describe the condensation products.  Let us first recall that a Frobenius algebra $(\A,m,\Delta,\eta,\epsilon)$ with multiplication $m$, comultiplication $\Delta$, unit $\eta$, and counit $\epsilon$ in a strict tensor category is {\it strongly separable} if there exist nonzero complex numbers $\alpha, \beta$ such that $m \circ \Delta=\alpha \cdot \textrm{id}_\A $ and $\epsilon\circ \eta=\beta \cdot \textrm{id}_1$ (see Definition 2.5 of Ref. \cite{Muger03}).

\begin{definition}
\label{IC-def}
Let $\B$ be a strict braided tensor category and let $\A$ be a strongly separable Frobenius algebra in $\B$. Let $\widetilde{\mQ} = \B/\A$ be the quotient pre-category formed via Definition \ref{quotient-cat-def}. By Proposition 2.15 of Ref. \cite{Muger03}, $\widetilde{\mQ}$ is a tensor category under these conditions, although it may not be semisimple.

 Let us form the canonical idempotent completion ${\mQ}$ of $\widetilde{\mQ}$ as follows:
\begin{enumerate}
\item
The objects of ${\mQ}$ are given by pairs $(X,p)$, where $X \in \Obj \widetilde{\mQ}$ and $p = p^2 \in \End_{\widetilde{\mQ}}(X)$.
\item
The morphisms of ${\mQ}$ are given by
\begin{equation}
\Hom_{{\mQ}}((X,p),(Y,q)) = \{ f \in \Hom_{\widetilde{\mQ}}(X,Y): f \circ p = q \circ f\}.
\end{equation}
\end{enumerate}
Then, by Proposition 2.15 of Ref. \cite{Muger03}, the result ${\mQ}$ is a semisimple tensor category, and it is the desired quotient.

We denote by ${F}:\B \rightarrow \B/\A$ the monoidal functor that sends each object of $\B$ to the corresponding quotient object, and by $I: \B/\A \rightarrow \B$ its right adjoint.
\end{definition}

When $\B = \mZ(\mC)$ is a Drinfeld center describing a topological order, and $\A \in \B$ is a Lagrangian algebra/gapped boundary corresponding to an indecomposable module category $\M$ of $\mC$ (see Theorem \ref{indecomposable-module-lagrangian-algebra}), the resulting quotient is indeed the proposed category $\text{Fun}_\mC(\M,\M)$ of boundary excitations \cite{Muger03}. In general, condensation to a domain wall between two topological phases with topological orders $\B,\mD$ is mathematically described as the procedure

\begin{equation}
\label{eq:condensation-quotient-IC}
F: \mZ(\mC) = \B \xrightarrow{\text{quotient}} \B/\A = \widetilde{\mQ} \xrightarrow{\text{I.C.}}  {\mQ} = \msC \oplus \msD.
\end{equation}

\noindent
(Here, I.C. is the idempotent completion and $\oplus$ is the sum of abelian categories). After condensation to a domain wall, all excitations in $\msC$ become {\it confined} to the domain wall, and all excitations in $\msD$ are {\it deconfined} and can enter the phase $\mD$. Physically, an excitation is said to be confined if there is an energy cost to move the excitation growing linearly with the distance of movement.  Mathematically, $\msD$ is the abelian subcategory of $\mQ$ generated by all simple objects whose double braidings are trivial with respect to the condensate $\A$, and $\msC$ the one generated by simple objects whose double braidings are nontrivial.  In this paper, we consider only cases where $\msD$ is empty.

As mentioned in the introduction, the above condensation functor $F$ corresponds physically to the green triangle operator of Fig. \ref{fig:main-results}.

In the case of gapped boundaries, $\msD = \text{Vec}$ is vacuum, so there are no nontrivial deconfined excitations, and all excitations on the boundary are confined. In the case of Kitaev models for Dijkgraaf-Witten theories, the condensation procedure described in Eq. (\ref{eq:condensation-quotient-IC}) is then exactly the condensation procedure of Theorem \ref{condensation-products} of Section \ref{sec:bd-excitations}. Furthermore, the right adjoint $I$ of this procedure is exactly the inverse condensation procedure of Theorem \ref{inverse-condensation-products}.

In Ref. \cite{KitaevKong}, Kitaev and Kong have claimed that in the Levin-Wen model based on input fusion category $\mC$, the excitations on a boundary given by the indecomposable module $\M$ are given by objects in the fusion category $\Fun_\mC(\M,\M)$. By Ref. \cite{Davydov12}, we know that this category is equivalent to the category ${\mQ}$ obtained through the procedure (\ref{eq:condensation-quotient-IC}). In what follows, we will prove the claim of Ref. \cite{KitaevKong} for the case of Kitaev models for Dijkgraaf-Witten theories.  Related earlier results are obtained in \cite{Fuchs2014,Bark13b,Bark13c}.

In Section \ref{sec:bd-ribbon-operators}, we claimed that the elementary excitations on a boundary of type $K$ in a Kitaev model based on group $G$ are given by pairs $(T,R)$, where $T = K r_T K \in K\backslash G / K$ is a double coset in $G$, and $R$ is an irreducible representation of the group-theoretical quasi-Hopf algebra $\mZ = Z(G,1,K,1)$. We would now like to present the relationship between $Z(G,1,K,1)$ and the group-theoretical category $\mC(G,1,K,1)$ of Ref. \cite{Etingof05}, which is defined as follows:

\begin{definition}
Let $\text{Vec}_G^{\omega}$ be the category of finite-dimensional $G$-graded vector spaces with associativity $\omega$, where $G$ is a finite group and $\omega \in H^3(G, \C^\times)$. Let $K \subseteq G$ be a subgroup of $G$ and $\psi \in H^2(K, \C^\times)$ be a 2-cocycle of $K$ such that $d \psi = \omega |_K$. Let $\text{Vec}_G^{\omega}(K)$ be the subcategory of $\text{Vec}_G^{\omega}$ of objects graded by $K$. The twisted group algebra $A = \C_\psi [K]$ is then an associative algebra in $\text{Vec}_G^{\omega}(K)$. The {\it group-theoretical category} $\mC(G,1,K,1)$ is defined as the category of $(A,A)$-bimodules in $\text{Vec}_G^{\omega}$. In particular, $\mC(G,1,K,1)$ is a fusion category with tensor product $\otimes_A$ and unit object $A$.
\end{definition}

The goal of this section is now to establish the equivalence between $\mC(G,1,K,1)$ and the representation category of $Z(G,1,K,1)$ as fusion categories.

To begin, we state the following theorem:

\begin{theorem}
\label{Z-irreps}
Let $G$ be a finite group, and let $K \subseteq G$ be a subgroup. The irreducible representations of $Z(G,1,K,1)$ are given by pairs $(T,R)$, where $T = K r_T K \in K\backslash G /K$ is a double coset and $R$ is an irreducible representation of the subgroup $K^{r_T} = K \cap r_T K r_T^{-1}$ of $K$.
\end{theorem}

\begin{proof}
See Refs. \cite{Zhu01} and \cite{Schauenburg02}.
\end{proof}

By Ref. \cite{Gelaki07}, have the following theorem:

\begin{theorem}
The pairs $(T,R)$, as described in Theorem \ref{Z-irreps}, are in one-to-one correspondence with the simple objects in the group-theoretical category $\mC(G,1,K,1)$.
\end{theorem}

The above theorem implies that the elementary excitations on the boundary in the group-theoretical case are indeed given by the simple objects in the fusion category $\Fun_\mC(\M,\M)$.

Finally, Refs. \cite{Zhu01,Schauenburg02} show the equivalence of $\Rep(Z(G,1,K,1))$ and $\mC(G,1,K,1)$ as fusion categories:

\begin{theorem}
The representation category of the group-theoretical quasi-Hopf algebra $Z(G,1,K,1)$ (or equivalently, the representation category of the coquasi-Hopf algebra $Y(G,1,K,1)$) is equivalent as a fusion category to the group-theoretical category $\mC(G,1,K,1)$.
\end{theorem}

\begin{proof}
See Refs. \cite{Zhu01} and \cite{Schauenburg02}.
\end{proof}

This shows that the ``bordered topological order'' introduced in Section \ref{sec:bd-excitations} is indeed given by the fusion category $\Fun_\mC(\M,\M)$.

\subsection{$M$ symbols}
\label{sec:m-symbols}

The $F-6j$ symbols for fusion categories are numerical data that encode the associativity of the tensor products. They satisfy the pentagons to be coherent and depend on gauge choices.  Though hard to find and difficult to work with, they are indispensable for applications of tensor categories to topological phases of matter and topological quantum computation.  In the following, we generalize $F-6j$ symbols to $M-6j$ symbols to encode numerically the associativity of the condensation functor of bulk anyons to the boundary.  

Currently, there are no general software packages to solve for $M-6j$ or $M-3j$ symbols, but they are needed for the application to topological quantum computation with gapped boundaries (as seen in \cite{Cong16a}).  The few $M-3j$ symbols in this section are calculated by hand analytically.

The $M$ symbols might have appeared in many other contexts before.  If a gapped boundary is modeled by an indecomposable module category, then the $M-3j$ symbols should encode the associativity of the categorical action on the module category.  It would be interesting to find the precise relations of $M-3j$  symbols with some known symbols in the literature such as \lq\lq boundary $3j$ symbols" in Section $3$ of \cite{Pet01}, \lq\lq mixed $F 6j$ symbols"  in Section $4.1$ of \cite{Fuchs02}.  Our Eq. (3.13) for $M-3j$ symbols is called \lq\lq mixed pentagons"  in the literature, e.g. as Eq. (3.7) of \cite{Pet01} and Eq. (3.12) of \cite{Fuchs02}.  Similarly, there might be a possible relation of the $M-6j$ symbols with the \lq\lq bimodule $F$ symbols"  in \cite{Fuchs05}.

\subsubsection{The $M$-3$j$ symbol}
\label{sec:m3j}

Thus far, we have discussed the condensation of a single bulk anyon $a$ into a gapped boundary $\A$ in a system with topological order, and the tunneling of a single anyon from one gapped boundary to another. In this section, we consider the case where multiple bulk anyons all condense into the boundary, and present the resulting relations that govern the commutativity of bulk anyon fusion and condensation.

\begin{figure}
\centering
\includegraphics[width = 0.4\textwidth]{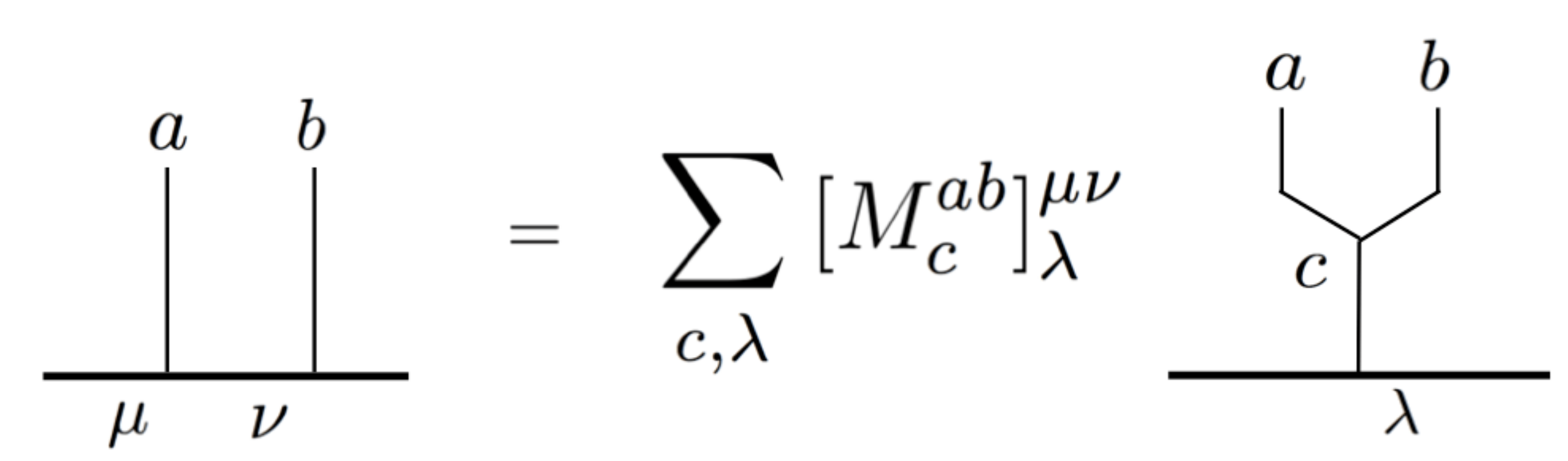}
\caption{Definition of the $M$-3$j$ symbol.}
\label{fig:m-3j}
\end{figure}

In a system with topological order, when three anyons fuse in the bulk, $F-6j$ symbols $F^{abc}_{d;ef}$ describe the associativity in the order of fusion. These 6$j$ symbols must satisfy certain pentagon and hexagon relations, corresponding to the pentagon and hexagon commutative diagrams for a modular tensor category. Similar associativity and braiding rules exist for the $M$ symbols, as we shall discuss below.

Let us first consider a relatively simple case, when the bulk anyons all condense to vacuum on the boundary. In Proposition \ref{separability-prop}, we showed that for any two anyons $a,b$ that can condense to vacuum on the boundary $\A$, there exists an injection $M$ from $\Hom(a,\A) \otimes \Hom(b,\A)$ to $\Hom(a \otimes b, \A)$. Physically, the $M$ operator corresponds to fusing the anyons $a,b$ in the bulk first, and then condensing to the boundary. The action of the $M$ operator is shown in Fig. \ref{fig:m-3j}.

In this case, each $M$ symbol has three topological charge indices, given by two original bulk anyons $a,b$ that condense to the boundary, and a third bulk anyon $c$ that results from the fusion of $a,b$ in the bulk. Furthermore, there will be local indices $\mu,\nu,\lambda$ corresponding to the condensation channels when $a,b,c$ condense to the boundary, respectively.

These $M$-3$j$ symbols are quite similar to the $\theta$-3$j$ symbols in a fusion category. If we start with three anyons in the bulk, the following pentagon diagram must commute:

\begin{equation}
\label{eq:m-3j-pentagon-fig}
\vcenter{\hbox{\includegraphics[width = 0.65\textwidth]{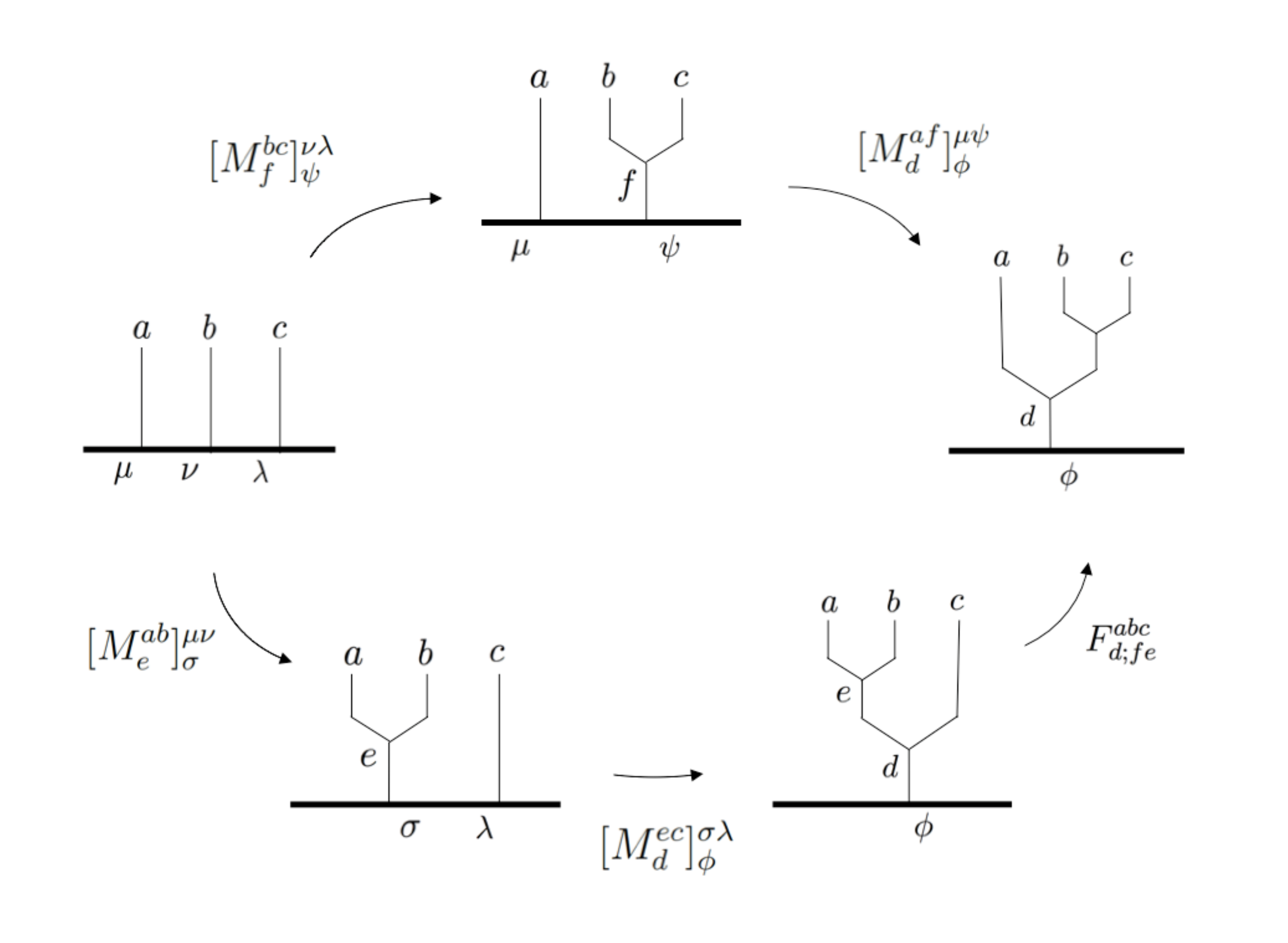}}}
\end{equation}

Algebraically, we can write Eq. \ref{eq:m-3j-pentagon-fig} as\footnote{In this and all associativity/braid relations that follow, we have assumed for simplicity of presentation that the anyon model has no fusion multiplicities. This is true in all of our examples, but the generalization is obvious.}:

\begin{equation}
\label{eq:m-3j-pentagon}
\sum_{e,\sigma} [M^{ab}_{e}]^{\mu\nu}_{\sigma} [M^{ec}_{d}]^{\sigma\lambda}_{\phi}F^{abc}_{d;fe} = \sum_\psi [M^{bc}_{f}]^{\nu\lambda}_{\psi} [M^{af}_{d}]^{\mu\psi}_{\phi}
\end{equation}

We note that the $M$ symbols will have gauge degrees of freedom, originating from the choice of basis for the condensation channels of each particle $a$. Specifically, we can define a unitary transformation $\Gamma^a_{\mu\nu}$ on the condensation space $V_a$: $\widetilde{\ket{a;\mu}} = \Gamma^a_{\mu\nu} \ket{a;\mu}$. These transformations yield new $M$ symbols, which are related to the original ones by the relation

\begin{equation}
\label{eq:m-3j-gauge-freedom}
[\tilde{M}^{ab}_c]^{\mu\nu}_\lambda =
\sum_{\mu',\nu',\lambda'} \Gamma^{a}_{\mu\mu'} \Gamma^{b}_{\nu\nu'}
[M^{ab}_{c}]^{\mu'\nu'}_{\lambda'}
[\Gamma^c]^{-1}_{\lambda\lambda'}
\end{equation}

The $M$ symbols will also be affected by gauge transformations of the bulk fusion space in the case of bulk fusion multiplicities.

Furthermore, when two bulk anyons $a,b$ condense to vacuum on the boundary, by the commutativity of the Frobenius algebra $\A$, it does not matter what order they condense in. Diagrammatically, this corresponds to

\begin{equation}
\label{eq:m-3j-braid-fig}
\vcenter{\hbox{\includegraphics[width = 0.4\textwidth]{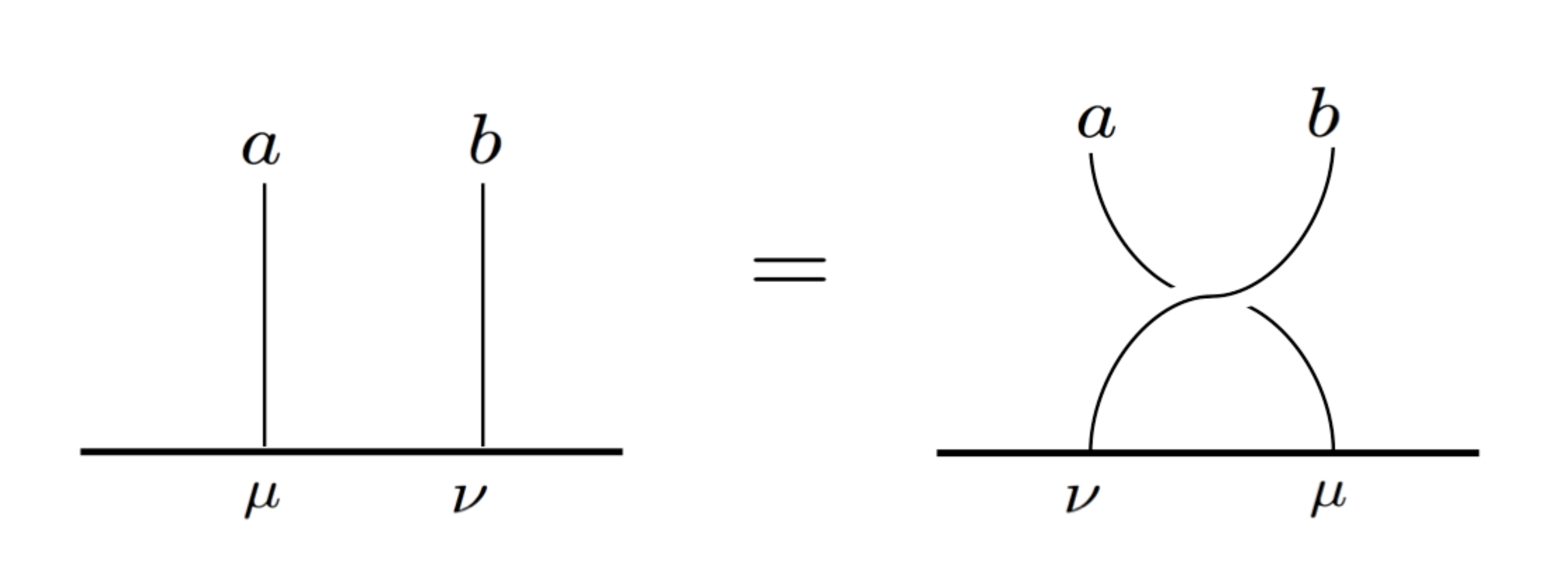}}}
\end{equation}

\noindent
This gives the following equation:

\begin{equation}
\label{eq:m-3j-braid}
\sum_c [M^{ba}_{c}]^{\nu\mu}_{\lambda} R^{ab}_c = \sum_c [M^{ab}_{c}]^{\mu\nu}_{\lambda}
\end{equation}

\noindent
Here, the sum is taken over all $c$ such that $\Hom(c,\A) \neq 0$.

Finally, we define some normalization conditions on these $M$-3$j$ symbols:

\begin{equation}
\label{eq:m-3j-normalization-1}
[M^{a\overbar{a}}_1]^{\mu\nu} = \frac{\delta_{\mu\nu}}{\sqrt{\textrm{FPdim}(\A)}}
\end{equation}

\begin{equation}
\label{eq:m-3j-normalization-2}
[M^{1a}_a]^\mu_\nu = [M^{a1}_a]^\mu_\nu = \delta_{\mu\nu}
\end{equation}

These normalizations are chosen so that $M$ becomes a partial isometry when we have the following basis vectors for the vector spaces $\Hom(a,\A)$, $\Hom(b,\A)$ and $\Hom(a \otimes b, \A)$:

\begin{equation}
\label{eq:m-3j-normalization-fig-1}
\vcenter{\hbox{\includegraphics[width = 0.45\textwidth]{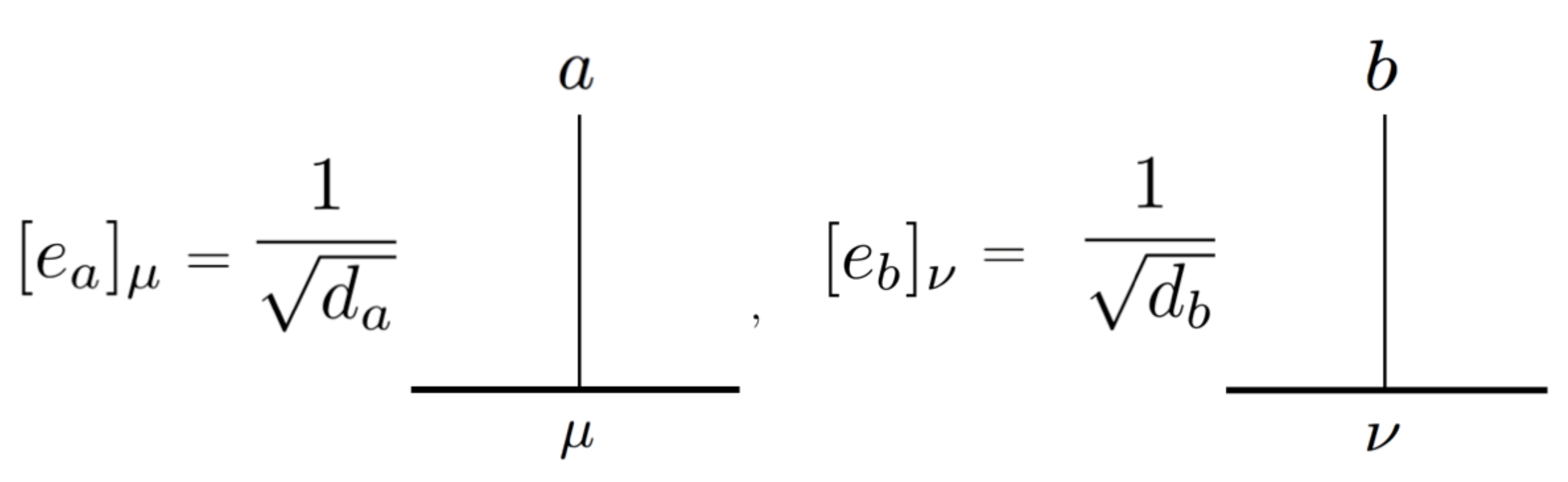}}}
\end{equation}

\begin{equation}
\label{eq:m-3j-normalization-fig-2}
\vcenter{\hbox{\includegraphics[width = 0.3\textwidth]{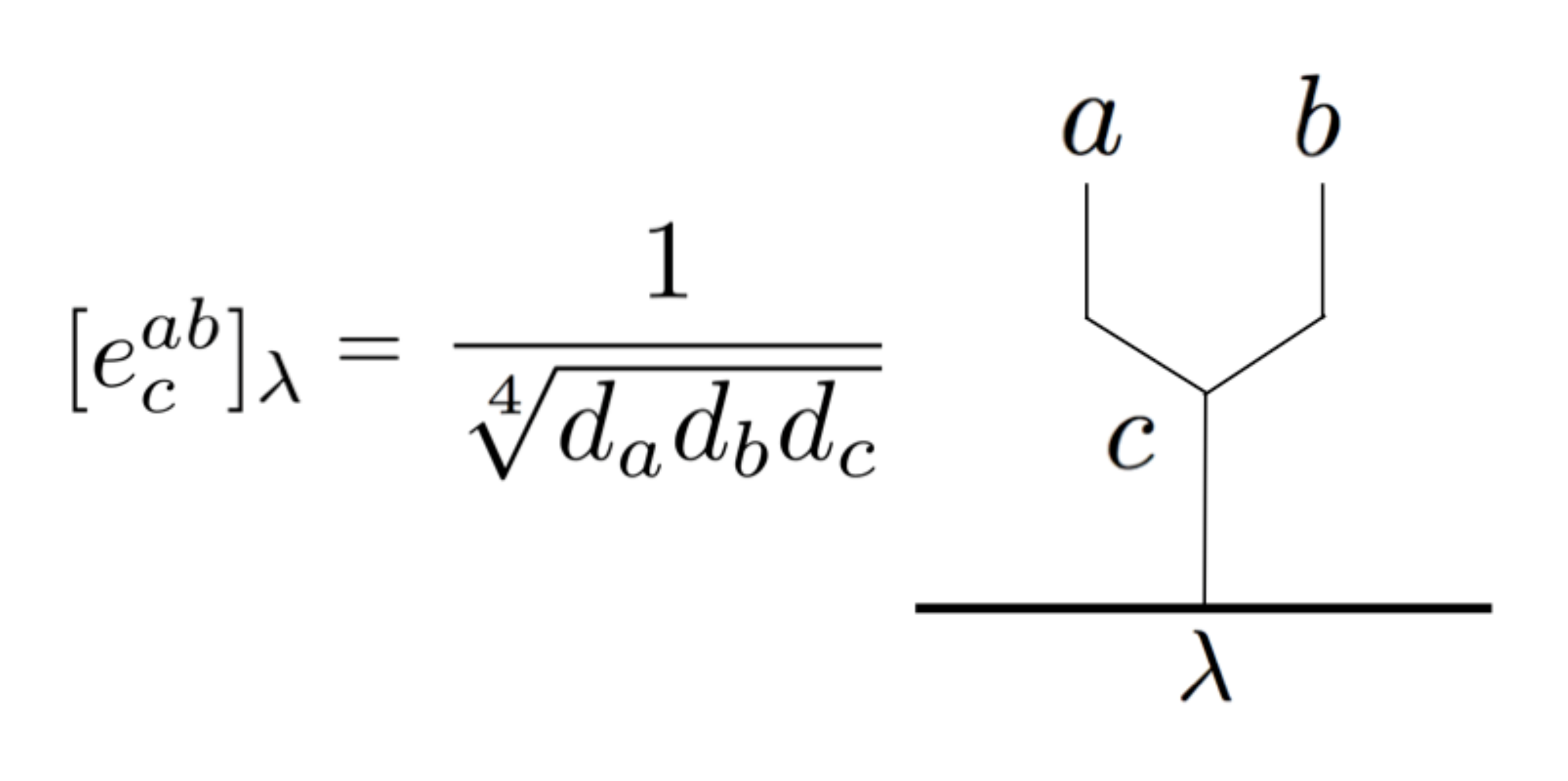}}}
\end{equation}

\noindent
The basis vectors for $\Hom(a,\A) \otimes \Hom(b,\A)$ are simply $[e_a]_\mu \otimes [e_b]_\nu$. These basis vectors are chosen because the following traces evaluate to 1, to provide orthonormal bases:

\begin{equation}
\vcenter{\hbox{\includegraphics[width = 0.22\textwidth]{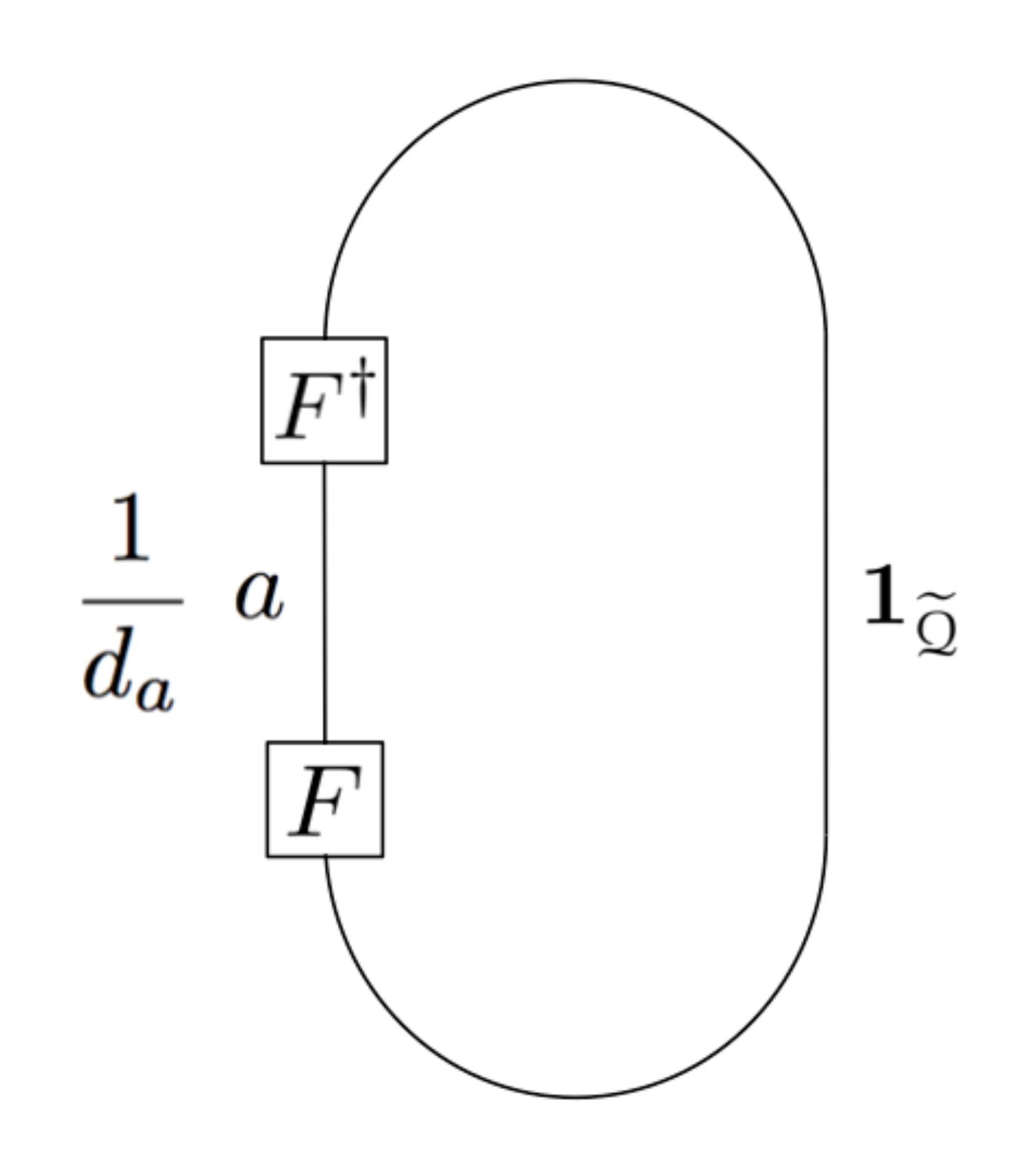}}} = 1,
\qquad
\vcenter{\hbox{\includegraphics[width = 0.28\textwidth]{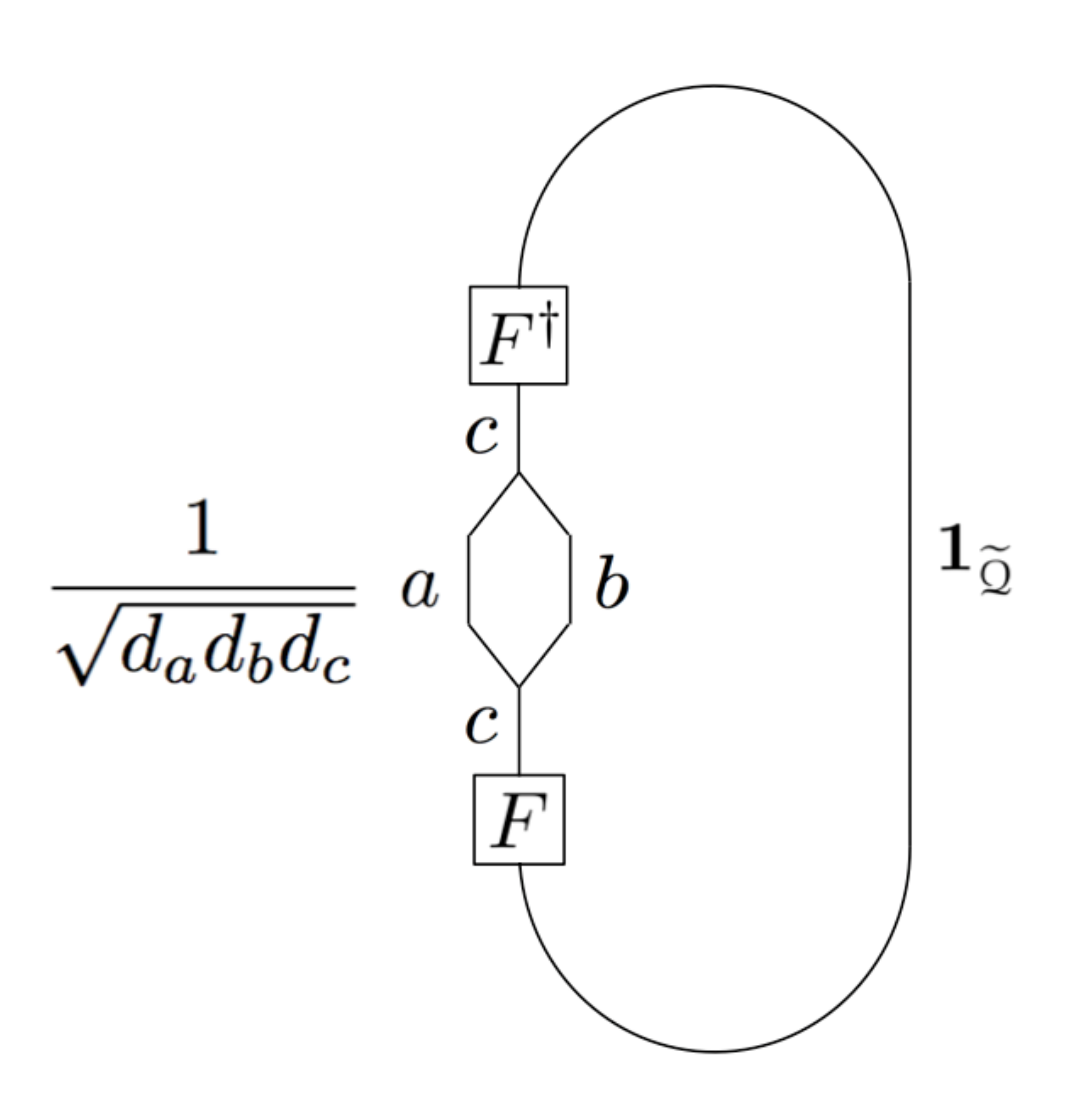}}} = 1
\end{equation}

\subsubsection{The $M$-6$j$ symbol}
\label{sec:m6j}

In the above discussion, we have considered only a special case, where all bulk anyons condense to vacuum on the boundary. More generally, we may consider a case where the bulk anyons become excitations on the boundary, as in Section \ref{sec:bd-excitations}. Here, we will define a $M$-6$j$ symbol, with six topological charge indices, whose action is shown in Fig. \ref{fig:m-6j}. We note that our $M-6j$ symbol is a symmetric version of the vertex lifting coefficients introduced in Ref. \cite{Eliens13}.

\begin{figure}
\centering
\includegraphics[width = 0.4\textwidth]{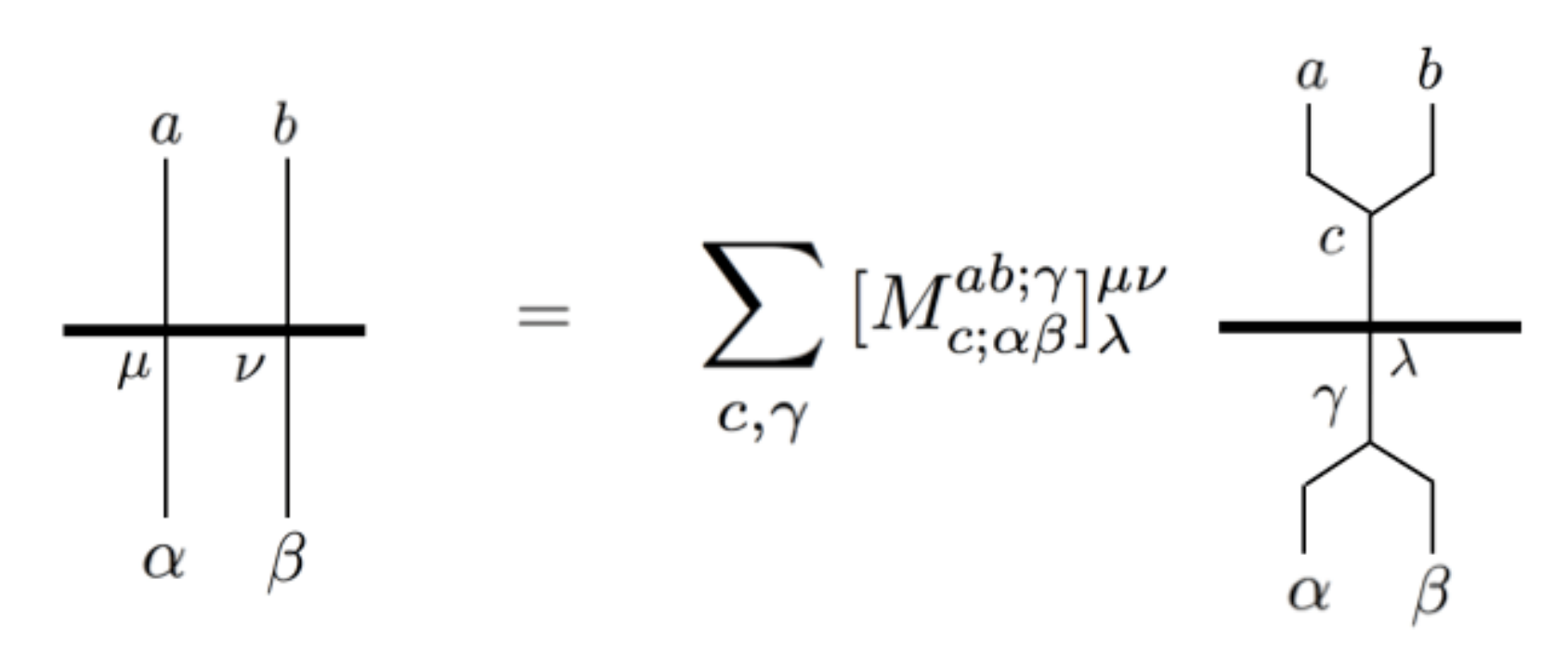}
\caption{Definition of the $M$-6$j$ symbol}
\label{fig:m-6j}
\end{figure}

The $M$-6$j$ symbol is indexed by bulk anyons $a,b,c$, and the boundary excitations $\alpha,\beta,\gamma$ they condense to. As before, we will also have condensation channel labels $\mu,\nu,\lambda$ for the multiplicity corresponding to the dimension of $\Hom(a,I(\alpha))$, etc.

As in the case of $M$-3$j$ symbols, these symbols must also satisfy a pentagon associativity relation. However, this relation will also depend on the $F-6j$ symbols of the fusion category $\Fun_\mC(\M, \M)$. The associativity is hence given by the following commuting pentagon:

\begin{equation}
\label{eq:m-6j-pentagon-fig}
\vcenter{\hbox{\includegraphics[width = 0.65\textwidth]{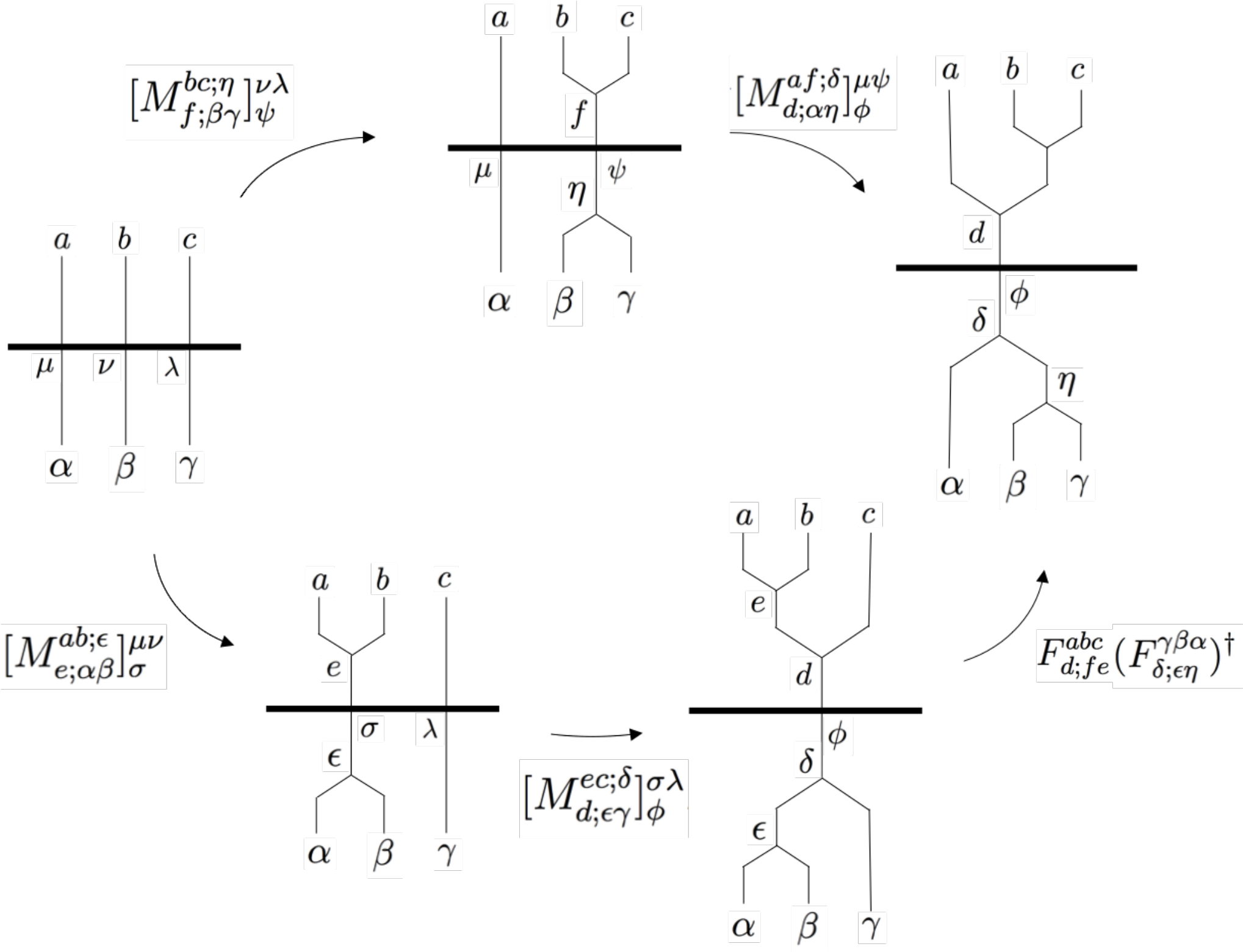}}}
\end{equation}

This is equivalent to the equation

\begin{equation}
\label{eq:m-6j-pentagon}
\sum_{e,\sigma,\epsilon} [M^{ab;\epsilon}_{e;\alpha\beta}]^{\mu\nu}_{\sigma} [M^{ec;\delta}_{d;\epsilon\gamma}]^{\sigma\lambda}_{\phi}F^{abc}_{d;fe} (F^{\gamma\beta\alpha}_{\delta;\epsilon\eta})^\dagger = \sum_\psi [M^{bc;\eta}_{f;\beta\gamma}]^{\nu\lambda}_{\psi} [M^{af;\delta}_{d;\alpha\eta}]^{\mu\psi}_{\phi}
\end{equation}

Similarly, the braiding relation is now given by the diagram

\begin{equation}
\label{eq:m-6j-braid-fig}
\vcenter{\hbox{\includegraphics[width = 0.4\textwidth]{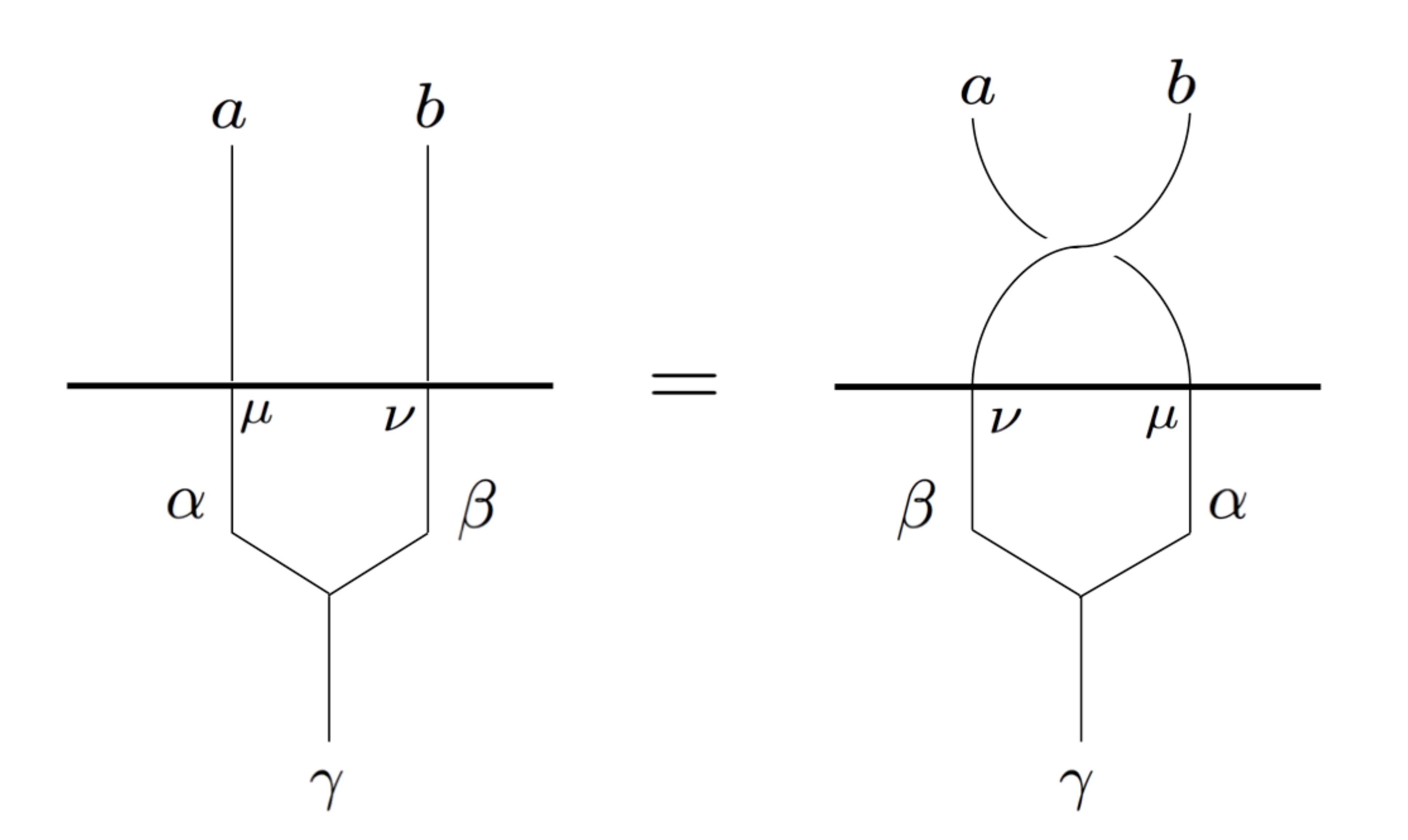}}}
\end{equation}

\noindent
which is equivalent to the equation

\begin{equation}
\label{eq:m-6j-braid}
\sum_c [M^{ba;\gamma}_{c;\beta\alpha}]^{\nu\mu}_{\lambda} R^{ab}_c = \sum_c [M^{ab;\gamma}_{c;\alpha\beta}]^{\mu\nu}_{\lambda}.
\end{equation}

\noindent
The sum is taken over all simple objects $c$ such that $\Hom(c,I(\gamma))\neq 0$, where $I$ is the right adjoint of the condensation procedure $F$.

Finally, we enforce $M$ to be a partial isometry with respect to the following choices of basis for $\Hom(a,I(\alpha)) \otimes \Hom(b,I(\beta))$ and $\Hom(a \otimes b, I(\alpha \otimes \beta))$:

\begin{equation}
\label{eq:m-6j-normalization-fig-1}
\vcenter{\hbox{\includegraphics[width = 0.45\textwidth]{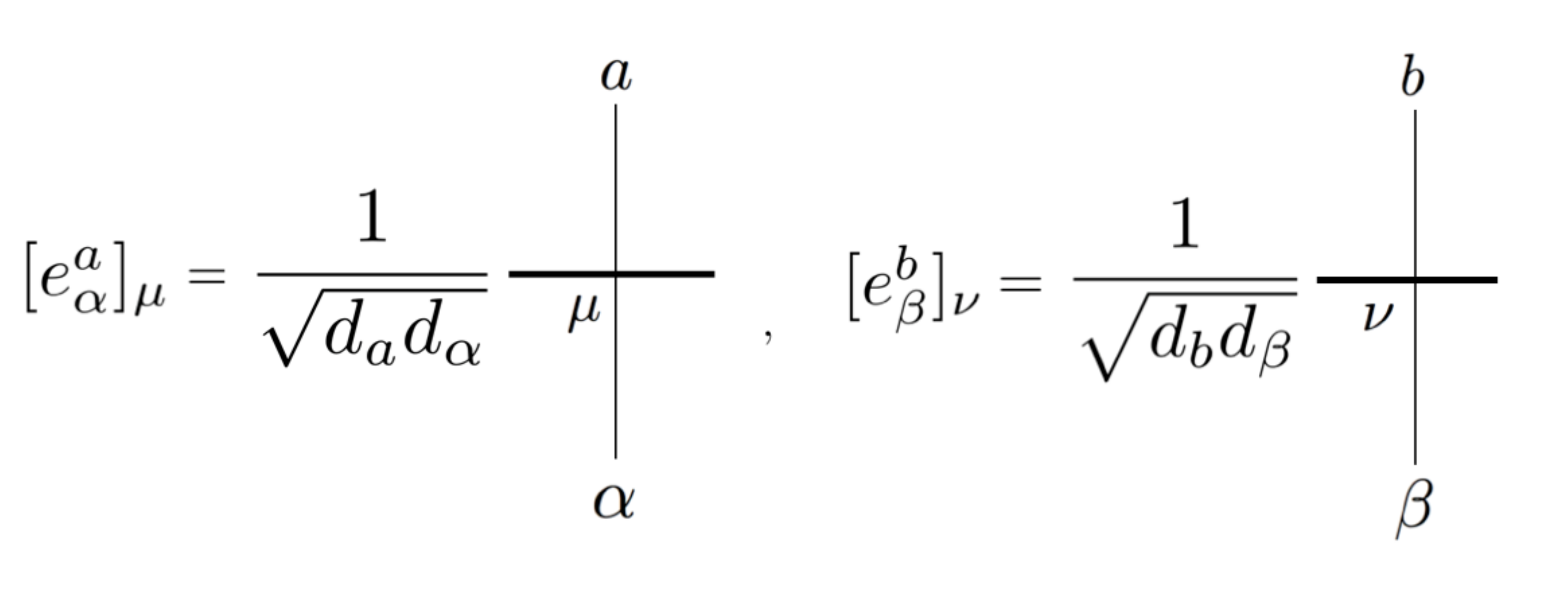}}}
\end{equation}

\begin{equation}
\label{eq:m-6j-normalization-fig-2}
\vcenter{\hbox{\includegraphics[width = 0.4\textwidth]{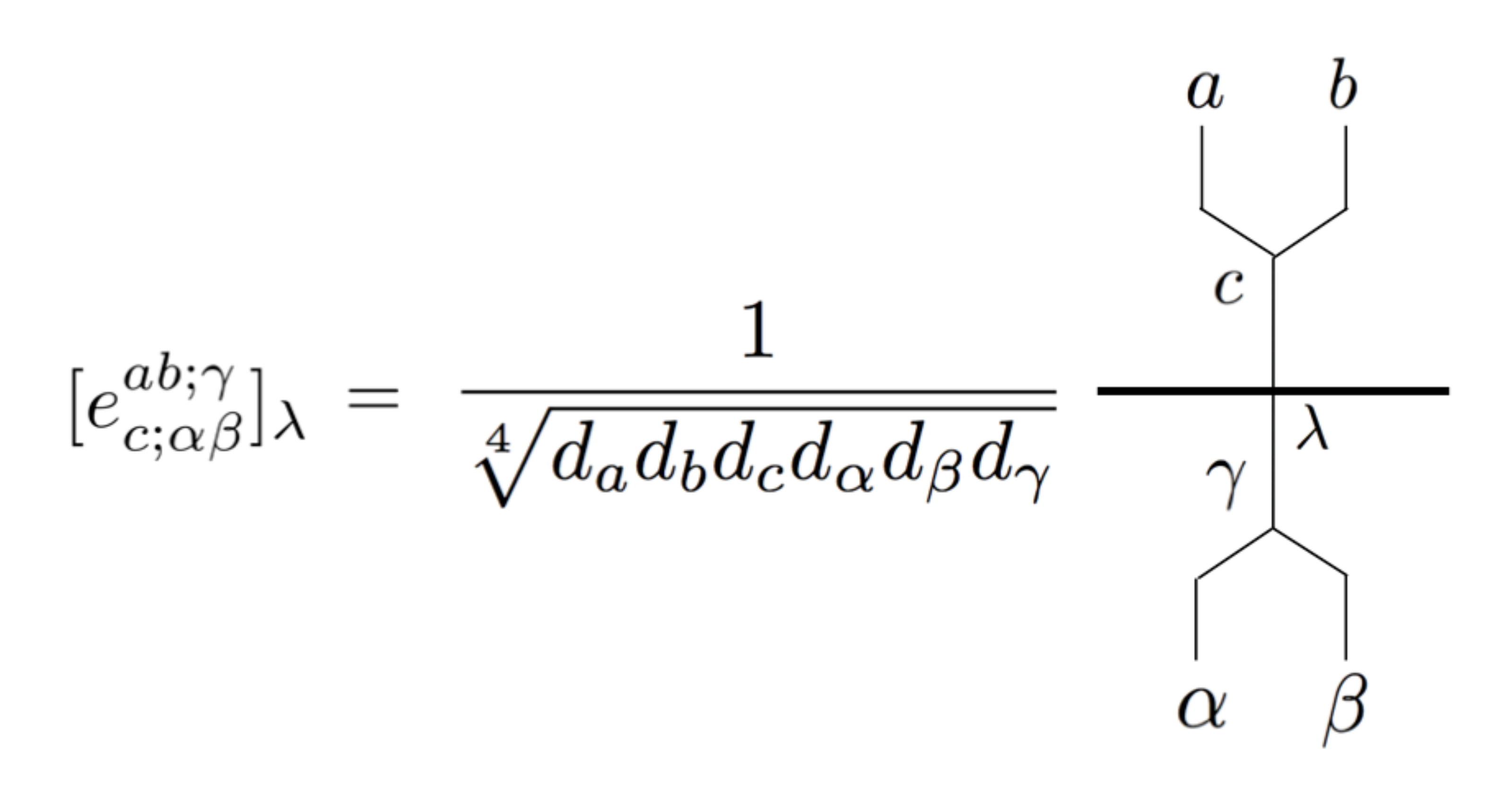}}}
\end{equation}

\noindent
As before, the basis vectors for $\Hom(a,I(\alpha)) \otimes \Hom(b,I(\beta))$ are given by $[e^a_\alpha]_\mu \otimes [e^b_\beta]_\nu$, and these bases are chosen because the following traces evaluate to 1:

\begin{equation}
\vcenter{\hbox{\includegraphics[width = 0.24\textwidth]{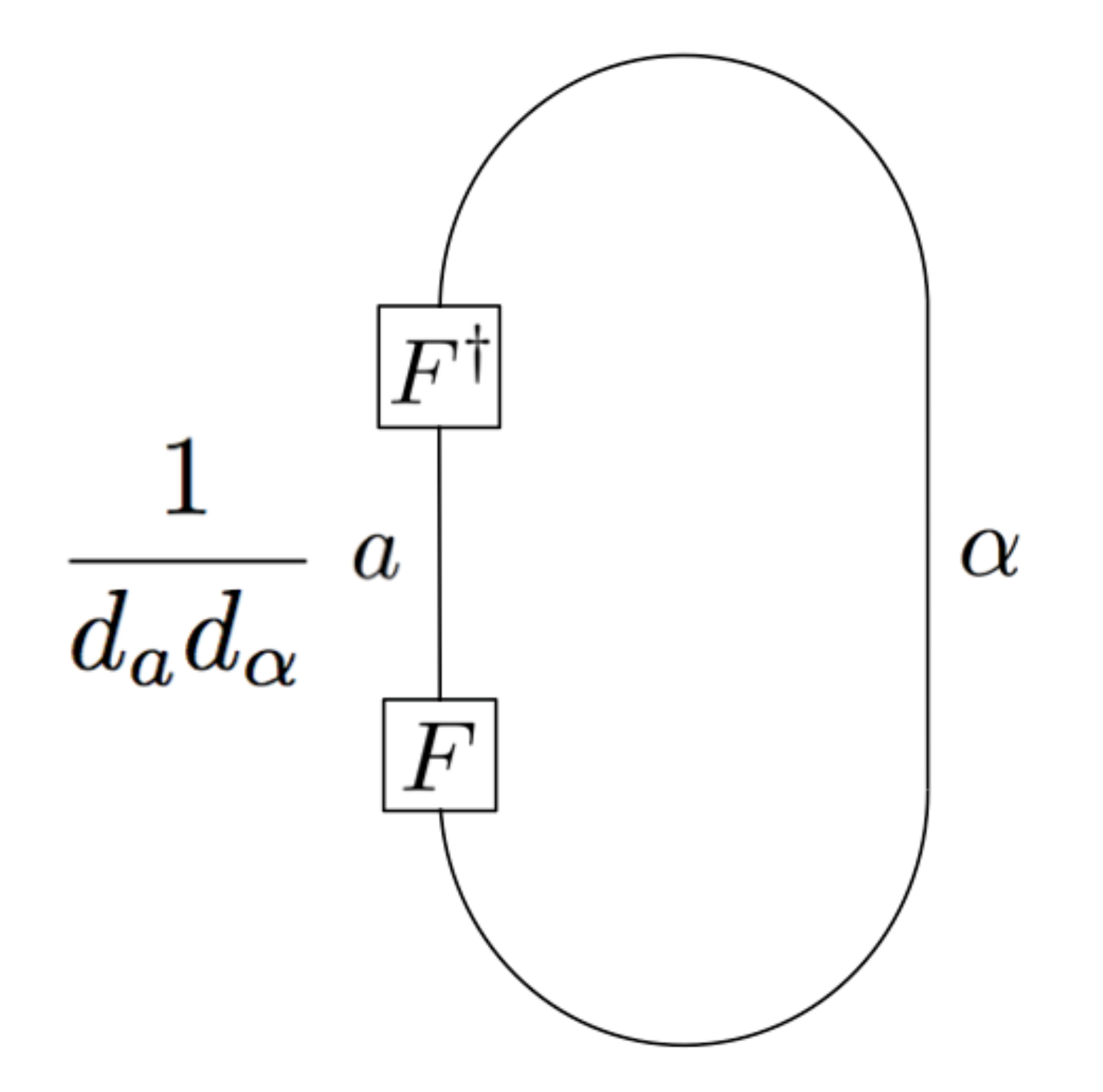}}} = 1,
\qquad
\vcenter{\hbox{\includegraphics[width = 0.35\textwidth]{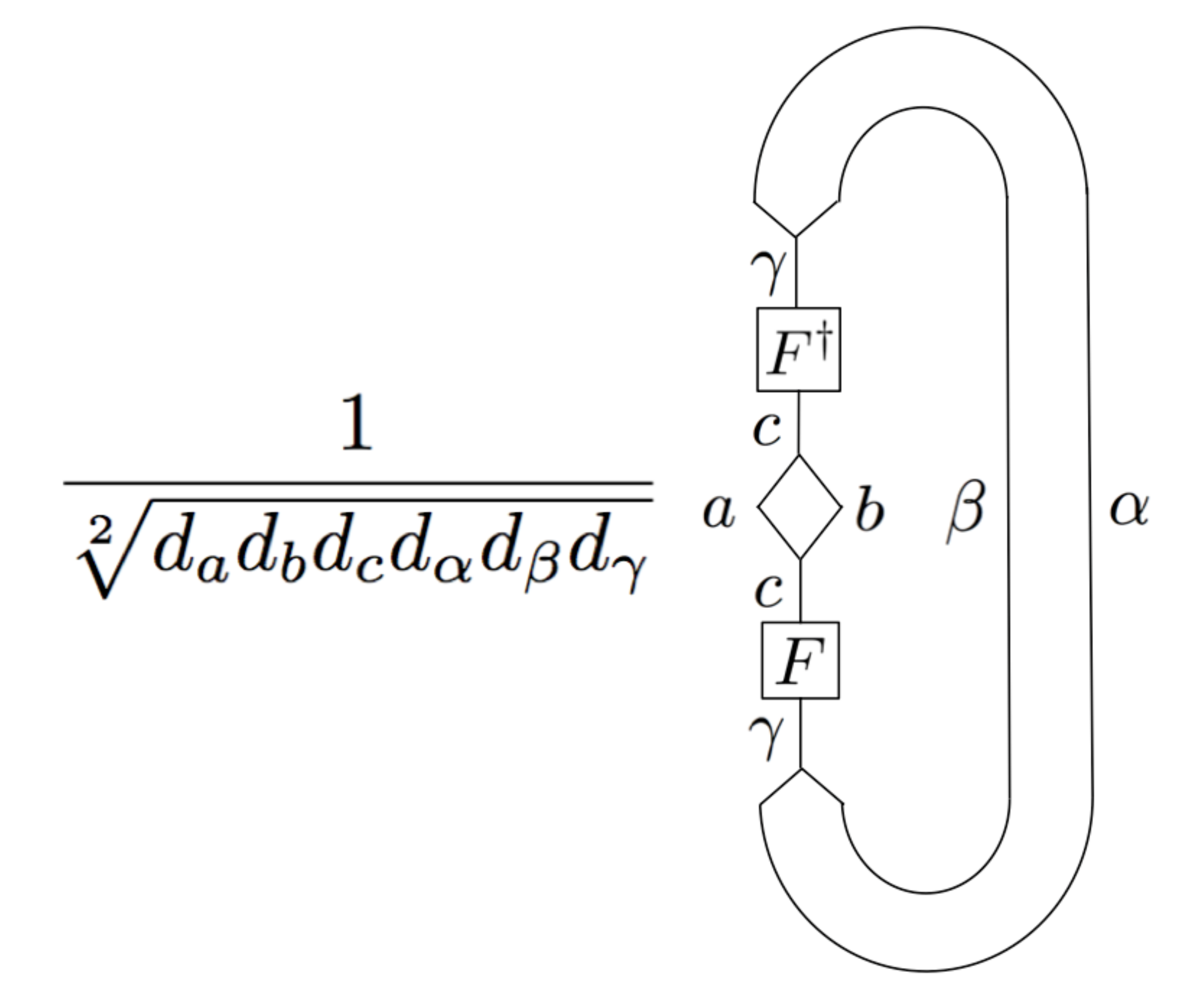}}} = 1
\end{equation}

This gives the normalization condition 

\begin{equation}
\label{eq:m-6j-normalization}
[M^{1a;\alpha}_{a;1\alpha}]^\mu_\nu = [M^{a1;\alpha}_{a;\alpha 1}]^\mu_\nu = \delta_{\mu\nu}.
\end{equation}

In general, $M$ symbols may be computed analytically using software packages. We note that given all data for the categories $\B = \mZ(\mC)$ and $\mQ = \Fun_\mC(\M,\M)$, these symbols are typically easier to calculate than the solutions to the standard $F-6j$ symbols, since Eq. (\ref{eq:m-6j-pentagon}) is at most quadratic.

\subsection{Example: $\mfD(S_3)$}
\label{sec:ds3-algebraic-example}

\subsubsection{Topological order}

When $G = S_3$, the topological order of the resulting Kitaev model is given by the modular tensor category $\B = \Rep(D(S_3)) = \mZ(\Rep(S_3))$. As discussed in Section \ref{sec:ds3-hamiltonian-example}, there are 8 simple objects in this category, $A,B,...,H$. The fusion rules \cite{Cui15} are given in Table \ref{tab:DS3-fusion}, and the $\mathcal{S}$, $\mathcal{T}$ matrices are listed below. The $F$ and $R$ symbols for this category may be found in Appendix A of Ref. \cite{Cui15}.

{\tiny
\begin{table}\caption{Fusion rules of $\mfD(S_3)$}\label{tab:DS3-fusion}
\begin{tabular}{|c|c|c|c|c|c|c|c|c|}
\hline $\otimes$ &$A$ &$B$ &$C$ &$D$ &$E$ &$F$ &$G$ &$H$\\ \hline
$A$ &$A$ &$B$ &$C$ &$D$ &$E$& $F$ &$G$ &$H$\\ \hline
$B$ &$B$ &$A$ &$C$& $E$ &$D$ &$F$ &$G$ &$H$\\ \hline
$C$ &$C$ &$C$ &$A\oplus B\oplus C$& $D\oplus E$ &$D\oplus E$ & $G\oplus H$& $F\oplus H$ &$F\oplus G$\\ \hline
\multirow{2}{*}{$D$} &\multirow{2}{*}{$D$} &\multirow{2}{*}{$E$} &\multirow{2}{*}{$D\oplus E$}& $A\oplus C\oplus  F$ & $B\oplus C\oplus F$ & \multirow{2}{*}{$D\oplus E$} & \multirow{2}{*}{$D\oplus E$} & \multirow{2}{*}{$D\oplus E$} \\
& & & & $\oplus G\oplus H$ & $\oplus G\oplus H$ & & &  \\ \hline
\multirow{2}{*}{$E$} &\multirow{2}{*}{$E$}& \multirow{2}{*}{$D$}& \multirow{2}{*}{$D\oplus E$} & $B\oplus C\oplus F$ & $A\oplus C\oplus F$ & \multirow{2}{*}{$D\oplus E$} &\multirow{2}{*}{$D\oplus E$} & \multirow{2}{*}{$D\oplus E$} \\
& & & & $\oplus G\oplus H$ & $\oplus G\oplus H$ & & &  \\  \hline
$F$ &$F$ & $F$& $G\oplus H$& $D\oplus E$ & $D\oplus E$ & $A\oplus B\oplus F$ & $H\oplus C$ & $G\oplus C$ \\ \hline
$G$ &$G$ & $G$& $F\oplus H$ & $D\oplus E$ & $D\oplus E$ & $H\oplus C$ & $A\oplus B\oplus G$ & $F\oplus C$ \\ \hline
$H$ &$H$ & $H$& $F\oplus G$ & $D\oplus E$ & $D\oplus E$ & $G\oplus C$ & $F\oplus C$ & $A\oplus B\oplus H$\\\hline
\end{tabular}
\end{table}
}

The modular $\mathcal{S}$ and $\mathcal{T}$ matrices of $\mfD(S_3)$ are given by \cite{Cui15}:

\begin{equation}
\label{eq:ds3-S}
\mathcal{S} = \frac{1}{6}
\begin{bmatrix}
1 & 1 & 2 & 3 & 3 & 2 & 2 & 2 \\
1 & 1 & 2 & -3 & -3 & 2 & 2 & 2 \\
2 & 2 & 4 & 0 & 0 & -2 & -2 & -2 \\
3 & -3 & 0 & 3 & -3 & 0 & 0 & 0 \\
3 & -3 & 0 & -3 & 3 & 0 & 0 & 0 \\
2 & 2 & -2 & 0 & 0 & 4 & -2 & -2 \\
2 & 2 & -2 & 0 & 0 & -2 & -2 & 4 \\
2 & 2 & -2 & 0 & 0 & -2 & 4 & -2 \\
\end{bmatrix}
\end{equation}

\begin{equation}
\label{eq:ds3-T}
\mathcal{T} = \text{diag}(1,1,1,1,-1,1,\omega,\omega^2)
\end{equation}

\noindent
where all rows and columns are ordered alphabetically, $A-H$, and $\omega = e^{2 \pi i / 3}$ is the primitive third root of unity.

\subsubsection{Lagrangian algebras}

We can determine gapped boundary types of the $\mfD(S_3)$ model by computing all Lagrangian algebras in $\B$, using the procedure of Section \ref{sec:frobenius-algebras}. This gives four gapped boundary types: $\mathcal{A}_1 = A+C+D$, $\mathcal{A}_2 = A+B+2C$, $\mathcal{A}_3 = A+F+D$, and $\mathcal{A}_4 = A+B+2F$.

\subsubsection{Condensation procedure: $A+C+D$ boundary}

We now illustrate the condensation procedure of Eq. (\ref{eq:condensation-quotient-IC}) on the $A+C+D$ boundary of the $\mfD(S_3)$ theory. We first form the quotient pre-category ${\widetilde{\mQ}}$, which has the same objects as $\B = \mZ(\Rep(S_3))$. By Definition \ref{quotient-cat-def}, we have:

\begin{equation}
\Hom_{\widetilde{\mQ}}(A,C) = \Hom_\B(A, C+A+B+C+D+E) \cong \C
\end{equation}

\noindent
Similarly, 

\begin{equation}
\begin{gathered}
\Hom_{\widetilde{\mQ}}(A,D) \cong \C \qquad \Hom_{\widetilde{\mQ}}(C,D) \cong \C \\
\Hom_{\widetilde{\mQ}}(A,A) \cong \C \qquad \Hom_{\widetilde{\mQ}}(B,B) \cong \C \\
\Hom_{\widetilde{\mQ}}(B,C) \cong \C \qquad \Hom_{\widetilde{\mQ}}(B,E) \cong \C \\
\Hom_{\widetilde{\mQ}}(F,D) \cong \C \qquad \Hom_{\widetilde{\mQ}}(F,F) \cong \C \\
\Hom_{\widetilde{\mQ}}(F,G) \cong \C \qquad \Hom_{\widetilde{\mQ}}(F,H) \cong \C
\end{gathered}
\end{equation}

\noindent
Many other hom-sets in ${\widetilde{\mQ}}$ may be constructed from the above by composition (simply tensor product the corresponding hom-spaces). All other hom-sets in this category are zero. Furthermore, no idempotent completion is necessary, as all endomorphism spaces in $\widetilde{\mQ}$ for simple objects are one-dimensional and hence have no nontrivial splitting idempotents. It follows that the following rules describe the condensation of simple bulk particles onto the $A+C+D$ boundary:

\begin{multicols}{2}
\begin{enumerate}[label=(\roman*),leftmargin=0.5in]
\item
$A \rightarrow {A}$
\item
$B \rightarrow {B}$
\item
$C \rightarrow {A} \oplus {B}$
\item
$D \rightarrow {A} \oplus {F}$
\item
$E \rightarrow {B} \oplus {F}$
\item
$F,G,H \rightarrow {F}$
\end{enumerate}
\end{multicols}

We note that this is exactly the same result we obtained in Section \ref{sec:ds3-hamiltonian-example}, if we identify the boundary excitation label $F$ with the label $C$ from that section.

The $A+F+D$ boundary is easily shown to have the same condensation rules and properties, with all instances of $C$ and $F$ switched.

\subsubsection{Condensation procedure: $A+B+2C$ boundary}

We now illustrate the condensation procedure of Eq. (\ref{eq:condensation-quotient-IC}) on the $A+B+2C$ boundary of the $\mfD(S_3)$ theory, as this example will require a nontrivial idempotent completion. As before, we first use Definition \ref{quotient-cat-def} to construct the quotient pre-category $\widetilde{\mQ}$. This gives the following hom-sets:

\begin{equation}
\begin{gathered}
\Hom_{\widetilde{\mQ}}(A,A) \cong \C \qquad \Hom_{\widetilde{\mQ}}(A,B) \cong \C \\
\Hom_{\widetilde{\mQ}}(A,C) \cong \C^2 \qquad \Hom_{\widetilde{\mQ}}(D,D) \cong \C^3 \\
\Hom_{\widetilde{\mQ}}(D,E) \cong \C^3 \qquad \Hom_{\widetilde{\mQ}}(E,E) \cong \C^3 \\
\Hom_{\widetilde{\mQ}}(F,F) \cong \C^2 \qquad \Hom_{\widetilde{\mQ}}(F,G) \cong \C^2 \\
\Hom_{\widetilde{\mQ}}(F,H) \cong \C^2 \qquad \Hom_{\widetilde{\mQ}}(G,G) \cong \C^2 \\
\Hom_{\widetilde{\mQ}}(G,H) \cong \C^2 \qquad \Hom_{\widetilde{\mQ}}(H,H) \cong \C^2 \\
\end{gathered}
\end{equation}

All other hom-sets in $\widetilde{\mQ}$ between simple objects in $\B$ are either tensor products of the above (in case of composition), or zero. It follows that the quotient functor acts as follows on the simple objects of $\B$:

\begin{multicols}{2}
\begin{enumerate}[label=(\roman*),leftmargin=0.5in]
\item
$A,B \rightarrow {A}$
\item
$C \rightarrow {2 \cdot A}$
\item
$D,E \rightarrow D$
\item
$F,G,H \rightarrow F$
\end{enumerate}
\end{multicols}

However, we would now like to note that $\widetilde{\mQ}$, with simple objects given by $A$, $D$, and $F$, is not semisimple, as we see $\Hom_{\widetilde{\mQ}}(D,D) \cong \C^3$ and $\Hom_{\widetilde{\mQ}}(F,F) \cong \C^2$ when they should be one-dimensional. This tells us that we must perform the canonical idempotent completion of $\widetilde{\mQ}$ to ${\mQ}$, which transforms the simple objects of $\widetilde{\mQ}$ as follows:

\begin{enumerate}[label=(\roman*),leftmargin=0.5in]
\item
$A \rightarrow A$
\item
$D \rightarrow (D,p_1) \oplus (D,p_2) \oplus (D,p_3)$
\item
$F \rightarrow (F,q_1) \oplus (F,q_2)$
\end{enumerate}

\noindent
where each $p_i$ is a splitting idempotent in $\Hom_{\widetilde{\mQ}}(D,D)$, and each $q_j$ is a splitting idempotent in $\Hom_{\widetilde{\mQ}}(F,F)$. (In general, if $\Hom_{\widetilde{\mQ}}(X,X)$ is $n$-dimensional, there are $n$ splitting idempotents).

Hence, we have the following rules for the overall condensation procedure of simple bulk anyons of $\mfD(S_3)$ to the $A+B+2C$ boundary:

\begin{multicols}{2}
\begin{enumerate}[label=(\roman*),leftmargin=0.5in]
\item
$A,B \rightarrow {A}$
\item
$C \rightarrow {2 \cdot A}$
\item
$D,E \rightarrow (D,p_1) \oplus (D,p_2) \oplus (D,p_3)$
\item
$F,G,H \rightarrow (F,q_1) \oplus (F,q_2)$
\end{enumerate}
\end{multicols}

We note that this is exactly the same result we obtained in Section \ref{sec:ds3-hamiltonian-example}, if we identify the above boundary excitation labels with those of Section \ref{sec:ds3-hamiltonian-example} as follows:

\begin{multicols}{2}
\begin{enumerate}[label=(\roman*),leftmargin=0.5in]
\item
$A \rightarrow 1$
\item
$(F,q_1) \rightarrow r$
\item
$(F,q_2) \rightarrow r^2$
\item
$(D,p_1) \rightarrow s$
\item
$(D,p_2) \rightarrow sr$
\item
$(D,p_3) \rightarrow sr^2$
\end{enumerate}
\end{multicols}

The $A+B+2F$ boundary is easily shown to have the same condensation rules and properties, with all instances of $C$ and $F$ switched.

\subsubsection{$M$-3$j$ symbols of the $A+C+D$ boundary}
\label{sec:ds3-m3j}

We have computed the $M$-3$j$ symbols of the $A+C+D$ boundary by hand, up to a sign in a few cases. By Equations (\ref{eq:m-3j-pentagon}) and (\ref{eq:m-3j-braid}-\ref{eq:m-3j-normalization-2}), we have:

\begin{equation}
\begin{gathered}
M^{AX}_X = 1, \text{ } X = A,C,D \\
M^{CC}_A = \frac{1}{\sqrt{6}} \qquad M^{CC}_C = \pm \frac{i}{\sqrt{2}} \\
M^{DD}_A = \frac{1}{\sqrt{6}} \qquad M^{DD}_C = \pm i \sqrt{\frac{2}{3}} \\
M^{CD}_D = M^{DC}_D = \mp i
\end{gathered}
\end{equation}

\noindent
Other $M-3j$ symbols for this boundary are all zero.

\vspace{2mm}
\section{Conclusions}
\label{sec:conclusions}

Based on Kitaev's quantum double models and Bombin and Martin-Delgado's two-parameter generalization, we find exactly solvable Hamiltonian realizations of gapped boundaries and their excitations in Dijkgraaf-Witten theories. By combining with an algebraic model, we develop a microscopic theory for these boundaries and excitations. 

We would like to conclude by considering several potential areas to generalize our work. First, many physics papers have studied gapped domain walls between different topological phases. While gapped boundaries are often considered as a special case of gapped domain walls, by the folding trick, they also completely cover the domain wall theory mathematically. Physically, however, it is still interesting to analyze the general gapped domain walls following our work.

Another direction is to generalize our theory to the Levin-Wen model, using quantum groupoids.

The most interesting question that we have not touched on is the stability of the topological degeneracy in our model.  Once our Hamiltonian moves off the fixed-point, finite-size splitting of the degeneracy would occur.  It would be interesting to study the energy spitting of the ground state degeneracies of our Hamiltonians $H_{\text{G.B.}}$ numerically under small perturbations.

Gapped boundaries and symmetry defects significantly enrich the physics of topological phases of matter in two spacial dimensions.  Their higher dimensional generalizations would be also very interesting.  Experimental confirmation of non-abelian objects such as non-abelian anyons, gapped boundaries, and parafermion zero modes would be a landmark in condensed matter physics.

\vspace{2mm}
\begin{appendix}

\vspace{2mm}

\section{Notations}
\label{sec:notations}

In this Appendix, we list all of the notations that are used throughout the paper.

In general, if $G$ is any finite group, we denote the set of irreducible representations of $G$ by $(G)_{\text{ir}}$.

We will adopt the following conventions for labeling qudits, anyons, gapped boundaries, and their excitations:

\begin{enumerate}
\item
The data qudits in the bulk will be labeled as $g_1, g_2, ...$, $h_1, h_2, ...$ for the Kitaev model, where they are members of a finite group $G$.
\item
The data qudits on the boundary will be labeled as $k_1, k_2, ...$ for the Kitaev model, where they are members of a subgroup $K \subseteq G$.
\item
In the more general case where data qudits are simple objects in a unitary fusion category $\CC$, we will label the bulk data qudits as $x_1, x_2,...$, $y_1, y_2,...$.
\item
In this same general case, the boundary data qudits will be labeled as $r_1, r_2, ...$, $s_1, s_2, ...$.
\item
Bulk excitations (a.k.a. anyons or topological charges), which are the simple objects within the modular tensor category $\B = \ZZ(\text{Rep}(G))$ or $\B = \ZZ(\CC)$, will be labeled by $a,b,c...$. Their dual excitations are labeled by $\overbar{a}, \overbar{b}, \overbar{c}, ...$, respectively.
\item
The gapped boundary will be given as a Lagrangian algebra $\A$ which is an object in $\B$.
\item
Excitations on the boundary will be labeled as $\alpha, \beta, \gamma, ...$. When necessary, the local degrees of freedom during condensation will be labeled as $\mu, \nu, \lambda, ...$.
\item We use FPdim to denote the quantum dimension, and dim for the usual linear algebra dimension of a vector space or representation.
\end{enumerate}

Furthermore, when using any $F$ symbols and $R$ symbols for a fusion category or a modular tensor category, we will adopt the following conventions for indices:

\begin{equation}
\vcenter{\hbox{\includegraphics[width = 0.48\textwidth]{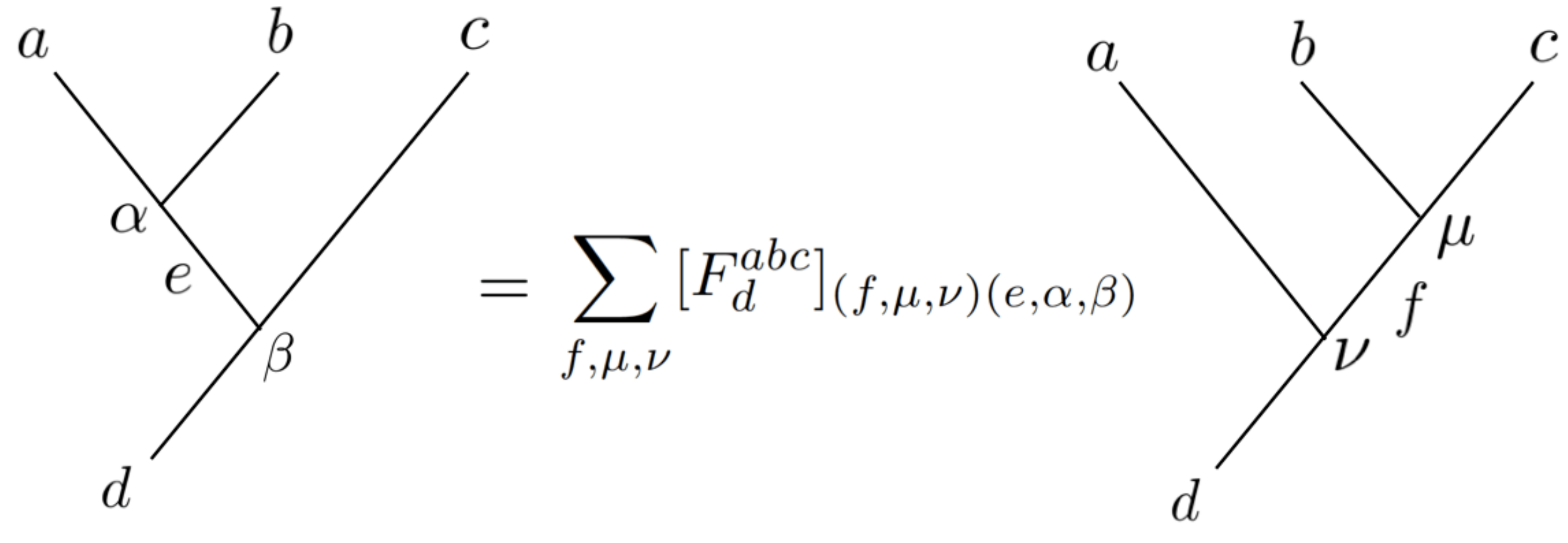}}}
\end{equation}

\begin{equation}
\vcenter{\hbox{\includegraphics[width = 0.28\textwidth]{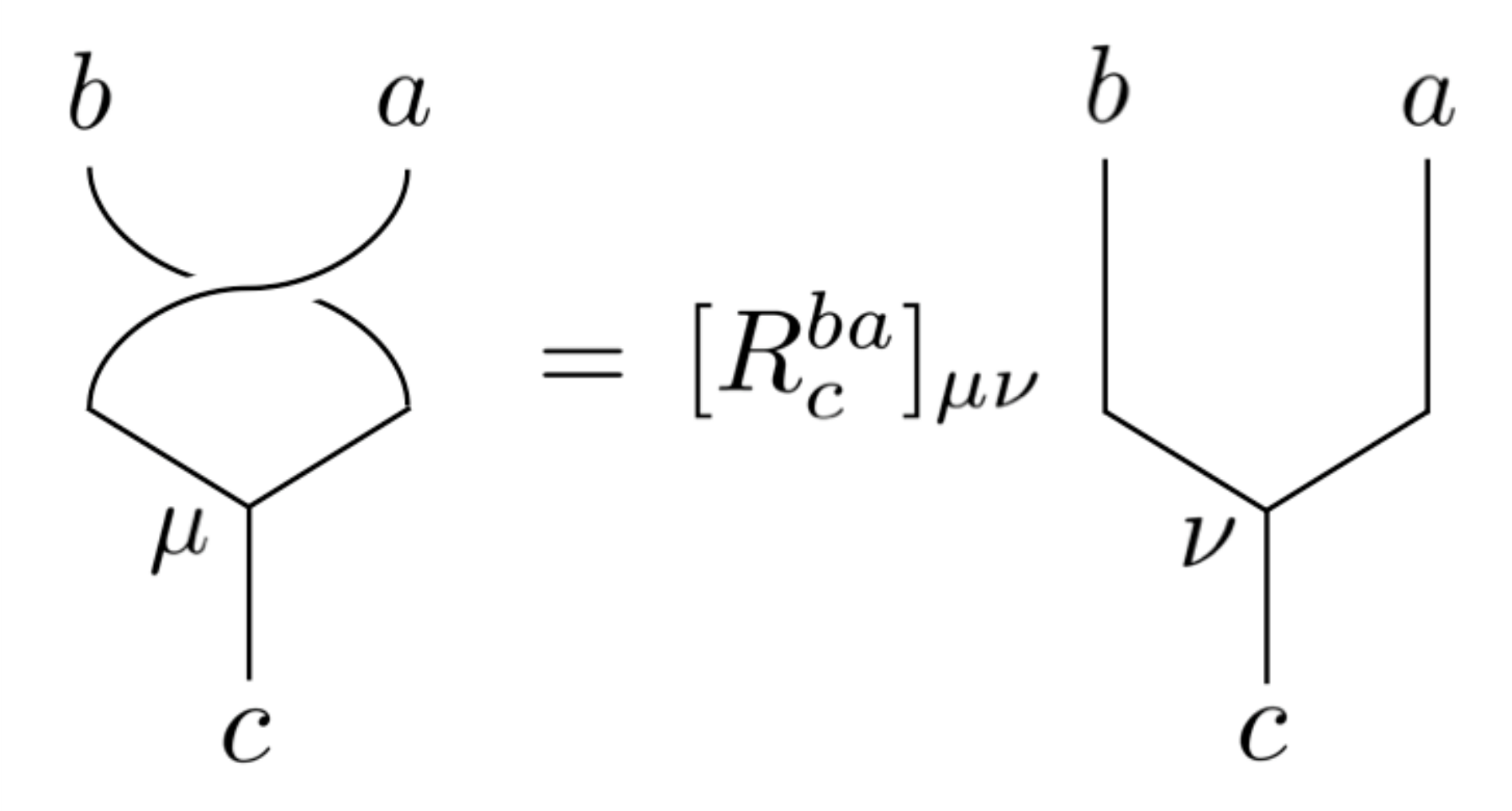}}}.
\end{equation}

\noindent
In the above equations, $\alpha,\beta,\mu,\nu$ denote the fusion multiplicities; when all fusion multiplicities are $\leq 1$, we will denote the $F$ symbols as $F^{abc}_{d;fe}$, and the $R$ symbols as $R^{ba}_c$.  The $F$ and $R$ symbols for a fusion category are highly non-unique because of the gauge choices during the solutions of pentagons and hexagons.

\end{appendix}

\end{document}